\documentclass[11pt]{article}

\usepackage[letterpaper,margin=1in]{geometry}
\usepackage[affil-it]{authblk}
\usepackage{cite}
\usepackage[unicode]{hyperref}
\usepackage[utf8]{inputenc}
\usepackage[english]{babel}
\usepackage{csquotes}
\usepackage{amsmath,amssymb,amsthm}
\usepackage{enumerate}
\usepackage{commath}
\usepackage{bm}
\usepackage{color}

\numberwithin{equation}{section}

\theoremstyle{definition}
\newtheorem{definition}{Definition}[section]
\newtheorem{example}{Example}[section]

\theoremstyle{plain}
\newtheorem{theorem}{Theorem}[section]
\newtheorem{proposition}{Proposition}[section]
\newtheorem{corollary}{Corollary}[section]

\newcommand{\Bool}{\{0,1\}}
\newcommand{\Words}{{\Bool^*}}
\newcommand{\WordsLen}[1]{{\Bool^{#1}}}

\DeclareMathOperator{\Sgn}{sgn}
\DeclareMathOperator{\Supp}{supp}

\DeclareMathOperator{\Img}{Im}

\DeclareMathOperator{\Prb}{Pr}
\DeclareMathOperator{\E}{E}
\DeclareMathOperator{\Var}{Var}

\DeclareMathOperator{\Ev}{ev}

\DeclareMathOperator{\T}{T}
\DeclareMathOperator{\R}{r}
\DeclareMathOperator{\A}{a}
\DeclareMathOperator{\M}{M}
\DeclareMathOperator{\UM}{UM}
\DeclareMathOperator{\Un}{U}
\DeclareMathOperator{\En}{c}

\newcommand{\Dtv}{\operatorname{d}_{\textnormal{tv}}}

\newcommand{\Argmin}[1]{\underset{#1}{\operatorname{arg\,min}}\,}
\newcommand{\Argmax}[1]{\underset{#1}{\operatorname{arg\,max}}\,}

\newcommand{\Nats}{\mathbb{N}}
\newcommand{\Ints}{\mathbb{Z}}
\newcommand{\Rats}{\mathbb{Q}}
\newcommand{\Reals}{\mathbb{R}}

\newcommand{\NatPoly}{\Nats[K_0, K_1 \ldots K_{n-1}]}
\newcommand{\NatPolyJ}{\Nats[J_0, J_1 \ldots J_{n-2}]}
\newcommand{\NatFun}{\Nats^n \rightarrow}

\newcommand{\Estr}{\bm{\lambda}}

\newcommand{\Lim}[1]{\lim_{#1 \rightarrow \infty}}

\newcommand{\Abs}[1]{\lvert #1 \rvert}
\newcommand{\Norm}[1]{\lVert #1 \rVert}
\newcommand{\Floor}[1]{\lfloor #1 \rfloor}
\newcommand{\Ceil}[1]{\lceil #1 \rceil}
\newcommand{\Chev}[1]{\langle #1 \rangle}

\newcommand{\Dist}{\mathcal{D}}
\newcommand{\GrowR}{\Gamma_{\mathfrak{R}}}
\newcommand{\GrowA}{\Gamma_{\mathfrak{A}}}
\newcommand{\Grow}{\Gamma:=(\GrowR,\GrowA)}
\newcommand{\MGrow}{\mathrm{M}\Gamma}
\newcommand{\Fall}{\mathcal{F}}
\newcommand{\EG}{\Fall(\Gamma)}
\newcommand{\ESG}{\Fall^\sharp(\Gamma)}
\newcommand{\EMG}{\Fall(\MGrow)}

\newcommand{\BoolR}[1]{\Bool^{\R_{#1}(K)}}

\newcommand{\GammaPoly}{\Gamma_{\textnormal{poly}}}
\newcommand{\GammaLog}{\Gamma_{\textnormal{log}}}
\newcommand{\FallU}{{\Fall_{\textnormal{uni}}^{(n)}}}
\newcommand{\FallUt}[1]{{\Fall_{\textnormal{uni}}^{(#1)}}}
\newcommand{\FallM}{\Fall_{\textnormal{mon}}^{(n)}}

\newcommand{\Alg}{\xrightarrow{\textnormal{alg}}}
\newcommand{\Markov}{\xrightarrow{\textnormal{mk}}}
\newcommand{\Scheme}{\xrightarrow{\Gamma}}
\newcommand{\MScheme}{\xrightarrow{\MGrow}}

\begin{document}

\title{Optimal Polynomial-Time Estimators: A Bayesian Notion of Approximation Algorithm}

\author{Vanessa Kosoy\thanks{\texttt{vanessa.kosoy@intelligence.org}}, Alexander Appel\thanks{\texttt{alexappel8@gmail.com}}}

\date{}

\maketitle

\begin{abstract}
We introduce a new concept of approximation applicable to decision problems and functions, inspired by Bayesian probability. From the perspective of a Bayesian reasoner with limited computational resources, the answer to a problem that cannot be solved exactly is uncertain and therefore should be described by a random variable. It thus should make sense to talk about the expected value of this random variable, an idea we formalize in the language of average-case complexity theory by introducing the concept of \enquote{optimal polynomial-time estimators.} We prove some existence theorems and completeness results, and show that optimal polynomial-time estimators exhibit many parallels with \enquote{classical} probability theory.%
\end{abstract}%

\pagebreak

\tableofcontents

\setcounter{section}{-1}

\section{Introduction}

\subsection{Motivation}
\label{subsec:mot}

Imagine you are strolling in the city with a friend when a car passes by with the license plate number \enquote{7614829}. Your friend proposes a wager, claiming that the number is composite and offering 10 : 1 odds in your favor. Knowing that your friend has no exceptional ability in mental arithmetic and that it's highly unlikely they saw this car before, you realize they are just guessing. Your mental arithmetic is also insufficient to test the number for primality, but is sufficient to check that ${7614829 \equiv 1 \pmod{3}}$ and $\frac{1}{\ln 7614829} \approx 0.06$. Arguing from the prime number theorem and observing that 7614829 is odd and is divisible neither by 3 nor by 5, you conclude that the probability 7614829 is prime is ${\frac{1}{\ln 7614829} \times 2 \times \frac{3}{2} \times \frac{5}{4} \approx 22\%}$. Convinced that the odds are in your favor, you accept the bet\footnote{Alas, $7614829 = 271 \times 28099$.}.

From the perspective of frequentist probability, the question \enquote{what is the probability 7614829 is prime?} seems meaningless. It is either prime or not, so there is no frequency to observe (unless the frequency is 0 or 1). From a Bayesian perspective, probability represents a degree of confidence; however, in classical Bayesian probability theory it is assumed that the only source of uncertainty is lack of information. The number 7614829 already contains all information needed to determine whether it is prime, so the probability again has to be 0 or 1. However, real life uncertainty is not only information-theoretic but also complexity-theoretic. Even when we have all of the information needed to obtain the answer, our computational resources are limited, and so we remain uncertain. The rigorous formalization of this idea is the main goal of the present work.

The idea of assigning probabilities to purely mathematical questions was studied by several authors\cite{Gaifman_2004,Hutter_2013,Demski_2012,Christiano_2014,Garrabrant_2015}, mainly in the setting of formal logic. That is, their approach was looking for functions from the set of sentences in some formal logical language to $[0,1]$. However, although there is a strong intuitive case for assigning probabilities to sentences like

\[\varphi_1:=\text{\enquote{7614829 is prime}}\]

it is much less clear there is a meaningful assignment of probabilities to sentences like 

\[\varphi_2 := \text{\enquote{there are no odd perfect numbers}}\] 

or (even worse) 

\[\varphi_3 := \text{``there is no cardinality } \kappa \text{ s.t. } \aleph_0 < \kappa < 2^{\aleph_0} \text{"}\]

A wager on $\varphi_1$ can be resolved in a predetermined finite amount of time (the amount of time it takes to test it directly). On the other hand, it is unknown how long the resolution of $\varphi_2$ will take. It is possible that there is an odd perfect number but finding it (or otherwise becoming certain of its existence) will take a very long time. It is also possible there is no odd perfect number, a fact that cannot be directly verified because of its infinite nature. It is possible that there is a proof of $\varphi_2$ within some formal theory, but accepting such a proof as resolution requires us to be completely certain of the consistency of the theory (whereas it is arguable that the consistency of formal mathematical theories, especially more abstract theories like ZFC, is itself only known empirically and in particular with less than absolute certainty). Moreover, there is no knowing a priori whether a proof exists or how long it will take to find it. For $\varphi_3$ there is no way to \enquote{directly} verify either the sentence or its negation, and it is actually known to be independent of ZFC.

In the present work we avoid choosing a specific category of mathematical questions\footnote{We do require that these questions can be represented as finite strings of bits.}. Instead, we consider the abstract setting of arbitrary distributional decision problems. This leads to the perspective that an assignment of probabilities is a form of \emph{approximate} solution to a problem. This is not the same sense of approximation as used in optimization problems, where the approximation error is the difference between the ideal solution and the actual solution. Instead, the approximation error is the prediction accuracy of our probability assignment. This is also different from average-case complexity theory, where the solution is required to be exact on most input instances. However, the language of average-case complexity theory (in particular, the concept of a distributional decision problem) turns out to be well-suited to our purpose.
The concept of \enquote{optimal polynomial-time estimator} that arises from the approach turns out to behave much like probabilities, or more generally expected values, in \enquote{classical} probability theory. They display an appropriate form of calibration. The \enquote{expected values} are linear in general and multiplicative for functions that are independent in an appropriate sense. There is a natural parallel of conditional probabilities. For simple examples constructed from one-way functions we get the probability values we expect. They are also well behaved in the complexity-theoretic sense that a natural class of reductions transforms optimal polynomial-time estimators into optimal polynomial-time estimators, and complete problems for these reductions exist for important complexity classes.

Optimal polynomial-time estimators turn out to be unique up to a certain equivalence relation. The existence of optimal polynomial-time estimators depends on the specific variety you consider. We show that in the non-uniform case (allowing advice) there is a variety of optimal polynomial-time estimators that exist for completely arbitrary problems. Uniform optimal polynomial-time estimators of this kind exist for a certain class of problems we call \enquote{samplable} which can be very roughly regarded as an average-case analogue of $\textsc{NP} \cap \textsc{coNP}$. More generally mapping the class of problems which admit optimal polynomial-time estimators allows for much further research.

\subsection{Overview}

Consider a language ${L \subseteq \Words}$ and a family ${\{\Dist^k\}}_{k \in \Nats}$ where each ${\Dist^k}$ is a probability distribution on ${\Words}$. We associate with $L$ its characteristic function $\chi_L: \Words \rightarrow \Bool$. A pair ${(\Dist,L)}$ is called a \emph{distributional decision problem}\cite{Bogdanov_2006}. Our goal is defining and studying the probabilities of \enquote{events} of the form ${x \in L}$\footnote{We will actually consider the more general case of a function ${f: \Words \rightarrow \Reals}$ and the \enquote{expected value} of ${f(x)}$, but for most purposes there is no difference of principle.}  associated with the uncertainty resulting from limited computational resources. (Specifically, we will consider the resources of time, randomness and advice.)

The distributional complexity class ${\textsc{Heur}_{\textnormal{neg}}\textsc{P}}$ is defined as the set of distributional decision problems which admit a polynomial-time heuristic algorithm with negligible error probability\cite{Bogdanov_2006}. That is, ${(\Dist, L) \in \textsc{Heur}_{\textnormal{neg}}\textsc{P}}$ iff there is ${A: \Nats \times \Words \Alg \Bool}$ (an algorithm which takes input in ${\Nats \times \Words}$ and produces output in ${\Bool}$) s.t. ${A(k,x)}$ runs in time polynomial in ${k}$ and ${\Prb_{x \sim \Dist^k}[A(k,x) \ne \chi_L(x)]}$ is a negligible function of ${k}$. We have the following equivalent condition. ${(\Dist, L) \in \textsc{Heur}_{\textnormal{neg}}\textsc{P}}$ iff there is ${P: \Nats \times \Words \Alg \Rats}$ s.t. ${P(k,x)}$ runs in time polynomial in ${k}$ and ${\E_{x \sim \Dist^k}[(P(k,x)-\chi_L(x))^2]}$ is a negligible function of ${k}$. In the language of the present work, such a ${P}$ is a called an \enquote{${\Fall_{\text{neg}}(\Gamma_0^1, \Gamma_0^1)}$-perfect polynomial-time estimator for ${(\Dist,\chi_L)}$} (see Definition~\ref{def:perfect}, Example~\ref{exm:fall_neg} and Example~\ref{exm:gamma_zero}).

Our main objects of study are algorithms satisfying a related but weaker condition. Namely, we consider ${P}$ s.t. its error w.r.t. ${\chi_L}$ is not negligible but is \emph{minimal up to a negligible function}. That is, we require that for any ${Q: \Nats \times \Words \Alg \Rats}$ s.t. ${Q(k,x)}$ also runs in time polynomial in ${k}$, there is a negligible function ${\varepsilon(k)}$ s.t. 

\[\E_{x \sim \Dist^k}[(P(k,x)-\chi_L(x))^2] \leq \E_{x \sim \Dist^k}[(Q(k,x)-\chi_L(x))^2] + \varepsilon(k)\]

Such a ${P}$ is called an \enquote{${\Fall_{neg}(\Gamma_0^1,\Gamma_0^1)}$-optimal polynomial-time estimator for ${(\Dist,\chi_L)}$.} More generally, we replace negligible functions by functions that lie in some space ${\Fall}$ which can represent different asymptotic conditions (see Definition~\ref{def:fall}), and we consider estimators that use certain asymptotic amounts of randomness and advice represented by a pair ${\Gamma}$ of function spaces (see Definition~\ref{def:grow}). This brings us to the concept of an \enquote{${\EG}$-optimal polynomial-time estimator} (see Definition~\ref{def:op}).

Denote ${\textsc{OP}[{\EG}]}$ the set of distributional decision problems that admit ${\EG}$-optimal polynomial-time estimators. Obviously ${\textsc{OP}[\Fall_{neg}(\Gamma_0^1,\Gamma_0^1)] \supseteq \textsc{Heur}_\text{neg}\textsc{P}}$. Moreover, if one-way functions exist the inclusion is proper since it is possible to use any function with a hard-core predicate to construct an example where the constant ${\frac{1}{2}}$ is an ${\Fall_{neg}(\Gamma_0^1,\Gamma_0^1)}$-optimal polynomial-time estimator (see Theorem~\ref{thm:hard_core}). Thus, it seems that we constructed novel natural distributional complexity classes.

The distributional complexity class ${\textsc{HeurP}}$ is defined as the set of distributional decision problems which admit a polynomial-time heuristic scheme\cite{Bogdanov_2006}. That is, ${(\Dist,L) \in \textsc{HeurP}}$ iff there is\\ ${S: \Nats^2 \times \Words \Alg \Bool}$ s.t. ${S(K_0,K_1,x)}$ runs in time polynomial in ${K_0, K_1}$ and\footnote{We slightly reformulated the definition given in \cite{Bogdanov_2006}: replaced the rational input parameter ${\delta}$ by the integer input parameter ${K_1}$. The equivalence of the two formulations may be observed via the substitution ${\delta=(K_1+1)^{-1}}$.}\\ ${\Prb_{x \sim \Dist^{K_0}}[S(K_0,K_1,x) \ne \chi_L(x)] \leq (K_1+1)^{-1}}$. Analogously to before, we have the following equivalent condition. ${(\Dist,L) \in \textsc{HeurP}}$ iff there is ${P: \Nats^2 \times \Words \Alg \Rats}$ s.t. ${P(K_0,K_1,x)}$ runs in time polynomial in ${K_0,K_1}$ and for some ${M > 0}$, ${\E_{x \sim \Dist^{K_0}}[(P(K_0,K_1,x)-\chi_L(x))^2] \leq M (K_1+1)^{-1}}$. In the language of the present work, such a ${P}$ is a called an \enquote{${\Fall_{(K_1+1)^{-1}}(\Gamma_0^2, \Gamma_0^2)}$-optimal polynomial-time estimator for ${(\Dist^\eta,\chi_L)}$} (see Example~\ref{exm:fall_zeta}), where ${\Dist^\eta}$ is a two-parameter (${K_0,K_1 \in \Nats}$) family of distributions which is constant along the parameter ${K_1}$.

Again we can consider the corresponding weaker condition, that for all ${Q: \Nats^2 \times \Words \Alg \Rats}$ s.t. ${Q(K_0,K_1,x)}$ runs in time polynomial in ${K_0,K_1}$

\[\E_{x \sim \Dist^{K_0}}[(P(K_0,K_1,x)-\chi_L(x))^2] \leq \E_{x \sim \Dist^{K_0}}[(Q(K_0,K_1,x)-\chi_L(x))^2] + M (K_1+1)^{-1}\]

Such a ${P}$ is called an \enquote{${\Fall_{(K_1+1)^{-1}}(\Gamma_0^2,\Gamma_0^2)}$-optimal polynomial-time estimator for ${(\Dist^\eta,\chi_L)}$.}

It is also useful to introduce the closely related concept of an \enquote{${\ESG}$-optimal polynomial-time estimator} (see Definition~\ref{def:obe_sharp}). For example, an ${\Fall_{(K_1+1)^{-1}}^\sharp(\Gamma_0^2,\Gamma_0^2)}$-optimal polynomial-time estimator ${P}$ has to satisfy that for each ${S: \Nats^2 \times \Words \Alg \Rats}$ that is also polynomial-time there is ${M > 0}$ s.t.

\[\Abs{\E_{x \sim \Dist^{K_0}}[(P(K_0,K_1,x)-\chi_L(x))S(K_0,K_1,x)]} \leq M (K_1+1)^{-1}\]

We show that e.g. every ${\Fall_{(K_1+1)^{-1}}^\sharp(\Gamma_0^2,\GammaLog^2)}$-optimal polynomial-time estimator is in particular an ${\Fall_{(K_1+1)^{-1}}(\Gamma_0^2,\GammaLog^2)}$-optimal polynomial-time estimator (see Theorem~\ref{thm:con_ort}), whereas every\\ ${\Fall_{(K_1+1)^{-1}}(\Gamma_0^2,\GammaLog^2)}$-optimal polynomial-time estimator is in particular an ${\Fall_{(K_1+1)^{-\frac{1}{2}}}^\sharp(\Gamma_0^2,\GammaLog^2)}$-optimal polynomial-time estimator (see Theorem~\ref{thm:ort}). Here, ${\GammaLog^2}$ indicates that we consider algorithms with advice of logarithmic length (see Example~\ref{exm:gamma_log}).

We claim that the concept of an optimal polynomial-time estimator is a formalisation of the intuition outlined in \ref{subsec:mot}. A priori, this is plausible because the mean squared error is a proper scoring rule (the Brier score). Moreover, it is the only scoring rule which is \enquote{proper} for arbitrary expected value assignment rather than only probability assignment. To support this claim, we prove a number of results that form a parallel between probability theory and the theory of optimal polynomial-time estimators:

\begin{itemize}
\item 
According to Borel's law of large numbers, every event of probability ${p}$ occurs with asymptotic frequency ${p}$. Therefore, if some algorithm ${P}$ represents a notion of probability for ${x \in L}$, we expect that given ${a,b \in \Rats}$ and considering ${x \sim \Dist^k}$ s.t. ${a \leq P(x) \leq b}$, the frequency with which ${x \in L}$ is asymptotically (in ${k}$) between ${a}$ and ${b}$. In Bayesian statistics, probability assignments satisfying such a property are said to be \enquote{well calibrated} (see e.g. \cite{Dawid_1982}). With some assumptions about allowed advice and the portion of the distribution falling in the ${[a,b]}$ interval, ${\EG}$-optimal polynomial-time estimators are well calibrated (see Corollary~\ref{crl:calib}). In particular, if the aforementioned portion is bounded from below, this frequency lies in ${[a,b]}$ up to a function of the form ${\sqrt{\varepsilon}}$ for ${\varepsilon \in \Fall}$.
\item
Given ${L_1, L_2 \subseteq \Words}$ s.t. ${L_1 \cap L_2 = \varnothing}$ we expect a reasonable notion of probability to satisfy ${\Prb[x \in L_1 \cup L_2] = \Prb[x \in L_1] + \Prb[x \in L_2]}$. To satisfy this expectation, we show that given ${\Dist}$ any family of distributions, ${P_1}$ an ${\ESG}$-optimal polynomial-time estimator for ${(\Dist,L_1)}$ and ${P_2}$ an ${\ESG}$-optimal polynomial-time estimator for ${(\Dist,L_2)}$, ${P_1 + P_2}$ is an ${\ESG}$-optimal polynomial-time estimator for ${(\Dist, L_1 \cup L_2)}$. This observation in itself is trivial (see Proposition~\ref{prp:linearity}) but applying it to examples may require passing from an ${\EG}$-optimal polynomial-time estimator to an ${\ESG}$-optimal polynomial-time estimator using the non-trivial Theorem~\ref{thm:ort}.
\item
Consider ${L,M \subseteq \Words}$ and suppose we are trying to formalize the conditional probability\\ ${\Prb[x \in L \mid x \in M]}$. There are two natural approaches. One is reducing it to unconditional probability using the identity 

\[\Prb[x \in L \mid x \in M]=\frac{\Prb[x \in L \cap M]}{\Prb[x \in M]}\]

We can then substitute optimal polynomial-time estimators for the numerator and denominator. The other is considering an optimal polynomial time-estimator for a family of conditional distributions. Luckily, these two approach yield the same result. That is, we show that given ${\Dist}$ a family of distributions, ${P_{LM}}$ an optimal polynomial time estimator for ${(\Dist,L \cap M)}$, ${P_M}$ an optimal polynomial-time estimator for ${(\Dist,M)}$ and assuming ${\Dist^K(M)}$ is not too small (e.g. bounded from below), ${P_M^{-1}P_{LM}}$ is an optimal polynomial-time estimator for ${(\Dist \mid M, L)}$ (see Theorem~\ref{thm:cond}). Conversely, given ${P_{L \mid M}}$ an optimal polynomial-time estimator for ${(\Dist \mid M, L)}$, ${P_M P_{L \mid M}}$ is an optimal polynomial-time estimator for ${(\Dist, L \cap M)}$ (see Theorem~\ref{thm:con_cond}).
\item
For some pairs ${L_1, L_2 \subseteq \Words}$, the \enquote{events} ${x \in L_1}$ and ${x \in L_2}$ can be intuitively regarded as independent since learning whether ${x \in L_2}$ doesn't provide any information about whether ${x \in L_1}$ that a polynomial-time algorithm can use. We formalize one situation when this happens and show that in this situation the product of an ${\ESG}$-optimal polynomial-time estimator (in certain form) for ${(\Dist,L_1)}$ by an ${\ESG}$-optimal polynomial-time estimator for ${(\Dist,L_2)}$ is an ${\ESG}$-optimal polynomial-time estimator for ${(\Dist, L_1 \cap L_2)}$ (see Theorem~\ref{thm:mult}). This is precisely analogous to the property of probabilities where the probability of the conjunction of independent events is the product of the separate probabilities. This is one of the central results of the present work.
\end{itemize}

Different complexity classes often have corresponding types of reductions that preserve them. In particular, reductions in average-case complexity theory have to satisfy an extra-condition that intuitively means that typical problem instances should not be mapped to rare problem instances. We define a class of reductions s.t. pull-backs of optimal polynomial-time estimators are optimal polynomial-time estimators. This requires stronger conditions than what is needed for preserving average-case complexity. Namely, a reduction ${\pi}$ of ${(\Dist,L)}$ to ${(\mathcal{E},M)}$ has to be \enquote{pseudo-invertible} i.e. there should be a way to sample ${\Dist \mid \pi^{-1}(y)}$ in polynomial time for ${y}$ sampled from ${\pi_* \Dist}$, up to an error which is asymptotically small on average. 

We give separate proofs for the invariance of ${\ESG}$-optimal polynomial-time estimators (see Corollary~\ref{crl:p_reduce_sharp}) and the invariance of ${\EG}$-optimal polynomial-time estimators (see Corollary~\ref{crl:p_reduce}) without relying on Theorem~\ref{thm:ort} and Theorem~\ref{thm:con_ort} in order to produce a slightly stronger bound. We also show that this reduction class is rich enough to support complete problems for many problem classes e.g. ${\textsc{SampNP}}$ (see Theorem~\ref{thm:complete}).

Explicit construction of optimal polynomial-time estimators is likely to often be difficult because it requires proving a hardness result (that no polynomial-time estimator can outperform the given polynomial-time estimator). However, for a specific choice of ${\Fall}$ which we denote ${\FallU}$ (see Example~\ref{exm:e_uni}), we prove two broad existence theorems.

The first (Theorem~\ref{thm:exists_all}) shows that for suitable ${\Gamma}$ (in particular it has to allow sufficiently long advice strings, e.g. logarithmic advice is sufficient), \emph{any} distributional decision problem ${(\Dist,L)}$ admits an ${\FallU(\Gamma)}$-optimal polynomial-time estimator for ${(\Dist^\eta,L)}$. The construction of this estimator is rather trivial: the advice string for ${(K_0,K_1)}$ is the optimal (i.e. least ${\E_{x \sim D^{K_0}}[(P(x)-f(x))^2]}$) program that runs in time ${K_1}$ and is of length at most ${l(K_0,K_1)}$ where ${l: \Nats^2 \rightarrow \Nats}$ is some function which determines the allowed asymptotic advice length (${\Gamma}$ depends on ${l}$ and an analogous function ${r: \Nats^2 \rightarrow \Nats}$ which determines the allowed asymptotic number of random bits used by the estimators). The non-trivial part here is the definition of ${\FallU}$ which is s.t. allowing any estimator an amount of resources greater by a polynomial always translates to a reduction in error which lies in ${\FallU}$.

The second (Theorem~\ref{thm:exists_smp}), which is another central result, shows that for suitable ${\Gamma}$ (logarithmic advice and enough random e.g. logarithmic amount of random bits is sufficient), any distributional decision\footnote{All of the theorems are described for decision problems in the overview for the sake of simplicity but we actually prove them for \enquote{estimation} problems i.e. ${f: \Words \rightarrow \Reals}$ instead of ${L \subseteq \Words}$. Here this generalisation is more important since any efficient algorithm producing ${(x,t)}$ pairs is the sampler of some distributional estimation problem.} problem ${(\Dist,L)}$ which is \emph{samplable} (i.e. it is possible to efficiently sample pairs ${(x,t)}$ where ${x \in \Words}$ is distributed approximately according to ${\Dist}$ and ${t \in \Rats}$ is an estimate of ${\chi_L(x)}$ which is approximately unbiased on average) admits an ${\FallU(\Gamma)}$-optimal polynomial-time estimator with the same advice strings as the sampler. In particular, if the sampler is uniform the estimator is also uniform.

The samplability property allows recasting the estimation problem as a learning problem. That is, we use the sampler to generate a number (we use ${\mathcal{O}\left(\left(\log{K_1}\right)^2\right)}$) of problem instances for which an unbiased estimate of the correct answer is known, and we should now generalize from these instances to an instance for which the correct answer is unknown. The optimal polynomial-time estimator we construct accomplishes this using the empirical risk minimization principle from statical learning theory, applied to a hypothesis space which consists of programs. Specifically, the estimator iterates over all programs of length ${O\left(\log{K_1}\right)}$, runs each of them on the samples ${\{(x_i,t_i)\}_{i \in [\mathcal{O}\left(\left(\log{K_1}\right)^2\right)]}}$ for time ${K_1}$ getting estimates ${\{p_i\}_{i \in [\mathcal{O}\left(\left(\log{K_1}\right)^2\right)]}}$ and computes the empirical risk ${\sum_{i \in [\mathcal{O}\left(\left(\log{K_1}\right)^2\right)]}(p_i-t_i)^2}$. It then selects the program with the minimal risk and runs it on the input for time ${K_1}$ to get the desired estimate. This is similar to Levin's universal search which dovetails all programs to get optimality. The optimality of this estimator is also closely related to the fundamental theorem of statistical learning theory for agnostic PAC learning\cite{Shalev-Shwartz_2014}: like in agnostic PAC learning we get an estimate which is not  but is optimal within the hypothesis space (which in our case is the space of efficient estimators).

On the other hand, we rule out the existence of optimal polynomial-time estimators in the uniform case for certain problems. These negative results rely on the simple observation that if the veracity of ${x \in L}$ for ${x \sim \Dist^k}$ depends only on ${k}$, then advice strings of size $O(1)$ enable storing the exact answer to all such questions. Additionally, it is easy to see that an optimal polynomial-time estimator in the uniform case is still optimal when we allow ${O(1)}$ advice. This means that any optimal polynomial-time estimator for such a problem has to be a  polynomial-time estimator. So, any problem of this form that doesn't have uniform  polynomial-time estimators also doesn't have uniform optimal polynomial-time estimators. Consequently, any problem that is reducible to the former sort of problem also doesn't have optimal polynomial-time estimators.

Finally, we examine the uniqueness of optimal polynomial-time estimators for a fixed problem. We prove that if such an estimator exists, it is unique up to a difference which is asymptotically small on average (see Theorem~\ref{thm:uniq}). For example, given ${(\Dist,L)}$ a distributional decision problem s.t. the length of any ${x \sim \Dist^k}$ is bounded by some polynomial in ${k}$ and ${P_1,P_2}$ two ${\Fall^\sharp(\Gamma_0^1,\Gamma_0^1)}$-optimal polynomial time estimators, ${\E_{x \sim \Dist^k}[(P_1(k,x)-P_2(k,x))^2]}$ is a function of ${k}$ that lies in ${\Fall}$.

We are able to prove a stronger uniqueness result for optimal polynomial-time estimators for problems of the form ${(\Dist \mid M, L)}$ (see Theorem~\ref{thm:uniq_cond}). Namely, if there is an optimal polynomial-time estimator ${P_M}$ for ${(\Dist,M)}$ which takes values with a sufficiently strong lower bound then any ${P_{L1},P_{L2}}$ optimal polynomial-time estimators for ${(\Dist \mid M, L)}$ have an asymptotically small difference on average with respect to ${\Dist}$ (rather than ${\Dist \mid M}$). Informally, this means that whenever determining that ${x \not\in M}$ is sufficiently hard, there are well-defined (up to an asymptotically small perturbation) probabilities for events of the form ${x \in L}$ conditioned by ${x \in M}$, even for instances which actually lie outside of ${M}$. That is, optimal polynomial-time estimators allow us asking counterfactual \enquote{what if} questions that are meaningless from a \enquote{classical} mathematical perspective due to the principle of explosion.

Many of our results make use of algorithms with advice strings, where the allowed asymptotic length of the advice strings is determined by the space of functions ${\GrowA}$. Such algorithms are not entirely realistic, but one way to interpret them is as real-time efficient (since we assume polynomial time) algorithms that require inefficient precomputation (at least this interpretation is valid when the advice strings are computable). The strength of the concept of an \enquote{${\EG}$-optimal polynomial-time estimator} depends ambiguously on the size of ${\GrowA}$, since on the one hand larger ${\GrowA}$ allows for a greater choice of candidate optimal polynomial-time estimators, on the other hand the estimator is required to be optimal in a larger class\footnote{The same observation is true about the space ${\GrowR}$ which controls the allowed quantity of random bits.}. Sometimes it is possible to get the best of both worlds by having an estimator which uses few or no advice but is optimal in a class of estimators which use much advice (see e.g. Theorem~\ref{thm:exists_smp}).

Note that most of the theorems we get about ${\EG}$-optimal polynomial-time estimators require a lower bound on ${\GrowA}$ through the assumption that ${\Fall}$ is ${\GrowA}$-ample (see Definition~\ref{def:ample}). Theorem~\ref{thm:ort} which shows when an ${\EG}$-optimal polynomial-time estimator is also an ${\Fall^{\frac{1}{2}\sharp}(\Gamma)}$-optimal polynomial-time estimator (see Definition~\ref{def:fall_space_power}) also assumes a lower bound on ${\GrowA}$, but a weaker one. On the other hand, the converse Theorem~\ref{thm:con_ort} makes no such assumption and so do all other theorems about ${\ESG}$-optimal polynomial-time estimators (except indirectly since Theorem~\ref{thm:ort} is often required to construct an ${\ESG}$-optimal polynomial-time estimator in the first place).

\subsection{Related Work}

Several authors starting from Gaifman studied the idea of assigning probabilities to sentences in formal logic\cite{Gaifman_2004,Hutter_2013,Demski_2012,Christiano_2014,Garrabrant_2015}. Systems of formal logic such as Peano Arithmetic are very expressive, so such an assignment would have much broader applicability than most of the examples we are concerned about in the present work. On the other hand, the constructions achieved by those authors are either much further from realistic algorithms (e.g. require halting oracles or at least very expensive computations\footnote{In fact, Theorem~\ref{thm:exists_all} shows optimal polynomial-time estimators exist for completely arbitrary distributional estimation problems, but the price is the need for advice strings which might be expensive or even uncomputable, depending on the problem. Nevertheless, these estimators are still \enquote{real-time efficient} which makes them semi-realistic in some sense.}) or have much weaker properties to attest to their interpretation as \enquote{probabilities}.

Lutz\cite{Lutz_1998} uses the theory of computable martingales to define when a set of sequences \enquote{appears for a polynomial-time observer} to have certain ${\nu}$-measure with respect to a fixed probability measure ${\nu}$ on the set of infinite strings ${\Bool^\omega}$. In particular, if a singleton ${\{x\}}$ has Lutz measure 1 (where ${x \in \Bool^\omega}$), this means that ${x}$ \enquote{looks like} a random sequence sampled from ${\nu}$, as far as a polynomial-time observer can tell. This seems closely related to our idea of assigning \enquote{subjective probabilities for polynomial-time observers} to events that are otherwise deterministic. Formally relating and comparing the two setups remains a task for future work.

The notion that computational hardness often behaves like information-theoretical uncertainty is well-known in complexity theory, although it hasn't been systematically formalized. For example see discussion of Theorem 7.5 in \cite{Goldreich_2008} or section 6.1 in \cite{Bogdanov_2006}. Results such as Yao's XOR lemma can be interpreted as the transformation of \enquote{computational probabilities} under certain operations, which is resonant with our results e.g. Theorem~\ref{thm:mult}. It seems likely that it is possible to fruitfully investigate these relations further.

Barak, Shaltiel and Wigderson \cite{Barak_2003} discuss notions of \enquote{entropy} for probability distributions that take computational hardness into account. Zheng \cite{Zheng_2014} (Chapter 7) considers prediction markets where traders perform transactions via Boolean circuits of polynomial size. This is similar to our optimal polynomial-time estimators, in the sense that a loss function which is a proper scoring rule is minimized under computational resource constraints. However, Zheng doesn't study this concept beyond deriving a relation to the \enquote{pseudoentropy} mentioned above.

Different brands of \enquote{optimal algorithms} were previously defined and investigated in various contexts. Levin's universal search is an algorithm that solves the candid search form of any problem in ${\textsc{NP}}$ in time which is minimal up to a polynomial (see Theorem 2.33 in \cite{Goldreich_2008}). Barak\cite{Barak_2002} uses instance checkers to construct algorithms optimal in this sense for decision problems (in particular for any problem that is ${\textsc{EXP}}$-complete). This concept also has a non-deterministic counterpart called \enquote{optimal proof system}: see survey by Hirsch\cite{Hirsch_2010}, which additionally discusses \enquote{optimal acceptors} (optimal algorithms that halt only on the \enquote{yes} instances of the problem). Notably, the latter survey also discusses the average-case rather than only the worst-case.

Khot's Unique Games Conjecture implies that many optimization problems have an algorithm which produces the best approximation factor possible in polynomial-time (see e.g. \cite{Khot_2010}). Barak and Steurer \cite{Barak_2014} speculate that even if the Unique Games Conjecture is false, the existence of an algorithm that is optimal in this sense for a large class of problems is plausible, and propose the Sum-of-Squares algorithm as a candidate.

Optimal polynomial-time estimators are optimal in a sense different from the examples above: they simultaneously run in polynomial-time, are applicable to decision problems and are of average-case nature. The metric they optimize is the average squared difference (Brier score) with the true function. Nevertheless, it might be interesting to explore connections and similarities with other types of optimal algorithms.\newline

The structure of the paper is as follows. Section~\ref{sec:notation} fixes notation. Section~\ref{sec:fundamentals} introduces the main definitions and gives a simple example using one-way functions. Section~\ref{sec:probability} shows the parallel between properties of optimal polynomial-time estimators and classical probability theory. Section~\ref{sec:reductions} discusses behavior of optimal polynomial-time estimators under reductions and shows certain natural classes have complete problems under reductions that are appropriate. Section~\ref{sec:e_and_u} discusses existence and uniqueness of optimal polynomial-time estimators. Section~\ref{sec:discussion} discusses possible avenues for further research. The Appendix briefly reviews relevant material about hard-core predicates and one-way functions.

\section{Notation}
\label{sec:notation}
\subsection{Sets, Numbers and Functions}

$\Nats$ is the set of natural numbers. We will use the convention in which natural numbers start from 0, so $\Nats = \{0, 1, 2 \ldots \}$. 

$\Ints$ is the ring of integers, $\Rats$ is the field of rational numbers, $\Reals$ is the field of real numbers.

For $F \in \{\Rats,\Reals\}$, $F^{>0} := \{x \in F \mid x > 0\}$, $F^{\geq 0} := \{x \in F \mid x \geq 0\}$.

Given ${n \in \Nats}$, ${\NatPoly}$ will stand for the set of polynomials with natural coefficients in the ${n}$ variables ${K_0, K_1 \ldots K_{n-1}}$.

For any $t \in \Reals$, $\Floor{t} := \max \{n \in \Ints \mid n \leq t\}$, $\Ceil{t} := \min \{n \in \Ints \mid n \geq t\}$.

$\log: \Reals^{\geq 0} \rightarrow \Reals \sqcup \{-\infty\}$ will denote the logarithm in base 2.

Given $n \in \Nats$, $[n]:=\{i \in \Nats \mid i < n\}$. Given sets $X_0, X_1 \ldots X_{n-1}$, ${x \in \prod_{i \in [n]} X_i}$ and $m \in [n]$, $x_m \in X_m$ is the $m$-th component of the $n$-tuple $x$ i.e. ${x=(x_0, x_1 \ldots x_{n-1})}$.

Given a set $X$ and $x,y \in X$, $\delta_{xy}$ (or $\delta_{x,y})$ will denote the the Kronecker delta

$$\delta_{xy} := \begin{cases}1 & \text{if } x=y \\ 0 & \text{if } x \ne y \end{cases}$$

Given a set $X$ and a subset $Y$, $\chi_Y: X \rightarrow \Bool$ will denote the indicator function of $Y$ (when $X$ is assumed to be known from the context)

$$\chi_Y(x):=\begin{cases}1 & \text{if } x \in Y \\ 0 & \text{if } x \not\in Y \end{cases}$$

$\theta: \Reals \rightarrow \Bool$ will denote the Heaviside step function $\theta:=\chi_{[0,\infty)}$. ${\Sgn: \Reals \rightarrow \{-1,+1\}}$ will denote the function $2 \theta - 1$.

\subsection{Probability Distributions}
\label{subsec:notation__prob}

For ${X}$ a set, ${\mathcal{P}(X)}$ will denote the set of probability distributions on ${X}$. A probability distribution on $X$ can be represented by a function $\Dist: X \rightarrow [0,1]$ s.t. $\sum_{x \in X} \Dist(x) = 1$. Abusing notation, we will use the same symbol to denote the function and the probability distribution. Given ${A}$ a subset of ${X}$, we will use the notation

\[\Dist(A):=\Prb_{x \sim \Dist}[x \in A] = \sum_{x \in A} \Dist(x)\]

For $X$ a set, $\Dist \in \mathcal{P}(X)$, $V$ a finite dimensional vector space over $\Reals$ and $f: X \rightarrow V$, $\E_{x \sim \Dist}[f(x)]$ will denote the expected value of $f$ with respect to $\Dist$, i.e. 

\[\E_{x \sim \Dist}[f(x)] := \sum_{x \in X} \Dist(x) f(x)\]

We will the abbreviated notations $\E_\Dist[f(x)]$, $\E[f(x)]$, $\E_\Dist[f]$, $\E[f]$ when no confusion is likely to occur.

Given a set $X$ and $\Dist \in \mathcal{P}(X)$, $\Supp \Dist$ will denote the support of $\Dist$ i.e.

\[\Supp \Dist = \{x \in X \mid \Dist(x) > 0\}\]

Given $X,Y$ sets, $\Dist \in \mathcal{P}(X)$ and $f: X \rightarrow Y$ a mapping, $f_*\Dist \in \mathcal{P}(Y)$ will denote the corresponding pushforward distribution i.e.

\[(f_*\Dist)(y):= \sum_{x \in f^{-1}(y)} \Dist(x)\]

Given $X,Y$ sets, the notation $f: X \xrightarrow{\textnormal{mk}} Y$ signifies $f$ is a Markov kernel with source $X$ and target $Y$. Given $x \in X$, $f_x$ is the corresponding probability distribution on $Y$ and $f(x)$ is a random variable sampled from $f_x$. Given $\Dist \in \mathcal{P}(X)$, $\Dist \ltimes f \in \mathcal{P}(X \times Y)$ (resp. $f \rtimes \Dist \in \mathcal{P}(Y \times X)$) is the semidirect product distribution. $f_*\Dist \in \mathcal{P}(Y)$ is the pushforward distribution, i.e. $f_*\Dist:=\pi_*(\Dist \ltimes f)$ where $\pi: X \times Y \rightarrow Y$ is the projection.

For $X$ a set, $\Dist \in \mathcal{P}(X)$ and $A$ a subset of $X$ s.t. ${\Dist(A) > 0}$, $\Dist \mid A$ will denote the corresponding conditional probability distribution, i.e. $(\Dist \mid A)(B):=\frac{\Dist(B \cap A)}{\Dist(A)}$. Given $Y$ another set, $f: X \Markov Y$ and $A$ a subset of $Y$ s.t. $(\Dist \ltimes f)(X \times A) > 0$, $\Dist \mid f^{-1}(A) \in \mathcal{P}(X)$ is defined by 

\[(\Dist \mid f^{-1}(A))(B):=(\Dist \ltimes f \mid X \times A)(B \times Y)\]

Note that when $f$ is deterministic (i.e. $f_x$ is a Dirac measure for every $x$), this corresponds to conditioning by the inverse image of $A$ with respect to $f$. When $A=\{a\}$ we will use the shorthand notation $\Dist \mid f^{-1}(a)$.

Given $X$ a set and $\Dist,\mathcal{E} \in \mathcal{P}(X)$, $\Dtv(\Dist,\mathcal{E})$ will denote the total variation distance between $\Dist$ and $\mathcal{E}$ i.e.

\[\Dtv(\Dist,\mathcal{E}):=\frac{1}{2}\sum_{x \in X} \Abs{\Dist(x) - \mathcal{E}(x)}\]

For $X$ a set and $x \in X$, $\delta_x$ will denote the Dirac measure associated with $x$, i.e. $\delta_x(y):=\delta_{xy}$.

\subsection{Algorithms}

$\Words$ is the set of all finite binary strings (words), i.e. $\Words:=\bigsqcup_{n \in \Nats} \Bool^n$. For any ${x \in \Words}$, $\Abs{x}$ is the length of $x$ i.e. $x \in \WordsLen{\Abs{x}}$. ${\Estr \in \Words}$ is the empty string. For any $n \in \Nats$

\begin{align*}
\Bool^{\leq n}&:=\{x \in \Words \mid \Abs{x} \leq n\} \\
\Bool^{>n}&:=\{x \in \Words \mid \Abs{x} > n\}
\end{align*}

For any $x \in \Words$ and $n \in \Nats$, $x_{< n}$ stands for the prefix of $x$ of length $n$ if $\Abs{x} \geq n$ and $x$ otherwise. Given $x,y \in \Words$, $xy$ stands for the concatenation of $x$ and $y$ (in particular $\Abs{xy}=\Abs{x}+\Abs{y}$). Given ${n \in \Nats}$ and ${x_0, x_1 \ldots x_{n-1} \in \Words}$, ${\prod_{i \in [n]} x_i}$ is also concatenation. Given $n \in \Nats$ and $x,y \in \WordsLen{n}$, $x \cdot y$ stands for $\bigoplus_{i \in [n]} x_i y_i$. For any $n \in \Nats$, $\Un^n \in \mathcal{P}(\WordsLen{n})$ is the uniform probability distribution.

Given $n \in \Nats$ and ${x_0, x_1 \ldots x_{n-1} \in \Words}$, $\Chev{x_0,x_1 \ldots x_{n-1}} \in \Words$ denotes the encoding of\\ $(x_0,x_1 \ldots x_{n-1})$ obtained by repeating each bit of $x_0, x_1 \ldots x_{n-1}$ twice and inserting the separators 01.
\begin{definition}

An \emph{encoded set} is a set $X$ together with an injection ${\En_X: X \rightarrow \Words}$ (the encoding) s.t. $\Img \En_X$ is decidable in polynomial time.

\end{definition}

There are standard encodings we implicitly use throughout. $\bm{1}$ denotes an encoded set with 1 element $\bullet$ whose encoding is the empty string. $\Words$ is an encoded set with the trivial encoding ${\En_\Words(x):=x}$. $\Nats$ is an encoded set where $\En_\Nats(n)$ is the binary representation of $n$. $\Rats$ is an encoded set where ${\En_\Rats(\frac{n}{m}):=\Chev{n,m}}$ for an irreducible fraction $\frac{n}{m}$. For any encoded set $X$ and $L \in \textsc{P}$, $\{x \in X \mid \En_X(x) \in L\}$ is an encoded set whose encoding is the restriction of $\En_X$. For $X_0,X_1 \ldots X_{n-1}$ encoded sets, $\prod_{i \in [n]} X_i$ is an encoded set with encoding 

\[\En_{\prod_{i \in [n]} X_i}(x_0,x_1 \ldots x_{n-1}):=\Chev{\En_{X_0}(x_0),\En_{X_1}(x_1) \ldots \En_{X_{n-1}}(x_{n-1})}\]

For any $n \in \Nats$ we use the shorthand notation $\En^n:=\En_{(\Words)^n}$.

Given $n \in \Nats$, encoded sets $X_0, X_1 \ldots X_{n-1}$ and encoded set $Y$ we use the notation\\ ${A: \prod_{i \in [n]} X_i \Alg Y}$ to mean a Turing machine with $n$ input tapes that halts on every input for which the $i$-th tape is initialized to a value in $\Img \En_X$ and produces an output in $\Img \En_Y$. Given $\{x_i \in X_i\}_{i \in [n]}$ the notation $A(x_0, x_1 \ldots x_{n-1})$ stands for the unique $y \in Y$ s.t. applying $A$ to the input composed of $\En_{X_i}(x_i)$ results in output $\En_Y(y)$. We use different input tapes for different components of the input instead of encoding the $n$-tuple as a single word in order to allow $A$ to process some components of the input in time smaller than the length of other components. This involves abuse of notation since a Cartesian product of encoded sets is naturally an encoded set, but hopefully this won't cause much confusion.

Given $A: X \Alg Y$ and $x \in X$, $\T_A(x)$ stands for the number of time steps in the computation of $A(x)$.

For any $n \in \Nats$, we fix $\mathcal{U}_n$, a prefix free universal Turing machine with $n+1$ input tapes: 1 program tape and $n$ tapes that serve as input to the program. Given ${n,k \in \Nats}$, ${a \in \Words}$ and ${\{x_i \in \Words\}_{i \in [n]}}$, ${\Ev^k(a;x_0,x_1 \ldots x_{n-1})}$ stands for the output of ${\mathcal{U}_n}$ when executed for ${k}$ time steps on program ${a}$ (continued by an infinite sequence of 0s) and inputs ${\{x_i \in \Words\}_{i \in [n]}}$.

\section{Fundamentals}
\label{sec:fundamentals}

\subsection{Basic Concepts}

\subsubsection{Distributional Estimation Problems}

We start with a simple model to help build intuition and motivate the following definitions.

Consider finite sets $X$ and $Y$, $\Dist \in \mathcal{P}(X)$, a mapping $m: X \rightarrow Y$ and a function $f: X \rightarrow \Reals$. Suppose $x$ was sampled from $\Dist$ and we were told $y := m(x)$ (but not told $x$ itself). Our expected value of $f(x)$ in these conditions is ${\E_{x \sim \Dist}[f(x) \mid m(x) = y]}$.

Let $P: X \rightarrow \Reals$ be the function $P(x) := \E_{x' \sim \Dist}[f(x') \mid m(x') = m(x)]$. How can we characterize $P$ without referring to the concept of a conditional expected value? For any $Q: X \rightarrow \Reals$ we can consider the \enquote{error} $\E_\Dist[(Q - f)^2]$. $Q$ is called \enquote{efficient} when it factors as $Q = q \circ m$ for some $q: Y \rightarrow \Reals$. It is easy to see that $P$ has the least error among all efficient functions.

Note that the characterization of $P$ depends not only on $f$ but also on $\Dist$. That is, the accuracy of an estimator depends on the prior probabilities to encounter different questions. In general, we assume that the possible questions are represented by elements of $\Words$. Thus we need to consider a probability distribution on $\Words$. However, in the spirit of average-case complexity theory we will only require our estimators to be \emph{asymptotically} optimal. Therefore instead of considering a single probability distribution we consider a family of probability distribution indexed by integer parameters\footnote{It is convenient to allow more than 1 parameter for reasons that will become clear in section~\ref{sec:e_and_u}. Roughly, some parameters represent the complexity of the input whereas other parameters represent the amount of computing resources available for probability estimation.}, where the role of the parameters is defining the relevant limit. We thereby arrive at the following:

\begin{definition}

Fix ${n \in \Nats}$. A \emph{word ensemble of rank ${n}$} is a family ${\{\Dist^{K} \in \mathcal{P}(\Words)\}_{K \in \Nats^n}}$.

We will use the notation $\Supp \Dist := \bigcup_{K \in \Nats^n} \Supp \Dist^K$.

\end{definition}

We now introduce our abstraction for a \enquote{class of mathematical questions} (with quantitative real-valued answers). This abstraction is a trivial generalization of the concept of a distributional decision problem from average-case complexity theory (see e.g. \cite{Bogdanov_2006}).

\begin{definition}

Fix ${n \in \Nats}$. A \emph{distributional estimation problem of rank ${n}$} is a pair $(\Dist,f)$ where $\Dist$ is a word ensemble of rank ${n}$ and $f: \Supp \Dist \rightarrow \Reals$ is bounded.

\end{definition}

\subsubsection{Growth Spaces and Polynomial-Time \texorpdfstring{$\Gamma$}{Γ}-Schemes}

In the motivational model, the estimator was restricted to lie in a class of functions that factor through a fixed mapping. Of course we are interested in more realistic notions of efficiency. In the present work we consider restrictions on time complexity, access to random bits and size of advice strings. Spatial complexity is also of interest but treating it is out of our current scope. It is possible to consider weaker or stronger restrictions which we represent using the following abstraction which is closely tied to big-$\mathcal{O}$ notation:

\begin{samepage}
\begin{definition}
\label{def:grow}
Fix $n$. A \emph{growth space} $\Gamma$ of rank $n$ is a set of functions ${\gamma: \NatFun \Nats}$ s.t.

\begin{enumerate}[(i)]

\item\label{con:def__grow__zero} $0 \in \Gamma$

\item\label{con:def__grow__add} If $\gamma_1, \gamma_2 \in \Gamma$ then $\gamma_1 + \gamma_2 \in \Gamma$.

\item\label{con:def__grow__ineq} If $\gamma_1 \in \Gamma$, $\gamma_2: \NatFun \Nats$ and $\forall K \in \Nats^n: \gamma_2(K) \leq \gamma_1(K)$ then $\gamma_2 \in \Gamma$.

\item\label{con:def__grow__poly} For any $\gamma \in \Gamma$ there is a $p \in \NatPoly$ s.t. $\gamma \leq p$.

\end{enumerate}

\end{definition}
\end{samepage}

\begin{example}
\label{exm:gamma_zero}

For any $n \in \Nats$, we define $\Gamma_0^n$, a growth space of rank $n$. $\gamma \in \Gamma_0^n$ iff $\gamma \equiv 0$.

\end{example}

\begin{samepage}
\begin{example}

For any $n \in \Nats$, we define $\Gamma_1^n$, a growth space of rank $n$. $\gamma \in \Gamma_1^n$ iff there is ${c \in \Nats}$ s.t. ${\gamma \leq c}$.

\end{example}
\end{samepage}

\begin{example}

For any $n \in \Nats$, we define $\GammaPoly^n$, a growth space of rank $n$. 

\[\GammaPoly^n := \{\gamma: \NatFun \Nats \mid \exists p \in \NatPoly: \gamma \leq p\}\]

\end{example}

\begin{example}
\label{exm:gamma_log}

For any $n \in \Nats$, we define $\GammaLog^n$, a growth space of rank $n$. $\gamma \in \GammaLog^n$ iff there is $c \in \Nats$ s.t. $\gamma(K_0, K_1 \ldots K_{n-1}) \leq c \sum_{i \in [n]} \log(K_i+1)$.

\end{example}

\begin{samepage}
\begin{definition}
\label{def:sgrow}
Fix ${n \in \Nats^{>0}}$. ${\gamma: \NatFun \Nats}$ is said to be \emph{steadily growing} when
\begin{enumerate}[(i)]
\item ${\gamma \in \GammaPoly^n}$
\item ${\forall J \in \Nats^{n-1}, k,l \in \Nats: k < l \implies \gamma(J,k) \leq \gamma(J,l)}$
\item\label{con:def__sgrow__poly} There is ${s \in \NatPoly}$ s.t. ${\forall J \in \Nats^{n-1}, k \in \Nats: \gamma(J,k) \leq \frac{1}{2}\gamma(J,s(J,k))}$.
\end{enumerate}
\end{definition}
This could be thought of as a polynomial that is monotonically increasing in the last argument quickly enough that a polynomial increase in the last argument can double the available resources.
\end{samepage}

\begin{samepage}
\begin{example}

For any ${n \in \Nats^{>0}}$ and ${\gamma^*}$ steadily growing, we define ${\Gamma_{\gamma^*}}$, a growth space of rank ${n}$. ${\gamma \in \Gamma_{\gamma^*}}$ iff there is ${p \in \NatPoly}$ s.t. ${\gamma(J,k) \leq \gamma^*(J,p(J,k))}$. This is the space of functions that are bounded above by the reference function ${\gamma^*}$ with the last argument growing at a polynomial rate.

To verify condition~\ref{con:def__grow__add}, consider ${\gamma_1}$, ${\gamma_2}$ s.t. ${\gamma(J,k) \leq \gamma^*(J,p_1(J,k))}$ and ${\gamma_2(J,k) \leq \gamma^*(J,p_2(J,k))}$. Choose ${p,s \in \NatPoly}$ s.t. ${p \geq \max(p_1,p_2)}$ and ${s}$ is as in condition~\ref{con:def__sgrow__poly} of Definition~\ref{def:sgrow}.

\[\gamma_1(J,k)+\gamma_2(J,k) \leq \gamma^*(J,p_1(J,k))+\gamma^*(J,p_2(J,k))\]

\[\gamma_1(J,k)+\gamma_2(J,k) \leq 2\gamma^*(J,p(J,k))\]

\[\gamma_1(J,k)+\gamma_2(J,k) \leq \gamma^*(J,s(J,p(J,k)))\]

In particular taking ${\gamma^*_{\text{poly}}(J,k):=k}$ and ${\gamma^*_{\text{log}}(J,k):=\Floor{\log (k+1)}}$ we have ${\Gamma_{\text{poly}}^n=\Gamma_{\gamma^*_{\text{poly}}}}$,\\ $\GammaLog^n=\Gamma_{\gamma^*_{\text{log}}}$.

\end{example}
\end{samepage}

We now introduce our notion of an \enquote{efficient} algorithm.

\begin{samepage}
\begin{definition}

Fix $n \in \Nats$ and $\Gamma=(\GrowR$, $\GrowA)$ a pair of growth spaces of rank $n$ that correspond to the length of the random and advice strings. Given encoded sets $X$ and $Y$, a \emph{polynomial-time $\Gamma$-scheme of signature $X \rightarrow Y$} is a triple $(S,\R_S,\A_S)$ where\\ ${S: \Nats^n \times X \times \Words \times \Words \Alg Y}$, $\R_S: \Nats^n \times \Words \Alg \Nats$ and $\A_S: \NatFun \Words$ are s.t.

\begin{enumerate}[(i)]

\item $\max_{x \in X} \max_{y,z \in \Words} \T_S(K,x,y,z) \in \GammaPoly^n$

\item $\max_{z \in \Words} \T_{\R_S}(K,z) \in \GammaPoly^n$. Note that $\R_S$, the polynomial-time function that outputs the number of random bits to read, takes the advice string z as input.

\item The function $r: \NatFun \Nats$ defined by $r(K):=\R_S(K,\A_S(K))$ lies in $\GrowR$.

\item $\Abs{\A_S} \in \GrowA$

\end{enumerate}
Abusing notation, we denote the polynomial-time $\Gamma$-scheme $(S,\R_S,\A_S)$ by $S$. $S^K(x,y,z)$ will denote $S(K,x,y,z)$, $S^K(x,y)$ will denote $S(K,x,y,\A_S(K))$ and $S^K(x)$ will denote the $Y$-valued random variable which equals $S(K,x,y,a(K))$ for $y$ sampled from $\Un^{\R_S(K)}$. $\Un_S^K$ will denote $\Un^{\R_S(K)}$. We think of $S$ as a randomized algorithm with advice where $y$ are the internal coin tosses and $\A_S$ is the advice\footnote{Note that the number of random bits $\R_S(K)$ has to be efficiently computable modulo the advice $\A_S(K)$ rather than being an arbitrary function. This requirement is needed to prevent using the function $\R_S$ as advice in itself. In particular, when $\GrowA=\Gamma_0^2$, $S$ represents a uniform randomized algorithm.}. Similarly, $\R_S(K)$ will denote $\R_S(K,\A_S(K))$.

We will use the notation $S: X \Scheme Y$ to signify $S$ is a polynomial-time $\Gamma$-scheme of signature $X \rightarrow Y$.

\end{definition}
\end{samepage}

There is a natural notion of composition for polynomial-time $\Gamma$-schemes.

\begin{samepage}
\begin{definition}

Fix $n \in \Nats$ and $\Gamma=(\GrowR$, $\GrowA)$ a pair of growth spaces of rank $n$. Consider encoded sets $X$, $Y$, $Z$ and $S: X \Scheme Y$, $T: Y \Scheme Z$. Choose $p \in \NatPoly$ s.t. $\Abs{\A_S(K)} \leq p(K)$ and $\Abs{\A_T(K)} \leq p(K)$. We can then construct $U: X \Scheme Z$ s.t. for any $K \in \Nats^n$, $a,b \in \Bool^{\leq p(K)}$, ${v \in \Bool^{\R_T(K,a)}}$, ${w \in \Bool^{\R_S(K,b)}}$ and $x \in X$

\begin{align}
\A_U(K) &= \Chev{\A_T(K),\A_S(K)} \\
\R_U(K, \Chev{a,b}) &= \R_T(K,a)+\R_S(K,b) \\
U^K(x,vw,\Chev{a,b}) &= T^K(S^K(x,w,b),v,a)
\end{align}

Such a $U$ is called the \emph{composition} of $T$ and $S$ and denoted $U = T \circ S$. There is a slight abuse of notation due to the freedoms in the construction of $U$ but these freedoms have no real significance since all versions of $T \circ S$ induce the same Markov kernel from $X$ to $Z$.

\end{definition}
\end{samepage}

It will also be useful to consider families of polynomial-time $\Gamma$-schemes satisfying uniform resource bounds.

\begin{definition}
\label{def:family}

Fix $n \in \Nats$, $\Gamma=(\GrowR$, $\GrowA)$ a pair of growth spaces of rank $n$ and encoded sets $X$, $Y$. A set $F$ of polynomial-time $\Gamma$-schemes of signature $X \rightarrow Y$ is called a \emph{uniform family} when

\begin{enumerate}[(i)]

\item\label{con:def__family__time} $\max_{S \in F} \max_{x \in X} \max_{y,z \in \Words} \T_S(K,x,y,z) \in \GammaPoly^n$

\item\label{con:def__family__rtime} $\max_{S \in F} \max_{z \in \Words} \T_{\R_S}(K,z) \in \GammaPoly^n$

\item\label{con:def__family__rand} $\max_{S \in F} \R_S \in \GrowR$

\item\label{con:def__family__adv} $\max_{S \in F} \Abs{\A_S(K)} \in \GrowA$

\item There are only finitely many different machines ${S}$ and ${\R_S}$ for ${S \in F}$.

\end{enumerate}

\end{definition}

The details of this definition are motivated by the following proposition.

\begin{proposition}
\label{prp:fam_diag}

Fix $n \in \Nats$ and $\Gamma=(\GrowR$, $\GrowA)$ a pair of growth spaces of rank $n$ s.t. $1 \in \GrowA$. Consider $X$, $Y$ encoded sets, $F$ a uniform family of polynomial-time $\Gamma$-schemes of signature $X \rightarrow Y$ and a collection ${\{\mathcal{S}_K \in F\}_{K \in \Nats^n}}$. Then, there is $\Delta_\mathcal{S}: X \xrightarrow{
\Gamma} Y$ s.t. for any $K \in \Nats^n$, $x \in X$ and $y \in Y$, ${\Pr[\Delta_\mathcal{S}^K(x)=y] = \Pr[\mathcal{S}_K^K(x)=y]}$.

\end{proposition}

\begin{proof}

Choose ${p,q \in \NatPoly}$ (time bounds to emulate an arbitrary randomness function and algorithm from the uniform family) and ${\{a_K, b_K \in \Words\}}_{K \in \Nats^n}$ 
(encodings of the randomness functions and algorithms, which by the definition of a uniform family, can be a finite set) s.t. there is only a finite number of different words $a_K$ and $b_K$, and for any ${K,L \in \Nats^n}$, ${x \in X}$ and ${y,z \in \Words}$

\begin{align*}
\Ev^{q(L)}(b_K;\En_{\Nats^n}(L),z) &= \R_{S_K}(L,z) \\
\Ev^{p(L)}(a_K;\En_{\Nats^n}(L),x,y,z) &= S_K^L(x,y,z) 
\end{align*}

Now, use the nonzero advice string to encode which algorithm is to be used on which input. Construct ${\Delta_{\mathcal{S}}}$ s.t. for any ${K \in \Nats^n}$, ${x \in X}$, ${y, w \in \Words}$, ${u \in \Bool^{\leq \max_{K \in \Nats^n} \Abs{a_K}}}$ and\\ ${v \in \Bool^{\leq \max_{K \in \Nats^n} \Abs{b_K}}}$

\begin{align*}
\A_{\Delta_{\mathcal{S}}}(K)&=\Chev{a_K,b_K,\A_{\mathcal{S}_K}(K)} \\
\R_{\Delta_{\mathcal{S}}}(K,\Chev{u,v,w})&=\Ev^{q(K)}(v;\En_{\Nats^n}(K),w) \\
\Delta_{\mathcal{S}}^K(x,y,\Chev{u,v,w})&=\Ev^{p(K)}(u;\En_{\Nats^n}(K),x,y,w)
\end{align*}
\end{proof}

\subsubsection{Fall Spaces}

Fix $n \in \Nats$ and $\Gamma$ a pair of growth spaces of rank $n$. Given a distributional estimation problem $(\Dist,f)$ and $Q: \Words \Scheme \Rats$, we can consider the estimation error $\E_{(x,y) \sim \Dist^{K} \times \Un_Q^K}[(Q^K(x,y) - f(x))^2]$. It makes little sense to require this error to be minimal for every $K \in \Nats^n$, since we can always hard-code a finite number of answers into $Q$ without violating the resource restrictions. Instead we require minimization up to an asymptotically small error. Since it makes sense to consider different kind of asymptotic requirements, we introduce an abstraction that corresponds to this choice.

\begin{definition}
\label{def:fall}

Given $n \in \Nats$, a \emph{fall space of rank $n$} is a set $\Fall$ of bounded functions $\varepsilon: \NatFun \Reals^{\geq 0}$ s.t.

\begin{enumerate}[(i)]

\item\label{con:def__fall__add} If $\varepsilon_1, \varepsilon_2 \in \Fall$ then $\varepsilon_1 + \varepsilon_2 \in \Fall$.

\item\label{con:def__fall__ineq} If $\varepsilon_1 \in \Fall$, $\varepsilon_2: \NatFun \Reals^{\geq 0}$ and $\forall K \in \Nats^n: \varepsilon_2(K) \leq \varepsilon_1(K)$ then $\varepsilon_2 \in \Fall$.

\item\label{con:def__fall__pol} There is $h \in \NatPoly$ s.t. $2^{-h} \in \Fall$.

\end{enumerate}
\end{definition}

\begin{example}
\label{exm:fall_neg}

We define $\Fall_{\text{neg}}$, a fall space of rank $1$. For any $\varepsilon: \Nats \rightarrow \Reals^{\geq 0}$ bounded, $\varepsilon \in \Fall_{\text{neg}}$ iff for any $d \in \Nats$, $\Lim{k} k^d \varepsilon(k) = 0$.

\end{example}

\begin{samepage}
\begin{example}
\label{exm:fall_zeta}

For any ${n \in \Nats}$ and ${\zeta: \Nats^n \rightarrow \Reals^{\geq 0}}$, we define ${\Fall_{\zeta}}$ to be the set of ${\varepsilon: \NatFun \Reals^{\geq 0}}$ bounded s.t. there is ${M \in \Reals}$ for which ${\varepsilon \leq M \zeta}$. If there is ${h \in \NatPoly}$ s.t. ${\zeta \geq 2^{-h}}$ then ${\Fall_\zeta}$ is a fall space of rank ${n}$.

\end{example}
\end{samepage}

\begin{samepage}
\begin{example}
\label{exm:e_uni}
For any ${n \in \Nats^{>0}}$ and ${\varphi: \Nats^{n-1} \rightarrow \Nats \sqcup \{ \infty \}}$, we define ${\FallUt{\varphi}}$, a fall space of rank ${n}$. For any ${\varepsilon: \NatFun \Reals^{\geq 0}}$ bounded, $\varepsilon \in \FallUt{\varphi}$ iff there are ${M \in \Reals^{>0}}$ and ${p \in \NatPolyJ}$ s.t.

\begin{equation}
\forall J \in \Nats^{n-1}: \sum_{k=2}^{\varphi(J)-1}\frac{\varepsilon(J,k)}{k \log k} \leq M \log \log p(J)
\end{equation}

To verify condition \ref{con:def__fall__pol} note that ${2^{-K_{n-1}} \in \FallUt{t}}$.

For ${\varphi \equiv \infty}$ we use the notation ${\FallU:=\FallUt{\varphi}}$.

For example, if $\varepsilon_{1}(J,k):=\frac{\log\log p(J)}{\log (k+2)}$ and $\varepsilon_{2}(J,k):=\frac{\log\log p(J)}{\log\log (k+2)}$, then $\varepsilon_{1}(J,k)\in\FallU$, but $\varepsilon_{2}(J,k)\not\in\FallU$ because it falls too slowly to force the sum to converge.

\end{example}
\end{samepage}

\begin{samepage}
\begin{example}
\label{exm:e_mon}

For any ${n \in \Nats^{>0}}$, we define ${\FallM}$, a fall space of rank ${n}$. For any ${\varepsilon: \NatFun \Reals^{\geq 0}}$ bounded, $\varepsilon \in \FallM$ iff the function ${\bar{\varepsilon}: \NatFun \Reals^{\geq 0}}$ defined by ${\bar{\varepsilon}(J,k):=\sup_{l \geq k} \varepsilon(J,l)}$ satisfies ${\bar{\varepsilon} \in \FallU}$.

\end{example}
\end{samepage}

The main motivation for examples \ref{exm:e_uni} and \ref{exm:e_mon} are the existence theorems proven in Section \ref{sec:e_and_u}.

We note a few simple properties of fall spaces which will be useful in the following.

\begin{proposition}
\label{prp:err_spc_zero}

For any fall space $\Fall$, $0 \in \Fall$.

\end{proposition}

\begin{proof}

Follows from conditions \ref{con:def__fall__ineq} and \ref{con:def__fall__pol}, since $0 \leq 2^{-h}$.
\end{proof}

\begin{proposition}

For any fall space $\Fall$, $\varepsilon \in \Fall$ and $c \in \Reals^{\geq 0}$, $c \varepsilon \in \Fall$.

\end{proposition}

\begin{proof}

By induction, condition~\ref{con:def__fall__add} implies that for any $m \in \Nats$, $m\varepsilon \in \Fall$. It follows that $c\varepsilon \in \Fall$ since $c\varepsilon \leq \Ceil{c}\varepsilon$.
\end{proof}

\begin{proposition}

For any fall space $\Fall$ and $\varepsilon_1, \varepsilon_2 \in \Fall$, $\max(\varepsilon_1,\varepsilon_2) \in \Fall$

\end{proposition}

\proof{$$\max(\varepsilon_1,\varepsilon_2) \leq \varepsilon_1+\varepsilon_2$$}

\begin{proposition}
\label{prp:fall_space_closed_wrt_power}

For any fall space $\Fall$, $\varepsilon \in \Fall$ and $\alpha \in \Reals$, if $\alpha \geq 1$ then $\varepsilon^\alpha \in \Fall$.

\end{proposition}

\begin{proof}

$$\varepsilon^\alpha = (\sup \varepsilon)^\alpha \left(\frac{\varepsilon}{\sup \varepsilon}\right)^\alpha \leq  (\sup \varepsilon)^\alpha \frac{\varepsilon}{\sup \varepsilon} \in \Fall$$
\end{proof}

\begin{samepage}
\begin{definition}
\label{def:fall_space_power}

For any fall space $\Fall$ and $\alpha \in \Reals^{>0}$, we define ${\Fall^\alpha := \{\varepsilon^\alpha \mid \varepsilon \in \Fall\}}$.

\end{definition}
\end{samepage}

\begin{proposition}

Consider $\Fall$ a fall space and $\alpha \in \Reals^{>0}$. Then, $\Fall^\alpha$ is a fall space.

\end{proposition}

\begin{proof}

To check condition~\ref{con:def__fall__add}, consider $\varepsilon_1, \varepsilon_2 \in \Fall$. 

If $\alpha > 1$, $(\varepsilon_1^\alpha + \varepsilon_2^\alpha)^\frac{1}{\alpha} \leq \varepsilon_1 + \varepsilon_2 \in \Fall$ hence $(\varepsilon_1^\alpha + \varepsilon_2^\alpha)^\frac{1}{\alpha} \in \Fall$ and $\varepsilon_1^\alpha + \varepsilon_2^\alpha \in \Fall^\alpha$.

If $\alpha \leq 1$, $(\varepsilon_1^\alpha + \varepsilon_2^\alpha)^\frac{1}{\alpha} = 2^\frac{1}{\alpha}(\frac{\varepsilon_1^\alpha + \varepsilon_2^\alpha}{2})^\frac{1}{\alpha} \leq 2^\frac{1}{\alpha} \frac{\varepsilon_1+\varepsilon_2}{2} \in \Fall$ hence $(\varepsilon_1^\alpha + \varepsilon_2^\alpha)^\frac{1}{\alpha} \in \Fall$ and $\varepsilon_1^\alpha + \varepsilon_2^\alpha \in \Fall^\alpha$.

Conditions \ref{con:def__fall__ineq} and \ref{con:def__fall__pol} are obvious.
\end{proof}

\begin{proposition}

Consider $\Fall$ a fall space and $\alpha_1,\alpha_2 \in \Reals^{>0}$ with $\alpha_1 \leq \alpha_2$. Then, ${\Fall^{\alpha_2} \subseteq \Fall^{\alpha_1}}$.

\end{proposition}

\begin{proof}

Follows from Proposition~\ref{prp:fall_space_closed_wrt_power}.
\end{proof}

\begin{samepage}
\begin{definition}

For any $n \in \Nats$, fall space $\Fall$ of rank $n$ and $\gamma: \NatFun \Reals$ s.t. $\inf \gamma > 0$, we define $\gamma \Fall := \{\gamma \varepsilon \text{ bounded} \mid \varepsilon \in \Fall\}$.

\end{definition}
\end{samepage}

\begin{samepage}
\begin{proposition}
\label{prp:tbd}

For any $n \in \Nats$, fall space $\Fall$ of rank $n$ and $\gamma: \NatFun \Reals$ s.t. $\inf \gamma > 0$, $\gamma \Fall$ is a fall space.

\end{proposition}
\end{samepage}

\begin{proof}

Conditions \ref{con:def__fall__add} and \ref{con:def__fall__ineq} are obvious. To verify condition~\ref{con:def__fall__pol} note that for any $\varepsilon \in \Fall$ we have ${\frac{\varepsilon}{\gamma} \leq \frac{\varepsilon}{\inf \gamma} \in \Fall}$ and therefore $\varepsilon = \gamma \frac{\varepsilon}{\gamma} \in \gamma \Fall$. In particular if $h \in \NatPoly$ is s.t. $2^{-h} \in \Fall$ then $2^{-h} \in \gamma \Fall$.
\end{proof}

We will use several shorthand notations for relations between functions that hold \enquote{up to a function in ${\Fall}$.} Given $f,g: \NatFun \Reals$, the notation ${f(K) \leq g(K) \pmod \Fall}$ means 

\[\exists \varepsilon \in \Fall \forall K \in \Nats^n: f(K) \leq g(K) + \varepsilon(K)\]

Similarly, $f(K) \geq g(K) \pmod \Fall$ means 

\[\exists \varepsilon \in \Fall \forall K \in \Nats^n: f(K) \geq g(K) - \varepsilon(K)\]

$f(K) \equiv g(K) \pmod \Fall$ means $\Abs{f-g} \in \Fall$.

For families ${\{f_\alpha,g_\alpha: \NatFun \Reals\}_{\alpha \in I}}$ (where ${I}$ is some set), ${f_\alpha(K) \overset{\alpha}{\leq} g_\alpha(K) \pmod \Fall}$ means that 

\[\exists \varepsilon \in \Fall \forall \alpha \in I, K \in \Nats^n: f_\alpha(K) \leq g_\alpha(K) + \varepsilon(K)\]

${f_\alpha(K) \overset{\alpha}{\geq} g_\alpha(K) \pmod \Fall}$ and ${f_\alpha(K) \overset{\alpha}{\equiv} g_\alpha(K) \pmod \Fall}$ are defined analogously.

\subsubsection{Optimal Polynomial-Time Estimators}

We are now ready to give our central definition, which corresponds to a notion of \enquote{expected value} for distributional estimation problems.

\begin{definition}
\label{def:op}

Fix $n \in \Nats$, $\Gamma$ a pair of growth spaces of rank $n$ and $\Fall$ a fall space of rank $n$. Consider $(\Dist,f)$ a distributional estimation problem and $P: \Words \Scheme \Rats$ with bounded range. $P$ is called an \emph{$\EG$-optimal polynomial-time estimator for $(\Dist,f)$} when for any $Q: \Words \Scheme \Rats$

\begin{equation}
\label{eqn:op}
\E_{\Dist^{K} \times \Un_P^K}[(P^K - f)^2] \leq \E_{\Dist^{K} \times \Un_Q^K}[(Q^K - f)^2] \pmod \Fall
\end{equation}

For the sake of brevity, we will say \enquote{${\EG}$-optimal estimator} rather than \enquote{${\EG}$-optimal polynomial-time estimator.}

\end{definition}

Distributional \emph{decision} problems are the special case when the range of $f$ is $\Bool$. In this special case, the outputs of an optimal polynomial-time estimator can be thought of as probabilities\footnote{With some caveats. First, $P$ can take values outside $[0,1]$ but it's easy to see that clipping all values to $[0,1]$ preserves optimality. Second, $P^{K}(x,y)=1$ doesn't imply $f(x) = 1$ and $P^{K}(x,y)=0$ doesn't imply $f(x)=0$. We can try to fix this using a logarithmic error function instead of the squared norm, however this creates other difficulties and is outside the scope of the present work.}.

\subsection{Basic Properties}

From now on we fix $n \in \Nats^{>0}$, $\Grow$ a pair of growth spaces of rank $n$ and $\Fall$ a fall space of rank $n$. All word ensembles and distributional estimation problems will be of rank ${n}$ unless specified otherwise.

In this subsection we discuss some basic properties of optimal polynomial-time estimators which will be used in the following.

\subsubsection{Optimality Relative to Uniform Families}

Note that $\varepsilon$ in \ref{eqn:op} depends on $Q$. However in some sense the optimality condition is automatically uniform w.r.t. the resources required by $Q$. The following Proposition \ref{prp:unif} can be used to reduce domination of a uniform family to domination of a single polynomial-time $\Gamma$-scheme constructed via Proposition \ref{prp:fam_diag}.

\begin{proposition}
\label{prp:unif}

Consider $(\Dist,f)$ a distributional estimation problem, $P$ an $\EG$-optimal estimator for $(\Dist,f)$ and $F$ a uniform family of polynomial-time $\Gamma$-schemes of signature $\Words \rightarrow \Rats$. Then there is $\varepsilon \in \Fall$ s.t. for any $Q \in F$

\begin{equation}
\E_{\Dist^{K} \times \Un_P^{K}}[(P^{K} - f)^2] \leq \E_{\Dist^{K} \times \Un_Q^{K}}[(Q^{K} - f)^2] + \varepsilon(K)
\end{equation}

\end{proposition}

\begin{proof}

For any $K \in \Nats^n$, $\{\E_{\Dist^{K} \times \Un_Q^{K}}[(Q^{K} - f)^2] \mid Q \in F\}$ is a finite set because $F$ is a uniform family so the runtime of $Q^{K}$ is bounded by a polynomial in $K$ that doesn't depend on $Q$. Therefore we can choose 

\[Q_{K} \in \Argmin{Q \in F} \E_{\Dist^{K} \times \Un_Q^{K}}[(Q^{K} - f)^2]\]

By Proposition~\ref{prp:fam_diag}, there is $\bar{Q}: \Words \Scheme \Rats$ s.t. $\bar{Q}^{K}(x)$ is distributed the same as $Q_{K}^{K}(x)$.

Since $P$ is an $\EG$-optimal estimator, there is $\varepsilon \in \Fall$ s.t.

\begin{equation}
\label{eqn:prp__unif__prf1}
\E_{\Dist^{K} \times \Un_P^{K}}[(P^{K} - f)^2] \leq \E_{\Dist^{K} \times \Un_{\bar{Q}}^{K}}[(\bar{Q}^{K} - f)^2] + \varepsilon(K)
\end{equation}

For any $Q \in F$, we have 

$$\E_{\Dist^{K} \times \Un_{\bar{Q}}^{K}}[(\bar{Q}^{K} - f)^2]=\E_{\Dist^{K} \times \Un_{Q_{K}}^{K}}[(Q_{K}^{K} - f)^2]$$

\begin{equation}
\label{eqn:prp__unif__prf2}
\E_{\Dist^{K} \times \Un_{\bar{Q}}^{K}}[(\bar{Q}^{K} - f)^2] \leq \E_{\Dist^{K} \times \Un_Q^{K}}[(Q^{K} - f)^2]
\end{equation}

Combining \ref{eqn:prp__unif__prf1} and \ref{eqn:prp__unif__prf2} we get the desired result.
\end{proof}

\subsubsection{Random versus Advice}

As usual, random is no more powerful than advice (see e.g. Theorem 6.3 in \cite{Goldreich_2008}). This is demonstrated by the following two propositions.

\begin{proposition}

Observe that $\bar{\Gamma}_{\mathfrak{R}}:=\GrowR+\GrowA$ is a growth space and denote $\bar{\Gamma}:=(\bar{\Gamma}_{\mathfrak{R}},\GrowA)$. Consider $(\Dist,f)$ a distributional estimation problem and $P$ an $\EG$-optimal estimator for $(\Dist,f)$. Then, $P$ is also an $\Fall(\bar{\Gamma})$-optimal estimator for $(\Dist,f)$.

\end{proposition}

\begin{proof}

The proof will proceed by taking a $Q$ with access to extra randomness, and then considering another algorithm $\overline{Q}$ with access to the old amount of randomness, which uses the advice to encode an optimal prefix to the random string. Then we just need to show that $\overline{Q}$ dominates $Q$ and is dominated by $P$. This proof strategy also applies to the next proposition.

Consider any $Q: \Words \xrightarrow{\bar{\Gamma}} \Rats$. Suppose $\R_Q=r_{\mathfrak{R}}+r_{\mathfrak{A}}$ where $r_{\mathfrak{R}} \in \GrowR$ and $r_{\mathfrak{A}} \in \GrowA$. For any $K \in \Nats^n$, choose 
\[\bar{\A}_Q(K) \in \Argmin{y \in \WordsLen{r_{\mathfrak{A}}(K)}} \E_{(x,z) \sim \Dist^{K} \times U^{r_{\mathfrak{R}}(K)}}[(Q^{K}(x,yz) - f(x))^2]\]

As is easy to see, there is $\bar{Q}: \Words \Scheme \Rats$ s.t. for all $K \in \Nats^n$, $x \in \Supp \Dist^{K}$ and $z \in \WordsLen{r_{\mathfrak{R}}(K)}$

\begin{align*}
\A_{\bar{Q}}(K)&=\Chev{\A_Q(K),\bar{\A}_Q(K)} \\
\R_{\bar{Q}}(K) &= r_{\mathfrak{R}}(K) \\
\bar{Q}^{K}(x,z)&=Q^{K}(x,\bar{\A}_Q(K)z)
\end{align*}

It follows that there is $\varepsilon \in \Fall$ s.t.

$$\E_{\Dist^{K} \times \Un_P^{K}}[(P^{K} - f)^2] \leq \E_{\Dist^{K} \times U^{r_{\mathfrak{R}}(K)}}[(\bar{Q}^{K} - f)^2] + \varepsilon(K)$$

Obviously $\E_{\Dist^{K} \times U^{r_{\mathfrak{R}}(K)}}[(\bar{Q}^{K} - f)^2] \leq \E_{\Dist^{K} \times \Un_Q^{K}}[(Q^{K} - f)^2]$ therefore

$$\E_{\Dist^{K} \times \Un_P^{K}}[(P^{K} - f)^2] \leq \E_{\Dist^{K} \times \Un_Q^{K}}[(Q^{K} - f)^2] + \varepsilon(K)$$
\end{proof}

\begin{proposition}

Denote $\bar{\Gamma}_{\mathfrak{R}}:=\GrowR+\GrowA$ and $\bar{\Gamma}:=(\bar{\Gamma}_{\mathfrak{R}},\GrowA)$. Consider $(\Dist,f)$ a distributional estimation problem and $\bar{P}$ an $\Fall(\bar{\Gamma})$-optimal estimator for $(\Dist,f)$. Then, there exists an $\EG$-optimal estimator for $(\Dist,f)$.

\end{proposition}

\begin{proof}
Suppose $\R_{\bar{P}}=r_{\mathfrak{R}}+r_{\mathfrak{A}}$ where $r_{\mathfrak{R}} \in \GrowR$ and $r_{\mathfrak{A}} \in \GrowA$. For any ${K \in \Nats^n}$, choose 

\[\bar{\A}_P(K) \in \Argmin{y \in \WordsLen{r_{\mathfrak{A}}(K)}} \E_{(x,z) \sim \Dist^{K} \times \Un^{r_{\mathfrak{R}}(K)}}[(\bar{P}^{K}(x,yz) - f(x))^2]\]

We can construct $P: \Words \Scheme \Rats$ so that for all $K \in \Nats^n$, $x \in \Supp \Dist^{K}$ and ${z \in \WordsLen{r_{\mathfrak{R}}(K)}}$

\begin{align*}
\A_P(K) &:=\Chev{\A_{\bar{P}}(K),\bar{\A}_P(K)} \\
\R_P(K) &= r_{\mathfrak{R}}(K) \\
P^{K}(x,z) &=\bar{P}^{K}(x,\bar{\A}_P(K)z)
\end{align*}

Clearly ${\E_{\Dist^{K} \times \Un^{r_{\mathfrak{R}}(K)}}[(P^{K} - f)^2] \leq \E_{\Dist^{K} \times \Un_{\bar{P}}^{K}}[(\bar{P}^{K} - f)^2]}$ and therefore $P$ is an $\EG$-optimal estimator for $(\Dist,f)$.
\end{proof}

\subsubsection{Optimality of Weighted Error}

Although the word ensemble plays a central role in the definition of an optimal polynomial-time estimator, the dependence on the word ensemble is lax in some sense. To see this, consider the following proposition.

\begin{definition}
\label{def:ample}

Given a growth space $\Gamma_*$ of rank $n$, $\Fall$ is called \emph{$\Gamma_*$-ample} when there is\\ $\zeta: \NatFun (0,\frac{1}{2}]$ s.t.  $\zeta \in \Fall$ and $\Floor{\log \frac{1}{\zeta}} \in \Gamma_*$.

\end{definition}
The intuitive interpretation of this is that, when $\Gamma_*$ represents the amount of advice, the advice bits are sufficient to write down an approximation to some parameter with error at most $\zeta$.
\begin{samepage}
\begin{example}

Any fall space of rank ${n}$ is ${\GammaPoly^n}$-ample, due to condition~\ref{con:def__fall__pol} of Definition~\ref{def:fall}.

\end{example}
\end{samepage}

\begin{samepage}
\begin{example}

${\FallU}$ is ${\GammaLog^n}$-ample since we can take ${\zeta(K):=(K_{n-1}+2)^{-1}}$.

\end{example}
\end{samepage}

\begin{proposition}
\label{prp:weight}

Assume $\Fall$ is $\GrowA$-ample. Consider $(\Dist,f)$ a distributional estimation problem, $P$ an $\EG$-optimal estimator for $(\Dist,f)$, $Q: \Words \Scheme \Rats$ and ${W: \Words \Scheme \Rats^{\geq 0}}$ bounded s.t. ${\R_W \geq \max(\R_P, \R_Q)}$. Denote ${\Dist_W^K:=\Dist^K \times \Un_W^K}$. Then

\begin{equation}
\E_{\Dist_W^{K}}[W^{K}(x,y)(P^{K}(x,y_{<\R_P(K)}) - f(x))^2] \leq \\ \E_{\Dist_W^{K}}[W^{K}(x,y)(Q^{K}(x,y_{<\R_Q(K)}) - f(x))^2] \pmod \Fall
\end{equation}

\end{proposition}

This essentially says that if there is enough advice available, then an optimal estimator continues to be optimal when a poly-time adversary assigns weights to how important the various problem instances are. The proof will come after the following corollary.

The relationship to the role of the word ensemble is as follows.

\begin{samepage}
\begin{corollary}
\label{crl:weight}

Assume $\Fall$ is $\GrowA$-ample. Consider $(\Dist,f)$ a distributional estimation problem and $P$ an $\EG$-optimal estimator for $(\Dist,f)$. Consider ${W: \Words \Scheme \Rats^{\geq 0}}$ bounded s.t. for any $K \in \Nats^n$ there is $x \in \Supp \Dist^K$ and $y \in \WordsLen{\R_W(K)}$ s.t. $W^K(x,y) > 0$. Define ${\gamma: \NatFun \Reals}$ by ${\gamma(K):=\E_{\Dist^K \times \Un_W^K}[W^K]^{-1}}$ and denote ${\Fall_W:= \gamma\Fall}$. Define the word ensemble $\mathcal{E}$ by 

\[\mathcal{E}^K(x):=\frac{\E_{y \sim \Un_W^K}[W^K(x,y)] \Dist^K(x)}{\E_{(x',y) \sim \Dist^K \times \Un_W^K}[W^K(x',y)]}\]

Then, $P$ is an $\Fall_W(\Gamma)$-optimal estimator for $(\mathcal{E},f)$.
That is, if the distribution on problem instances is reweighted by a poly-time adversary, an optimal estimator will continue being optimal, although with an increased error if the expected weight assigned by the adversary keeps falling as K grows. Therefore, the property of being an optimal estimator is robust against distributional shift when enough advice is available.
\end{corollary}
\end{samepage}

\begin{proof}

Consider any $Q: \Words \Scheme \Rats$. Proposition~\ref{prp:weight} implies there is $\varepsilon \in \Fall$ s.t.

$$\E_{\Dist^{K} \times \Un_P^{K} \times \Un_W^{K}}[W^{K}(P^{K} - f)^2] \leq \\ \E_{\Dist^{K} \times \Un_Q^{K} \times \Un_W^{K}}[W^{K}(Q^{K} - f)^2] + \varepsilon(K)$$

$$\E_{\Dist^{K} \times \Un_P^{K}}[\E_{\Un_W^{K}}[W^{K}](P^{K} - f)^2] \leq \\ \E_{\Dist^{K} \times \Un_Q^{K}}[\E_{\Un_W^{K}}[W^{K}](Q^{K} - f)^2] + \varepsilon(K)$$

Dividing both sides of the inequality by $\E_{\Dist^{K} \times \Un_W^{K}}[W^{K}(x)]$ we get

$$\E_{\mathcal{E}^{K} \times \Un_P^{K}}[(P^{K} - f)^2] \leq \\ \E_{\mathcal{E}^{K} \times \Un_Q^{K}}[(Q^{K} - f)^2] + \frac{\varepsilon(K)}{\E_{\Dist^{K} \times \Un_W^{K}}[W^{K}(x)]}$$

Let $M$ be the supremum of the left hand side.

$$\E_{\mathcal{E}^{K} \times \Un_P^{K}}[(P^{K} - f)^2] \leq \\ \E_{\mathcal{E}^{K} \times \Un_Q^{K}}[(Q^{K} - f)^2] + \min\left(\frac{\varepsilon(K)}{\E_{\Dist^{K} \times \Un_W^{K}}[W^{K}(x)]},M\right)$$

The second term on the right hand side is clearly in $\Fall_W$.
\end{proof}

We now give the proof of Proposition~\ref{prp:weight}.

\begin{proof}[Proof of Proposition \ref{prp:weight}]

The proof will construct a uniform family of algorithms that use their advice to encode an approximation of some number $t$, and use $P$ if $W$ assigns a weight less than the approximation, and $Q$ otherwise. $P$ dominates all the algorithms in this family, and after some reshuffling, and integrating over $t$, $W$ can be recovered, and this leads to the desired result.

Consider $\zeta: \NatFun (0,\frac{1}{2}]$ s.t.  $\zeta \in \Fall$ and $\Floor{\log \frac{1}{\zeta}} \in \GrowA$. For any $K \in \Nats^n$ and $t \in \Reals$, let $\rho_\zeta^{K}(t) \in \Argmin{s \in \Rats \cap [t-\zeta(K),t+\zeta(K)]} \Abs{\En_\Rats(s)}$. Denote $M:= \sup W$. It is easy to see that there is $\gamma \in \GrowA$ s.t. for any $t \in [0, M]$, ${\Abs{\En_\Rats(\rho_\zeta^{K}(t))} \leq \gamma(K)}$.

For any $t \in \Reals$ there is $Q_t: \Words \Scheme \Rats$ s.t. $\R_Q=\R_W$ and for any ${x \in \Supp \Dist^{K}}$ and ${y \in \WordsLen{\R_W(K)}}$

$$Q_t^{K}(x,y)=\begin{cases}Q^{K}(x,y_{< \R_Q(K)}) \text{ if } W^{K}(x,y) \geq \rho^{K}_\zeta(t) \\ P^{K}(x,y_{< \R_P(K)}) \text{ if } W^{K}(x,y) < \rho^{K}_\zeta(t)\end{cases}$$

Moreover we can construct the $Q_t$ for all $t \in [0, M]$ s.t. they form a uniform family. By Proposition~\ref{prp:unif} there is $\varepsilon \in \Fall$ s.t. for all $t \in [0, M]$

$$\E_{\Dist^{K} \times \Un_P^{K}}[(P^{K}-f)^2] \leq \E_{\Dist^{K} \times \Un_W^{K}}[(Q_t^{K}-f)^2] + \varepsilon(K)$$

$$\E_{(x,y) \sim \Dist^{K} \times \Un_W^{K}}[(P^{K}(x,y_{< \R_P(K)})-f(x))^2-(Q_t^{K}(x,y)-f(x))^2] \leq \varepsilon(K)$$

The expression inside the expected values vanishes when $W^{K}(x,y) < \rho^{K}_\zeta(t)$. In other cases, 
\[Q_t^{K}(x,y) = Q^{K}(x,y_{< \R_Q(K)})\]

We get

$$\E_{(x,y) \sim \Dist^{K} \times \Un_W^{K}}[\theta(W^{K}(x,y)-\rho_\zeta^{K}(t)) \cdot ((P^{K}(x,y_{< \R_P(K)})-f(x))^2-(Q^{K}(x,y_{< \R_Q(K)})-f(x))^2)] \leq \varepsilon(K)$$

We integrate both sides of the inequality over $t$ from 0 to $M$.

\begin{equation}
\label{eqn:prp__weight__prf1}
\E\left[\int_0^M\theta(W^{K}-\rho_\zeta^{K}(t)) \dif t \cdot ((P^{K}-f)^2-(Q^{K}-f)^2)\right] \leq M \varepsilon(K)
\end{equation}

For any $s \in \Reals$

$$\int_0^M \theta(s-\rho_\zeta^{K}(t)) \dif t = \int_0^{s-\zeta(K)} \theta(s-\rho_\zeta^{K}(t)) \dif t + \int_{s-\zeta(K)}^{s+\zeta(K)} \theta(s-\rho_\zeta^{K}(t)) \dif t + \int_{s+\zeta(K)}^M \theta(s-\rho_\zeta^{K}(t)) \dif t$$

$\Abs{\rho_\zeta^{K}(t)-t} \leq \zeta(K)$ therefore the integrand in the first term is 1 and in the last term 0:

$$\int_0^M \theta(s-\rho_\zeta^{K}(t)) \dif t = \int_0^{s-\zeta(K)} \dif t + \int_{s-\zeta(K)}^{s+\zeta(K)} \theta(s-\rho_\zeta^{K}(t)) \dif t$$

$$\int_0^M \theta(s-\rho_\zeta^{K}(t)) \dif t = s-\zeta(K) + \int_{s-\zeta(K)}^{s+\zeta(K)} \theta(s-\rho_\zeta^{K}(t)) \dif t$$

$$\int_0^M \theta(s-\rho_\zeta^{K}(t)) \dif t - s = -\zeta(K) + \int_{s-\zeta(K)}^{s+\zeta(K)} \theta(s-\rho_\zeta^{K}(t)) \dif t$$

\begin{equation}
\label{eqn:prp__weight__prf2}
\int_0^M \theta(s-\rho_\zeta^{K}(t)) \dif t - s \in [-\zeta(K),\zeta(K)]
\end{equation}

Combining \ref{eqn:prp__weight__prf1} and \ref{eqn:prp__weight__prf2} we conclude that for some $M' \in \Reals$

$$\E[W^{K} \cdot ((P^{K}-f)^2-(Q^{K}-f)^2)] \leq M \varepsilon(K) + M'\zeta(K)$$
\end{proof}

\subsubsection{Amplification from Zero to \texorpdfstring{$O(1)$}{O(1)} Advice}

The following will be handy to prove negative existence results (see section~\ref{sec:e_and_u}).

\begin{samepage}
\begin{proposition}
\label{prp:adv_amp}

Assume ${\GrowA=\Gamma_0^n}$. Consider ${(\Dist,f)}$ a distributional estimation problem and ${P}$ an ${\EG}$-optimal estimator for ${(\Dist,f)}$. Denote ${\Gamma_1:=(\GrowR,\Gamma_1^n)}$. Then, ${P}$ is also an ${\Fall(\Gamma_1)}$-optimal estimator for ${(\Dist,f)}$.

\end{proposition}
\end{samepage}

\begin{proof}
This proof proceeds by using the standard "domination of a uniform family" result to dominate all the algorithms with a bounded-size advice string that never changes. An algorithm with constant advice can be interpreted as switching around within this family, and thus is dominated.
Consider any ${Q: \Words \xrightarrow{\Gamma_1} \Rats}$. Choose ${l \in \Nats}$ s.t. ${\forall K \in \Nats^n: \Abs{\A_Q(K)} \leq l}$. For each\\ $a \in \Bool^{\leq l}$, construct ${Q_a: \Words \Scheme \Rats}$ s.t. for any ${K \in \Nats^n}$, ${x,y \in \Words}$

\begin{align*}
\R_{Q_a}(K) &= \R_Q(K,a) \\
Q_a^K(x,y) &= Q^K(x,y,a) 
\end{align*}

For some ${\varepsilon_a \in \Fall}$ we have

\[\E_{\Dist^{K} \times \Un_P^{K}}[(P^{K} - f)^2] \leq \E_{\Dist^{K} \times \Un_{Q_a}^{K}}[(Q_a^{K} - f)^2] + \varepsilon_a(K)\]

Since the above holds for every ${a \in \Bool^{\leq l}}$, we get 

\[\E_{\Dist^{K} \times \Un_P^{K}}[(P^{K} - f)^2] \leq \E_{\Dist^{K} \times \Un_{Q}^{K}}[(Q^{K} - f)^2] + \varepsilon_{\A_Q(K)}(K)\]

\[\E_{\Dist^{K} \times \Un_P^{K}}[(P^{K} - f)^2] \leq \E_{\Dist^{K} \times \Un_{Q}^{K}}[(Q^{K} - f)^2] + \sum_{a \in \Bool^{\leq l}} \varepsilon_a(K)\]

\end{proof}

\subsection{Orthogonality Theorems}

There is a variant of Definition~\ref{def:op} which is nearly equivalent in many cases and often useful.

We can think of functions $f: \Supp \Dist \rightarrow \Reals$ as vectors in a real inner product space with inner product $\Chev{f,g}:=\E_\Dist[fg]$. Informally, we can think of polynomial-time $\Gamma$-schemes as a subspace (although a polynomial-time $\Gamma$-scheme is not even a function) and an $\EG$-optimal estimator for $(\Dist,f)$ as the nearest point to $f$ in this subspace. Now, given an inner product space $V$, a vector $f \in V$, an actual subspace $W \subseteq V$ and $p = \Argmin{q \in W} \Norm{q - f}^2$, we have $\forall v \in W: \Chev{p-f,v}=0$. This motivates the following:

\begin {definition}
\label{def:obe_sharp}

Consider $(\Dist,f)$ a distributional estimation problem and ${P: \Words \Scheme \Rats}$ with bounded range. $P$ is called an \emph{$\ESG$-optimal polynomial-time estimator for $(\Dist,f)$} when for any\\ ${S: \Words \times \Rats \Scheme \Rats}$ with bounded range\footnote{The $\Rats$-valued argument of $S$ is only important for non-trivial $\GrowR$, otherwise we can absorb it into the definition of $S$ using $P$ as a subroutine.}

\begin{equation}
\label{eqn:op_sharp}
\E_{(x,y,z) \sim \Dist^{K} \times \Un_P^{K} \times \Un_S^{K}}[(P^{K}(x,y) - f(x))S^{K}(x,P^{K}(x,y),z)] \equiv 0 \pmod \Fall
\end{equation}

For the sake of brevity, we will say \enquote{${\ESG}$-optimal estimator} rather than \enquote{${\ESG}$-optimal polynomial-time estimator.}
This definition is interesting because it can be interpreted as a game against an adversary that is allowed to look at what the estimator outputs, which then predicts whether the estimator will overestimate or underestimate the true value. $\ESG$-optimal polynomial-time estimators are inexploitable against this class of adversaries. As we will show shortly, inexploitability is a slightly stronger condition than optimality, in the sense than any $\ESG$-optimal polynomial-time estimator is $\mathcal{F}$-optimal, but going in the other direction requires at most logarithmic advice and is associated with an increase in the error. The inexploitability property will be used in many additional proofs.
\end {definition}

The following theorem is the analogue in our language of the previous fact about inner product spaces. The notation $\Fall^{\frac{1}{2}}$ refers to Definition~\ref{def:fall_space_power}, i.e. it is just the set of square roots of all the functions in $\Fall$.

\begin{theorem}
\label{thm:ort}

Assume there is $\zeta: \NatFun (0,\frac{1}{4}]$ s.t. $\zeta \in \Fall^{\frac{1}{2}}$ and ${\Floor{\log \log \frac{1}{\zeta}} \in \GrowA}$\footnote{If $\GammaLog^n \subseteq \GrowA$ then this condition holds for any $\Fall$ since we can take $\zeta = 2^{-h}$ for $h \in \NatPoly$.}. Consider $(\Dist,f)$ a distributional estimation problem and $P$ an $\EG$-optimal estimator for $(\Dist,f)$. Then, $P$ is also an $\Fall^{\frac{1}{2}\sharp}(\Gamma)$-optimal estimator for $(\Dist,f)$.

\end{theorem}

\begin{proof}

Assume without loss of generality that there is ${h \in \NatPoly}$ s.t. $\zeta \geq 2^{-h}$ (otherwise we can take any $h \in \NatPoly$ s.t. $2^{-h} \in \Fall$ and consider $\zeta':=\zeta+2^{-h}$). Fix $S: \Words \times \Rats \Scheme \Rats$ bounded. Consider any ${\sigma: \NatFun \{ \pm 1 \}}$ and $m: \NatFun \Nats$ s.t. $m \leq \log \frac{1}{\zeta}$ (in particular ${m \leq h}$). Define ${t(K) := \sigma(K) 2^{-m(K)}}$. It is easy to see there is ${Q_t: \Words \Scheme \Rats}$ s.t. ${\R_{Q_t}=\R_P+\R_S}$ and given $K \in \Nats^n$, $x \in \Supp \Dist^{K}$, ${y \in \WordsLen{\R_P(K)}}$ and ${z \in \WordsLen{ \R_S(K)}}$

$$Q_t^{K}(x,yz) = P^{K}(x,y) - t(K) S^{K}(x,P^{K}(x,y),z)$$

Moreover, we can construct $Q_t$ for all admissible choices of $t$ (but fixed $S$) to get a uniform family.

Applying Proposition~\ref{prp:unif}, we conclude that there is $\varepsilon \in \Fall$ which doesn't depend on $t$ s.t.

$$\E_{\Dist^{K} \times \Un_P^{K}}[(P^{K} - f)^2] \leq \E_{\Dist^{K} \times \Un_P^{K} \times \Un_S^{K}}[(Q_t^{K} - f)^2] + \varepsilon(K)$$

$$\E_{\Dist^{K} \times \Un_P^{K}}[(P^{K} - f)^2] \leq \E_{\Dist^{K} \times \Un_P^{K} \times \Un_S^{K}}[(P^{K} - t(K)S^{K}  - f)^2] + \varepsilon(K)$$

$$\E_{\Dist^{K} \times \Un_P^{K} \times \Un_S^{K}}[(P^{K} - f)^2 - (P^{K} - t(K)S^{K} - f)^2] \leq \varepsilon(K)$$

$$\E_{\Dist^{K} \times \Un_P^{K} \times \Un_S^{K}}[(-t(K)(S^{K})^2 + 2 (P^{K} - f)) S^{K}] t(K) \leq \varepsilon(K)$$

$$-\E_{\Dist^{K} \times \Un_P^{K} \times \Un_S^{K}}[(S^{K})^2] t(K)^2 + 2 \E_{\Dist^{K} \times \Un_P^{K} \times \Un_S^{K}}[(P^{K} - f) S^{K}] t(K) \leq \varepsilon(K)$$

$$2 \E_{\Dist^{K} \times \Un_P^{K} \times \Un_S^{K}}[(P^{K} - f) S^{K}] t(K) \leq \E_{\Dist^{K} \times \Un_P^{K} \times \Un_S^{K}}[(S^{K})^2] t(K)^2 + \varepsilon(K)$$

$$2 \E_{\Dist^{K} \times \Un_P^{K} \times \Un_S^{K}}[(P^{K} - f) S^{K}] t(K) \leq (\sup \Abs{S^{K}})^2 t(K)^2 + \varepsilon(K)$$

$$2 \E_{\Dist^{K} \times \Un_P^{K} \times \Un_S^{K}}[(P^{K} - f) S^{K}] \sigma(K) 2^{-m(K)} \leq (\sup \Abs{S^{K}})^2 4^{-m(K)} + \varepsilon(K)$$

Multiplying both sides by $2^{m(K)-1}$ we get

$$\E_{\Dist^{K} \times \Un_P^{K} \times \Un_S^{K}}[(P^{K} - f) S^{K}] \sigma(K) \leq \frac{1}{2}\left((\sup \Abs{S^{K}})^2 2^{-m(K)} + \varepsilon(K) 2^{m(K)}\right)$$

Let $\sigma(K):=\Sgn \E_{\Dist^{K} \times \Un_S^{K}}[(P^{K} - f) S^{K}]$.

$$\Abs{\E_{\Dist^{K} \times \Un_P^{K} \times \Un_S^{K}}[(P^{K} - f) S^{K}]} \leq \frac{1}{2}((\sup \Abs{S^{K}})^2 2^{-m(K)} + \varepsilon(K) 2^{m(K)})$$

Let $m(K):=\min\left(\Floor{\frac{1}{2}\log \max(\frac{1}{\varepsilon(K)},1)},\Floor{\log \frac{1}{\zeta(K)}}\right)$.

$$\Abs{\E[(P^{K} - f) S^{K}]} \leq (\sup \Abs{S^{K}})^2 \max\left(\min\left(\varepsilon(K)^{\frac{1}{2}},1\right),\zeta(K)\right) + \frac{1}{2}\varepsilon(K) \min\left(\max\left(\varepsilon(K)^{-\frac{1}{2}},1\right),\zeta(K)^{-1}\right)$$

$$\Abs{\E[(P^{K} - f) S^{K}]} \leq (\sup \Abs{S^{K}})^2 \max\left(\varepsilon(K)^{\frac{1}{2}},\zeta(K)\right) + \frac{1}{2} \max\left(\varepsilon(K)^{\frac{1}{2}},\varepsilon(K)\right)$$

The right hand side is obviously in $\Fall^{\frac{1}{2}}$.
\end{proof}

Note that it would still be possible to prove Theorem~\ref{thm:ort} if in Definition~\ref{def:obe_sharp} we allowed ${S}$ to depend on ${y}$ directly instead of only through ${P}$. However, the definition as given appears more natural since it seems necessary to prove Theorem~\ref{thm:mult} in full generality.

Conversely to Theorem~\ref{thm:ort}, we have the following:

\begin{theorem}
\label{thm:con_ort}

Consider $(\Dist,f)$ a distributional estimation problem and $P$ an $\ESG$-optimal estimator for $(\Dist,f)$. Then, $P$ is also an $\EG$-optimal estimator for $(\Dist,f)$.

\end{theorem}

\begin{proof}

Consider any $Q: \Words \Scheme \Rats$. We have

$$\E_{\Dist^{K} \times \Un_Q^{K}}[(Q^{K}-f)^2]=\E_{\Dist^{K} \times \Un_Q^{K} \times \Un_P^{K}}[(Q^{K}-P^{K}+P^{K}-f)^2]$$

$$\E[(Q^{K}-f)^2]=\E[(Q^{K}-P^{K})^2]+2\E[(Q^{K}-P^{K})(P^{K}-f)]+\E[(P^{K}-f)^2]$$

$$\E[(P^{K}-f)^2]+\E[(Q^{K}-P^{K})^2]=\E[(Q^{K}-f)^2]+2\E[(P^{K}-Q^{K})(P^{K}-f)]$$

$$\E[(P^{K}-f)^2] \leq \E[(Q^{K}-f)^2] + 2\E[(P^{K}-Q^{K})(P^{K}-f)]$$

We can assume $Q$ is bounded without loss of generality since given any $Q$ it easy to construct bounded $\tilde{Q}$ s.t. $\E[(\tilde{Q}^{K}-f)^2] \leq \E[(Q^{K}-f)^2]$. Applying \ref{eqn:op_sharp}, we get \ref{eqn:op}.

\end{proof}

\subsection{Simple Example}
\label{sec:fundamentals__one_way}

The concept of an optimal polynomial-time estimator is in some sense complementary to the concept of pseudorandomness: a pseudorandom process deterministically produces output that appears random to bounded algorithms whereas optimal polynomial-time estimators compute the moments of the perceived random distributions of the outputs of deterministic processes. To demonstrate this complementarity and give an elementary example of an optimal polynomial-time estimator, we use the concept of a hard-core predicate (which may be regarded as an elementary example of pseudorandomness). The notation $\Fall_\text{neg}$ below refers to the fall space defined in Example~\ref{exm:fall_neg} (functions that fall faster than any polynomial). $\frac{1}{2}$ is an optimal polynomial-time estimator for a hard-core predicate.

\begin{theorem}
\label{thm:hard_core}

Consider $\Dist$ a word ensemble of rank ${1}$ s.t. for any different $k,l \in \Nats$,\\ ${\Supp \Dist^k \cap \Supp \Dist^l = \varnothing}$, ${f: \Supp \Dist \rightarrow \Words}$ one-to-one and $B$ a hard-core predicate of $(\Dist,f)$ (see Definition~\ref{def:hard_core}). Define ${m: \Supp \Dist \rightarrow \Nats}$ by 

\[\forall x \in \Supp \Dist^k: m(x):=k\]

For every $k \in \Nats$, define ${\Dist_f^k:=f_*^k\Dist^k}$.  Finally, define ${\chi_B: \Supp \Dist_f \rightarrow \Bool}$ by

\[\chi_B(f(x)):=B^{m(x)}(x)\]

Let $\Gamma:=(\GammaPoly^1,\Gamma_0^1)$. Let $P: \Words \Scheme \Rats$ satisfy $P \equiv \frac{1}{2}$. Then, $P$ is an $\Fall_{\textnormal{neg}}(\Gamma)$-optimal estimator for $(\Dist_f, \chi_B)$.

\end{theorem}

\begin{proof}

Assume to the contrary that $P$ is not optimal. Then there is ${Q: \Words \Scheme \Rats}$, $d \in \Nats$, an infinite set ${I \subseteq \Nats}$ and $\epsilon \in \Reals^{>0}$ s.t.

$$ \forall k \in I: \E_{\Dist_f^k}[(\frac{1}{2}-\chi_B)^2] \geq \E_{\Dist_f^k \times \Un_Q^k}[(Q^k-\chi_B)^2] +\frac{\epsilon}{k^d}$$

$$ \forall k \in I: \E_{\Dist_f^k \times \Un_Q^{k}}[(Q^{k}-\chi_B)^2] \leq \frac{1}{4} - \frac{\epsilon}{k^d} $$

$$ \forall k \in I: \E_{\Dist_f^k}[(\E_{\Un_Q^{k}}[Q^{k}]-\chi_B)^2] \leq \frac{1}{4} - \frac{\epsilon}{k^d} $$

There is $G: \Words \xrightarrow{\Gamma} \Bool$ s.t. for all ${x \in \Words}$, 

\[\Abs{\E[Q^{k}(x)]-\Pr[G^k(x)=1]}\leq 2^{-k}\] 

$G^k$ works by evaluating ${\alpha \leftarrow Q^{k}}$ and then returning 1 with probability ${\alpha \pm 2^{-k}}$ and 0 with probability $1-\alpha \pm 2^{-k}$, where the $2^{-k}$ error comes from rounding a rational number to a binary fraction. Denoting 

\[\delta(x):=\E[Q^{k}(x)]-\Pr[G^k(x)=1]\]

we get

$$ \forall k \in I: \E_{\Dist_f^k}[(\Prb_{\Un_G^k}[G^k=1]+\delta-\chi_B)^2] \leq \frac{1}{4} - \frac{\epsilon}{k^d} $$

$$ \forall k \in I: \E_{\Dist_f^k}[(\Prb_{\Un_G^k}[G^k=1]-\chi_f)^2]+2 \E_{\Dist_f^k}[(\Prb_{\Un_G^k}[G^k=1]-\chi_B)\delta]+\E_{\Dist_f^k}[\delta^2] \leq \frac{1}{4} - \frac{\epsilon}{k^d}$$

$$ \forall k \in I: \E_{\Dist_f^k}[(\Prb_{\Un_G^k}[G^k=1]-\chi_B)^2]-2 \cdot 2^{-k}- 4^{-k} \leq \frac{1}{4} - \frac{\epsilon}{k^d}$$

Since $2^{-k}$ falls faster than $k^{-d}$, there is $I_1 \subseteq \Nats$ infinite and $\epsilon_1 \in \Reals^{>0}$ s.t.

$$ \forall k \in I_1: \E_{\Dist_f^k}[(\Prb_{\Un_G^k}[G^k=1]-\chi_B)^2] \leq \frac{1}{4} - \frac{\epsilon_1}{k^d}$$

$$ \forall k \in I_1: \E_{\Dist_f^k}[\Abs{\Prb_{\Un_G^k}[G^k=1]-\chi_B}] \leq \sqrt{\frac{1}{4} - \frac{\epsilon_1}{k^d}} $$

$$ \forall k \in I_1: \E_{\Dist_f^k}[\Prb_{\Un_G^k}[G^k \ne \chi_B]] \leq \sqrt{\frac{1}{4} - \frac{\epsilon_1}{k^d}} $$

$$ \forall k \in I_1: \E_{x \sim \Dist^k}[\Prb_{\Un_G^k}[G^k(f(x)) \ne B^k(x)]] \leq \sqrt{\frac{1}{4} - \frac{\epsilon_1}{k^d}} $$

$$ \forall k \in I_1: \Prb_{\Dist^{k} \times \Un_G^k}[G^k(f(x)) \ne B^k(x)] \leq \sqrt{\frac{1}{4} - \frac{\epsilon_1}{k^d}} $$

Since $\sqrt{t}$ is a concave function and the derivative of $\sqrt{t}$ is $\frac{1}{2\sqrt{t}}$, we have $\sqrt{t} \leq \sqrt{t_0} + \frac{t-t_0}{2\sqrt{t_0}}$. Taking ${t_0}=\frac{1}{4}$ we get

$$ \forall k \in I_1: \Prb_{\Dist^{k} \times \Un_G^k}[G^k(f(x)) \ne B^k(x)] \leq \frac{1}{2}-\frac{\epsilon_1}{k^d}$$

$$ \forall k \in I_1: \Prb_{\Dist^{k} \times \Un_G^k}[G^k(f(x)) = B^k(x)] \geq \frac{1}{2}+\frac{\epsilon_1}{k^d}$$

This contradicts the definition of a hard-core predicate.
\end{proof}

\begin{corollary}
\label{crl:one_way}

Consider $f: \Words \Alg \Words$ a one-to-one one-way function. For every ${k \in \Nats}$, define $f^{(k)}: \WordsLen{k} \times \WordsLen{k} \rightarrow \Words$ by ${f^{(k)}(x,y):=\Chev{f(x),y}}$. Define the distributional estimation problem $(\Dist_{(f)}, \chi_f)$ by 

\begin{align*}
\Dist_{(f)}^k:=f_*^{(k)}(\Un^k \times \Un^k) \\
\chi_f(\Chev{f(x),y}):=x \cdot y
\end{align*}

Let $\Gamma:=(\GammaPoly^1,\Gamma_0^1)$. Let $P: \Words \Scheme \Rats$ satisfy $P \equiv \frac{1}{2}$. Then, $P$ is an $\Fall_{\textnormal{neg}}(\Gamma)$-optimal estimator for $(\Dist_{(f)}, \chi_f)$.

\end{corollary}

\begin{proof}

Follows immediately from Theorem~\ref{thm:hard_core} and Theorem~\ref{thm:goldreich_levin}.
\end{proof}

The following is the non-uniform version of Theorem~\ref{thm:hard_core} which we state without proof since the proof is a straightforward adaptation of the above.

\begin{theorem}
\label{thm:hard_core_circ}

Consider $\Dist$ a word ensemble s.t. for any different $k,l \in \Nats$,\\ $\Supp \Dist^k \cap \Supp \Dist^l = \varnothing$, $f: \Supp \Dist \rightarrow \Words$ one-to-one and $B$ a non-uniformly hard-core predicate of $(\Dist,f)$ (see Definition~\ref{def:hard_core_adv}). 

Let $\Gamma:=(\GammaPoly^1,\GammaPoly^1)$. Let $P: \Words \Scheme \Rats$ satisfy $P \equiv \frac{1}{2}$. Then, $P$ is an $\Fall_{\textnormal{neg}}(\Gamma)$-optimal estimator for $(\Dist_f, \chi_B)$.

\end{theorem}

\begin{corollary}

Consider $f: \Words \Alg \Words$ a one-to-one non-uniformly hard to invert one-way function.

Let $\Gamma:=(\GammaPoly^1,\GammaPoly^1)$. Let $P: \Words \Scheme \Rats$ satisfy $P \equiv \frac{1}{2}$. Then, $P$ is an $\Fall_{\textnormal{neg}}(\Gamma)$-optimal estimator for $(\Dist_f, \chi_f)$.

\end{corollary}

\begin{proof}

Follows immediately from Theorem~\ref{thm:hard_core_circ} and Theorem~\ref{thm:goldreich_levin_circ}.
\end{proof}

\section{Optimal Estimators and Probability Theory}
\label{sec:probability}

\subsection{Calibration}

From a Bayesian perspective, a good probability assignment should be well calibrated (see e.g. \cite{Dawid_1982}). For example, suppose there are 100 people in a room and you assign each person a probability they are married. If there are 60 people you assigned probabilities in the range 70\%-80\%, the number of married people among these 60 should be close to the interval $60 \times [0.7, 0.8] = [42,48]$. The same requirement can be made for expected value assignments. For example, if you now need to assign an expected value to the age of each person and you assigned an expected age in the range 30-40 to some sufficiently large group of people, the mean age in the group should be close to the interval $[30,40]$. 

We will now show that optimal polynomial-time estimators satisfy an analogous property.

\begin{theorem}
\label{thm:calib}

Assume $\Fall$ is $\GrowA$-ample. Consider $(\Dist,f)$ a distributional estimation problem, $P$ an $\EG$-optimal estimator for $(\Dist,f)$ and ${W: \Words \Scheme \Rats^{\geq 0}}$ bounded s.t. $\R_W \geq \R_P$ and for every $K \in \Nats^n$ there is $x \in \Supp \Dist^{K}$ and $y \in \Un_W^{K}$ with $W^{K}(x,y) > 0$. Denote

\begin{align*}
\alpha(K)&:=\E_{(x,y) \sim \Dist^{K} \times \Un_W^{K}}[W^{K}(x,y)] \\ 
\delta(K)&:=\E_{(x,y) \sim \Dist^{K} \times \Un_W^{K}}[W^{K}(x,y)(P^{K}(x,y_{<\R_P(K)})-f(x))]
\end{align*} 

Then, $\alpha^{-1}\delta^2 \in \Fall$.

\end{theorem}
Looking at the definition of an $\ESG$-optimal estimator, we see that, when $Q$ is $\{0,1\}$-valued to pick out a small subset of inputs, $P$ may be biased on that small subset, because the bias isn't normalized by dividing by the fraction of probability mass where $Q$ outputs 1. In the language of the above theorem, $\delta$ lies in $\Fall$, but $\frac{\delta}{\alpha}$ may not.

Proposition~\ref{prp:weight} says that, given ample advice, optimal estimators continue to be optimal (although with increased error) on small subsets of their input, which is a slightly stronger condition. Therefore, the above theorem essentially says that if enough advice is available for Proposition~\ref{prp:weight} to apply, the property of resistance to reweighting implies that, for the subset that $W$ picks out, the unnormalized bias times the normalized bias is a small term that lies in $\Fall$.

To see the relationship between Theorem~\ref{thm:calib} and calibration, consider the following corollary.

\begin{corollary}
\label{crl:calib}

Assume $\Fall$ is $\GrowA$-ample. Consider $(\Dist,f)$ a distributional estimation problem, $P$ an $\EG$-optimal estimator for $(\Dist,f)$ and $A,B: \bm{1} \Scheme \Rats$ s.t. $\R_A \equiv 0$ and $\R_B \equiv 0$. Denote

\[\alpha(K):=\Prb_{(x,y) \sim \Dist^{K} \times \Un_P^{K}}[A^{K} \leq P^{K}(x,y) \leq B^{K}]\] 

Then, there is $\varepsilon \in \Fall$ s.t. 

\begin{equation}
\label{eqn:crl__calib}
A^{K} - \sqrt{\frac{\varepsilon(K)}{\alpha(K)}} \leq \E_{(x,y) \sim \Dist^{K} \times \Un_P^{K}}[f(x) \mid A^{K} \leq P^{K}(x,y) \leq B^{K}] \leq B^{K} + \sqrt{\frac{\varepsilon(K)}{\alpha(K)}}
\end{equation}

\end{corollary}

The appearance of $\alpha$ in the denominator in \ref{eqn:crl__calib} is not surprising since we only expect calibration to hold for large sample size.

We now proceed with the proofs.

\begin{proof}[Proof of Corollary \ref{crl:calib}]

Construct $W: \Words \Scheme \Bool$ s.t. 

\begin{align*}
\R_W(K) &= \R_P(K) \\
W^{K}(x,y)&=\theta(P^{K}(x,y)-A^{K})\theta(B^{K}-P^{K}(x,y))
\end{align*}

Denote $\delta(K):=\E_{\Dist^{K} \times \Un_P^{K}}[W^{K}(P^{K}-f)]$ and $\varepsilon:=\frac{\delta^2}{\alpha}$. According to Theorem~\ref{thm:calib}, $\varepsilon \in \Fall$.
We get

$$\frac{\E_{\Dist^{K} \times \Un_P^{K}}[W^{K}(P^{K}-f)]^2}{\alpha(K)} = \varepsilon(K)$$

$$\frac{\E_{\Dist^{K} \times \Un_P^{K}}[\theta(P^{K}(x,y)-A^{K})\theta(B^{K}-P^{K}(x,y))(P^{K}-f)]^2}{\alpha(K)} = \varepsilon(K)$$

$$\frac{(\E_{\Dist^{K} \times \Un_P^{K}}[\theta(P^{K}(x,y)-A^{K})\theta(B^{K}-P^{K}(x,y))]\E[P^{K}-f \mid A^{K} \leq P^{K} \leq B^{K}])^2}{\alpha(K)} = \varepsilon(K)$$

$$\frac{(\alpha(K)\E[P^{K}-f \mid A^{K} \leq P^{K} \leq B^{K}])^2}{\alpha(K)} = \varepsilon(K)$$

$$\alpha(K)\E[P^{K}-f \mid A^{K} \leq P^{K} \leq B^{K}]^2 = \varepsilon(K)$$

\begin{equation}
\label{eqn:crl__calib__prf}
\Abs{\E[P^{K}-f \mid A^{K} \leq P^{K} \leq B^{K}]} = \sqrt{\frac{\varepsilon(K)}{\alpha(K)}}
\end{equation}

On the other hand

$$\E[f \mid A^{K} \leq P^{K} \leq B^{K}] = \E[P^{K}-P^{K}+f \mid A^{K} \leq P^{K} \leq B^{K}]$$

$$\E[f \mid A^{K} \leq P^{K} \leq B^{K}] = \E[P^{K} \mid A^{K} \leq P^{K} \leq B^{K}]-\E[P^{K}-f \mid A^{K} \leq P^{K} \leq B^{K}]$$

Applying \ref{eqn:crl__calib__prf}

$$\E[f \mid A^{K} \leq P^{K} \leq B^{K}] \leq \E[P^{K} \mid A^{K} \leq P^{K} \leq B^{K}]+\sqrt{\frac{\varepsilon(K)}{\alpha(K)}}$$

$$\E[f \mid A^{K} \leq P^{K} \leq B^{K}] \leq B^{K} + \sqrt{\frac{\varepsilon(K)}{\alpha(K)}}$$

In the same manner, we can show that

$$\E[f \mid A^{K} \leq P^{K} \leq B^{K}] \geq A^{K} - \sqrt{\frac{\varepsilon(K)}{\alpha(K)}}$$
\end{proof}

\begin{proof}[Proof of Theorem \ref{thm:calib}]

Consider $\zeta: \NatFun (0,\frac{1}{2}]$ s.t.  $\zeta \in \Fall$ and $\Floor{\log \frac{1}{\zeta}} \in \GrowA$. Define 

\begin{align*}
I&:=\{K \in \Nats^n \mid \frac{\Abs{\delta(K)}}{\alpha(K)} \geq \zeta(K)\} \\
E^{K}&:=\Rats \cap \left[\frac{\Abs{\delta(K)}}{2\alpha(K)},\frac{\Abs{\delta(K)}}{\alpha(K)}\right] \\
\epsilon(K) &\in (\Sgn \delta(K)) \cdot \Argmin{t \in E^{K}} \Abs{\En_\Rats(t)}
\end{align*}

It is easy to see that ${\Abs{\En_\Rats(\epsilon)} = O(\log \frac{\alpha}{\Abs{\delta}})}$, hence we can construct $Q: \Words \Scheme \Rats$ s.t. for any $K \in I$ and $x,y \in \Words$

\begin{align*}
\A_Q(K) &= \En_\Rats(\epsilon(K)) \\
\R_Q(K) &= \R_P(K) \\
Q^{K}(x,y) &= P^{K}(x,y)-\epsilon(K)
\end{align*}

This algorithm uses the advice string to check whether the normalized bias is too high, and if it is, it perturbs the estimated values accordingly.

Applying Proposition~\ref{prp:weight} to $P$, $Q$ and $W$, we conclude there is $\varepsilon \in \Fall$ s.t.

$$\E_{\Dist^{K} \times \Un_W^{K}}[W^{K}(P^{K} - f)^2] \leq \E_{\Dist^{K} \times \Un_W^{K}}[W^{K}(Q^{K}-f)^2] + \varepsilon(K)$$

$$\E_{\Dist^{K} \times \Un_W^{K}}[W^{K}(P^{K} - f)^2] \leq \E_{\Dist^{K} \times \Un_W^{K}}[W^{K}(P^{K}-f-\epsilon(K))^2] + \varepsilon(K)$$

$$\E_{\Dist^{K} \times \Un_W^{K}}[W^{K}((P^{K} - f)^2 - (P^{K}-f-\epsilon(K))^2] \leq \varepsilon(K)$$

$$ \epsilon(K) \E_{\Dist^{K} \times \Un_W^{K}}[W^{K}(2(P^{K} - f) - \epsilon(K))] \leq \varepsilon(K)$$

$$ \epsilon(K) (2 \E_{\Dist^{K} \times \Un_W^{K}}[W^{K}(P^{K} - f)]-\E_{\Dist^{K} \times \Un_W^{K}}[W^{K}]\epsilon(K)) \leq \varepsilon(K)$$

$$ \epsilon(K) (2 \delta(K) - \alpha(K)\epsilon(K)) \leq \varepsilon(K)$$

Dividing both sides by $ \alpha(K)$ we get

$$\epsilon(K) \left(\frac{2\delta(K)}{\alpha(K)} - \epsilon(K)\right) \leq \frac{\varepsilon(K)}{\alpha(K)}$$

$$\frac{\delta(K)^2}{\alpha(K)^2}-\left(\epsilon(K) - \frac{\delta(K)}{\alpha(K)}\right)^2 \leq \frac{\varepsilon(K)}{\alpha(K)}$$

$\epsilon$ is between $\frac{\delta}{2\alpha}$ and $\frac{\delta}{\alpha}$ therefore $(\epsilon-\frac{\delta}{\alpha})^2 \leq (\frac{\delta}{2\alpha} - \frac{\delta}{\alpha})^2$  which yields

$$\frac{\delta(K)^2}{\alpha(K)^2}-\left(\frac{\delta(K)}{2\alpha(K)} - \frac{\delta(K)}{\alpha(K)}\right)^2 \leq \frac{\varepsilon(K)}{\alpha(K)}$$

$$\frac{3}{4} \cdot \frac{\delta(K)^2}{\alpha(K)^2} \leq \frac{\varepsilon(K)}{\alpha(K)}$$

$$\frac{\delta(K)^2}{\alpha(K)} \leq \frac{4}{3}\varepsilon(K)$$
\end{proof}

\subsection{Algebraic Properties}
\label{subsec:alg}

In this subsection and subsection~\ref{subsec:indep_var}, we show that several algebraic identities satisfied by expected values have analogues for optimal polynomial-time estimators.

\subsubsection{Linearity}

Given $F_1,F_2$ random variables and $t_1,t_2 \in \Reals$, we have 

\begin{equation}
\E[t_1 F_1 + t_2 F_2] = t_1 \E[F_1] + t_2 \E[F_2]
\end{equation}

Optimal polynomial-time estimators have an analogous property:

\begin{proposition}
\label{prp:linearity}

Consider $\Dist$ a word ensemble, $f_1,f_2: \Supp \Dist \rightarrow \Reals$ bounded and $t_1,t_2 \in \Rats$. Denote $f: = t_1 f_1 + t_2 f_2$. Suppose $P_1$ is an $\ESG$-optimal estimator for $(\Dist,f_1)$ and $P_2$ is an $\ESG$-optimal estimator for $(\Dist,f_2)$. Construct $P: \Words \Scheme \Rats$ s.t. for any ${x \in \Supp \Dist^{K}}$, ${y_1 \in \WordsLen{\R_{P_1}(K)}}$ and $y_2 \in \WordsLen{\R_{P_1}(K)}$

\begin{align}
\A_P(K) &= \Chev{\A_{P_1}(K),\A_{P_2}(K)} \\
\R_P(K) &= \R_{P_1}(K) + \R_{P_2}(K) \\
P^{K}(x,y_1 y_2) &= t_1 P_1^{K}(x,y_1) + t_2 P_2^{K}(x, y_2)
\end{align}

Then, $P$ is an $\ESG$-optimal estimator for $(\Dist, f)$.

\end{proposition}

\begin{proof}

Consider any bounded $S: \Words \times \Rats \Scheme \Rats$. We have

$$\E[(P^{K} - f)S^{K}] = \E[(t_1 P_1^{K} + t_2 P_2^{K} - (t_1 f_1 + t_2 f_2))S^{K}]$$

$$\E[(P^{K} - f)S^{K}] = t_1 \E[(P_1^{K} - f_1)S^{K}] + t_2 \E[(P_2^{K} - f_2)S^{K}]$$

$$\Abs{\E[(P^{K} - f)S^{K}]} \leq \Abs{t_1} \cdot \Abs{\E[(P_1^{K} - f_1)S^{K}]} + \Abs{t_2} \cdot \Abs{\E[(P_2^{K} - f_2)S^{K}]}$$

Using \ref{eqn:op_sharp} for $P_1$ and $P_2$ we see that the right hand side is in $\Fall$.
\end{proof}

\subsubsection{Conditional Expectation}

Consider a random variable $F$ and an event $A$. Denote $\chi_A$ the $\Bool$-valued random variable corresponding to the indicator function of $A$. We have

\begin{equation}
\label{eqn:cond_ev}
\E[F \mid A] = \frac{\E[\chi_A F]}{\Prb[A]}
\end{equation}

This identity is tautologous if interpreted as a definition of $\E[F \mid A]$. However, from the perspective of Bayesian probability it is more natural to think of $\E[F \mid A]$ as an atomic entity (the subjective expectation of $F$ after observing $A$). 

The language of optimal polynomial-time estimators provides a natural way to define an analogue of conditional expectation. Namely, consider a distributional estimation problem $(\Dist, f)$ and a decision problem ${L \subseteq \Words}$. Then, $P: \Words \Scheme \Rats$ represents the conditional expectation of $f$ given $L$ when it is an optimal polynomial-time estimator for $(\Dist \mid L, f)$. That is, the conditional expectation is the best estimate of $f(x)$ when the problem instance $x$ is sampled with the \emph{promise} $x \in L$.

The above perspective allows us stating and proving non-tautological theorems analogous to \ref{eqn:cond_ev}. We give two such theorems, corresponding to two different ways to group the variables in \ref{eqn:cond_ev}. Let $\chi_{L}$ be the indicator function for $L$. The first states that an optimal estimator for $\chi_{L}f$ can be made by multiplying together an optimal estimator for $\chi_{L}$ and a less accurate optimal estimator for $f|L$, and the second theorem states that a less accurate optimal estimator for $f|L$ can be made by dividing the output of an optimal estimator for $\chi_{L}f$ by the output of an optimal estimator for $L$. The amplification of error appears because $L$ might be a low-probability event, and conditional probabilities for low-probability events are less accurate than conditional probabilities for high-probability events.

\begin{samepage}
\begin{theorem}
\label{thm:con_cond}

Consider $(\Dist, f)$ a distributional estimation problem and ${L \subseteq \Words}$ s.t. for all $K \in \Nats^n$, $\Dist^K(L) > 0$. Define $\gamma_L: \NatFun \Reals$ by $\gamma_L(K):=\Dist^{K}(L)^{-1}$ and $\Fall_L:=\gamma_L \Fall$. Let $P_L$ be an $\ESG$-optimal estimator for $(\Dist, \chi_L)$ and $P_{f \mid L}$ be an $\Fall_L^\sharp(\Gamma)$-optimal estimator for ${(\Dist \mid L, f)}$. Construct ${P_{\chi f}: \Words \Scheme \Rats}$ s.t. $\R_{P_{\chi f}}=\R_{P_L} + \R_{P_{f \mid L}}$ and for any ${x \in \Words}$, $y \in \BoolR{P_L}$ and ${z \in \BoolR{P_{f \mid L}}}$

\begin{equation}
P_{\chi f}^K(x,yz)=P_L^K(x,y) P_{f \mid L}^K(x,z)
\end{equation}

Then, $P_{\chi f}$ is an $\ESG$-optimal estimator for $(\Dist, \chi_L f)$.

\end{theorem}
\end{samepage}

\begin{proof}

Consider any $K \in \Nats^n$, $x \in \Supp \Dist^{K}$, $y \in \BoolR{P_L}$ and $z \in \BoolR{P_{f \mid L}}$.

\[P_{\chi f}^K(x,yz) - \chi_L(x) f(x) = P_L^K(x,y) P_{f \mid L}^K(x,z) - \chi_L(x) f(x)\]

\[P_{\chi f}^K(x,yz) - \chi_L(x) f(x) = P_L^K(x,y) P_{f \mid L}^K(x,z) - \chi_L(x) P_{f \mid L}^K(x,z) + \chi_L(x) P_{f \mid L}^K(x,z) - \chi_L(x) f(x)\]

\[P_{\chi f}^K(x,yz) - \chi_L(x) f(x) = (P_L^K(x,y) - \chi_L(x)) P_{f \mid L}^K(x,z) + \chi_L(x) (P_{f \mid L}^K(x,z) - f(x))\]

Consider any $S: \Words \times \Rats \Scheme \Rats$ bounded. We get

\[\E_{\Dist^{K} \times \Un_{P_{\chi f}}^K \times \Un_S^K}[(P_{\chi f}^K - \chi_L f)S^K] = \E_{\Dist^{K} \times \Un_{P_{\chi f}}^K \times \Un_S^K}[(P_L^K - \chi_L) P_{f \mid L}^KS^K)] + \E_{\Dist^{K} \times \Un_{P_{\chi f}}^K \times \Un_S^K}[\chi_L (P_{f \mid L}^K - f)S^K]\]

Using the fact that $P_L^K$ is $\ESG$-optimal for $(\Dist,\chi_L)$,

\[\E_{\Dist^{K} \times \Un_{P_{\chi f}}^K \times \Un_S^K}[(P_{\chi f}^K - \chi_L f)S^K] \equiv  \E_{\Dist^{K} \times \Un_{P_{\chi f}}^K \times \Un_S^K}[\chi_L (P_{f \mid L}^K - f)S^K] \pmod \Fall\]

\[\E_{\Dist^{K} \times \Un_{P_{\chi f}}^K \times \Un_S^K}[(P_{\chi f}^K - \chi_L f)S^K] \equiv \Dist^{K}(L)  \E_{(\Dist^{K} \mid L) \times \Un_{P_{\chi f}}^K \times \Un_S^K}[(P_{f \mid L}^K - f)S^K] \pmod \Fall\]
 
Using the fact that $P_{f \mid L}^K$ is $\Fall_L^\sharp(\Gamma)$-optimal for $(\Dist \mid L, f)$, we conclude
 
 \[\Abs{\E_{\Dist^{K} \times \Un_{P_{\chi f}}^K \times \Un_S^K}[(P_{\chi f}^K - \chi_L f)S^K]} \equiv 0 \pmod \Fall\]
\end{proof}
 
\begin{samepage}
\begin{theorem}
\label{thm:cond}

Consider $(\Dist, f)$ a distributional estimation problem and ${L \subseteq \Words}$ s.t. for all $K \in \Nats^n$, $\Dist^K(L) > 0$. Define $\gamma_L: \NatFun \Reals$ by $\gamma(K):=\Dist^{K}(L)^{-1}$ and $\Fall_L:=\gamma_L \Fall$. Let $P_L$ be an $\ESG$-optimal estimator for $(\Dist, \chi_L)$ and $P_{\chi f}$ be an $\ESG$-optimal estimator for $(\Dist, \chi_L f)$. Choose any $M \in \Rats$ s.t. ${M \geq \sup \Abs{f}}$ and construct $P_{f \mid L}: \Words \Scheme \Rats$ s.t. $\R_{P_{f \mid L}} = \R_{P_L} + \R_{P_{\chi f}}$ and for any ${x \in \Words}$, ${y \in \BoolR{P_L}}$ and $z \in \BoolR{P_{\chi f}}$ 

\begin{equation}
P_{f \mid L}^K(x,yz)=\begin{cases}P_L^K(x,y)^{-1} P_{\chi f}^K(x,z) \textnormal{ if this number is in } [-M,M] \\ M \textnormal{ if } P_L^K(x,y)=0 \textnormal{ or } P_L^K(x,y)^{-1} P_{\chi f}^K(x,z) > M\\ -M \textnormal{ if } P_L^K(x,y)^{-1} P_{\chi f}^K(x,z) < -M\end{cases}
\end{equation}

Then, $P_{f \mid L}$ is an $\Fall_L^\sharp(\Gamma)$-optimal estimator for $(\Dist \mid L, f)$.

\end{theorem}
\end{samepage}

In order to prove Theorem~\ref{thm:cond}, we will need the following.

Consider $s,t \in \Rats$, an $[s,t]$-valued random variable $F$ and an event $A$. Denote $\chi_A$ the $\Bool$-valued random variable corresponding to the indicator function of $A$. We have 

\begin{equation}
\Prb[A]s \leq \E[\chi_A F] \leq \Prb[A]t
\end{equation}

For optimal polynomial-time estimators the analogous inequalities don't have to hold strictly (they only hold within an asymptotically small error), but the following proposition shows they can always be enforced.

\begin{samepage}
\begin{proposition}
\label{prp:thm__cond__lemma}

Consider $(\Dist, f)$ a distributional estimation problem, ${L \subseteq \Words}$ and $s, t \in \Rats$ s.t. ${s \leq \inf f}$, $t \geq \sup f$. Let $P_L$ be an $\ESG$-optimal estimator for $(\Dist, \chi_L)$ and $P_{\chi f}$ be an $\ESG$-optimal estimator for $(\Dist, \chi_L f)$. Construct $\tilde{P}_{\chi f}: \Words \Scheme \Rats$ s.t. $\R_{\tilde{P}_{\chi f}} = \R_{P_L} + \R_{P_{\chi f}}$ and for any ${y \in \BoolR{P_L}}$ and $z \in \BoolR{P_{\chi f}}$, ${\tilde{P}_{\chi f}^K(x,yz)=\min(\max(P_{\chi f}^K(x,z),P_L^K(x,y) s),P_L^K(x,y) t)}$. Denote 

\[\Dist_P^K:=\Dist^{K} \times \Un_{P_L}^K \times \Un_{P_{\chi f}}^K\] 

Then, for any $S: \Words \times \Rats^2 \Scheme \Rats$ bounded

\begin{equation}
\E_{\Dist_P^K \times \Un_S^K}[(\tilde{P}_{\chi f}^K(x)-\chi_L(x)f(x))S^K(x,P_L^K(x),P_{\chi f}^K(x))]  \equiv 0 \pmod \Fall
\end{equation}

In particular, $\tilde{P}$ is also an $\ESG$-optimal estimator for $(\Dist, \chi_L f)$.

\end{proposition}
\end{samepage}

\begin{proof}

$P_L$ is an $\ESG$-optimal estimator for $(\Dist, \chi_L)$, therefore

\begin{equation}
\label{eqn:thm__cond__lemma__prf1}
\E_{\Dist_P^K}[(P_L^K- \chi_L) \theta(P_{\chi f}^K- P_L^K t)] \equiv 0 \pmod \Fall
\end{equation}

$P_{\chi f}$ is an $\ESG$-optimal estimator for $(\Dist, \chi_L f)$, therefore

\begin{equation}
\label{eqn:thm__cond__lemma__prf2}
\E_{\Dist_P^K}[(P_{\chi f}^K - \chi_L f) \theta(P_{\chi f}^K - P_L^K t)] \equiv 0 \pmod \Fall
\end{equation}

Multiplying \ref{eqn:thm__cond__lemma__prf1} by $t$ and subtracting \ref{eqn:thm__cond__lemma__prf2} we get

\[\E_{\Dist_P^K}[(P_L^K t - P_{\chi f}^K - \chi_L \cdot (t - f)) \theta(P_{\chi f}^K- P_L^K t)] \equiv 0 \pmod \Fall\]

\[\E_{\Dist_P^K}[(P_L^K t - P_{\chi f}^K ) \theta(P_{\chi f}^K- P_L^K t)] \equiv \E_{\Dist_P^K}[\chi_L \cdot (t- f) \theta(P_{\chi f}^K- P_L^K t)] \pmod \Fall\]

The left-hand side is non-positive and the right-hand side is non-negative, therefore

\[\E_{\Dist_P^K}[(P_L^K t - P_{\chi f}^K ) \theta(P_{\chi f}^K- P_L^K t)] \equiv 0 \pmod \Fall\]

\begin{equation}
\label{eqn:thm__cond__lemma__prf3}
\E_{\Dist_P^K}[(\tilde{P}_{\chi f}^K - P_{\chi f}^K) \theta(P_{\chi f}^K - \tilde{P}_{\chi f}^K)] \equiv 0 \pmod \Fall
\end{equation}

In the same way we can show that

\[\E_{\Dist_P^K}[(P_L^K s - P_{\chi f}^K) \theta(P_L^K s-P_{\chi f}^K)] \equiv 0 \pmod \Fall\]

\begin{equation}
\label{eqn:thm__cond__lemma__prf4}
\E_{\Dist_P^K}[(\tilde{P}_{\chi f}^K - P_{\chi f}^K) \theta(\tilde{P}_{\chi f}^K-P_{\chi f}^K )] \equiv 0 \pmod \Fall
\end{equation}

Subtracting \ref{eqn:thm__cond__lemma__prf3} from \ref{eqn:thm__cond__lemma__prf4}, we get

\[\E_{\Dist_P^K}[(\tilde{P}_{\chi f}^K - P_{\chi f}^K) (\theta(\tilde{P}_{\chi f}^K-P_{\chi f}^K)- \theta(P_{\chi f}^K - \tilde{P}_{\chi f}^K))] \equiv 0 \pmod \Fall\]

\begin{equation}
\label{eqn:thm__cond__lemma__prf5}
\E_{\Dist_P^K}[\Abs{\tilde{P}_{\chi f}^K - P_{\chi f}^K}] \equiv 0 \pmod \Fall
\end{equation}

Consider any $S: \Words \times \Rats^2 \Scheme \Rats$ bounded.

\[\E_{\Dist_P^K \times \Un_S^K}[(\tilde{P}_{\chi f}^K - \chi_L f) S^K(x,P_L^K,P_{\chi f}^K)]=\E_{\Dist_P^K \times \Un_S^K}[(\tilde{P}_{\chi f}^K - P_{\chi f}^K + P_{\chi f}^K - \chi_L f) S^K(x,P_L^K,P_{\chi f}^K)]\]

\[\E_{\Dist_P^K \times \Un_S^K}[(\tilde{P}_{\chi f}^K - \chi_L f) S^K]=\E_{\Dist_P^K \times \Un_S^K}[(\tilde{P}_{\chi f}^K - P_{\chi f}^K) S^K]+\E_{\Dist_P^K \times \Un_S^K}[( P_{\chi f}^K - \chi_L f) S^K]\]

Using the fact that $P_{\chi f}$ is an $\ESG$-optimal estimator for $(\Dist, \chi_L f)$, we get

\[\E_{\Dist_P^K \times \Un_S^K}[(\tilde{P}_{\chi f}^K - \chi_L f) S^K] \equiv \E_{\Dist_P^K \times \Un_S^K}[(\tilde{P}_{\chi f}^K - P_{\chi f}^K) S^K] \pmod \Fall\]

\[\Abs{\E_{\Dist_P^K \times \Un_S^K}[(\tilde{P}_{\chi f}^K - \chi_L f) S^K]} \leq \E_{\Dist_P^K \times \Un_S^K}[\Abs{\tilde{P}_{\chi f}^K - P_{\chi f}^K}] \sup S \pmod \Fall\]

Applying \ref{eqn:thm__cond__lemma__prf5} we conclude that

\[\E_{\Dist_P^K \times \Un_S^K}[(\tilde{P}_{\chi f}^K - \chi_L f) S^K] \equiv 0 \pmod \Fall\]
\end{proof}

\begin{proof}[Proof of Theorem \ref{thm:cond}]

Construct $\tilde{P}_{\chi f}: \Words \Scheme \Rats$ s.t. $\R_{\tilde{P}_{\chi f}} = \R_{P_L} + \R_{P_{\chi f}}$ and for any ${x \in \Words}$, ${y \in \BoolR{P_L}}$ and $z \in \BoolR{P_{\chi f}}$

\[\tilde{P}_{\chi f}^K(x,yz)=\min(\max(P_{\chi f}^K(x,z),-P_L^K(x,y) M),P_L^K(x,y) M)\] 

For any ${x \in \Words}$, ${y \in \BoolR{P_L}}$ and $z \in \BoolR{P_{\chi f}}$, we have 

\[
\tilde{P}_{\chi f}^K(x,yz) = P_L^K(x,y) P_{f \mid L}^K(x,yz)\]

\[\tilde{P}_{\chi f}^K(x,yz) - \chi_L(x) f(x) = P_L^K(x,y) P_{f \mid L}^K(x,yz) - \chi_L(x) f(x)\]

\[\tilde{P}_{\chi f}^K(x,yz) - \chi_L(x) f(x) = P_L^K(x,y) P_{f \mid L}^K(x,z) - \chi_L(x) P_{f \mid L}^K(x,yz) + \chi_L(x) P_{f \mid L}^K(x,yz) - \chi_L(x) f(x)\]

\[\tilde{P}_{\chi f}^K(x,yz) - \chi_L(x) f(x) = (P_L^K(x,y) - \chi_L(x)) P_{f \mid L}^K(x,yz) + \chi_L(x) (P_{f \mid L}^K(x,yz) - f(x))\]

\[\chi_L(x) (P_{f \mid L}^K(x,yz) - f(x)) = \tilde{P}_{\chi f}^K(x,yz) - \chi_L(x) f(x) - (P_L^K(x,y) - \chi_L(x)) P_{f \mid L}^K(x,yz)\]

Consider any $S: \Words \times \Rats \Scheme \Rats$ bounded. Denote 

\[\Dist_{PS}^K:=\Dist^{K} \times \Un_{P_L}^K \times \Un_{P_{\chi f}}^K \times \Un_S^K\]

We have

\[\E_{\Dist_{PS}^K}[\chi_L (P_{f \mid L}^K - f)S^K(x,P_{f \mid L}^K)] = E_{\Dist_{PS}^K}[(\tilde{P}_{\chi f}^K - \chi_L f)S^K(x,P_{f \mid L}^K)] - E_{\Dist_{PS}^K}[(P_L^K - \chi_L) P_{f \mid L}^K S^K(x,P_{f \mid L}^K)]\]

Applying Proposition~\ref{prp:thm__cond__lemma} to the first term on the right-hand side and the fact $P_L^K$ is an $\ESG$-optimal estimator for $(\Dist,\chi_L)$ to the second term on the right-hand side,

\[\E_{\Dist_{PS}^K}[\chi_L (P_{f \mid L}^K - f)S^K(x,P_{f \mid L}^K)] \equiv 0 \pmod \Fall\]

\[\Dist^{K}(L) \E_{(\Dist^{K} \mid L)\times \Un_{P_L}^K \times \Un_{P_{\chi f}}^K \times \Un_S^K}[(P_{f \mid L}^K - f)S^K(x,P_{f \mid L}^K)] \equiv 0 \pmod \Fall\]

\[\E_{(\Dist^{K} \mid L)\times \Un_{P_L}^K \times \Un_{P_{\chi f}}^K \times \Un_S^K}[(P_{f \mid L}^K - f)S^K(x,P_{f \mid L}^K)] \equiv 0 \pmod {\Fall_L}\]
\end{proof}
\subsection{Polynomial-Time \texorpdfstring{$\MGrow$}{MΓ}-Schemes and Samplers}

The next subsection and subsequent sections will require several new concepts. Here, we introduce these concepts and discuss some of their properties.

\subsubsection{Congruent Measure Families}

The notation $f(K) \equiv g(K) \pmod \Fall$ can be conveniently generalized from real-valued functions to families of probability distributions.

\begin{samepage}
\begin{definition}

Consider a set $X$ and two families $\{\Dist^K \in \mathcal{P}(X)\}_{K \in \Nats^n}$ and $\{\mathcal{E}^K \in \mathcal{P}(X)\}_{K \in \Nats^n}$. We say that \emph{$\Dist$ is congruent to $\mathcal{E}$ modulo $\Fall$} when $\Dtv(\Dist^K,\mathcal{E}^K) \in \Fall$. In this case we write ${\Dist^K \equiv \mathcal{E}^K \pmod \Fall}$ or $\Dist \equiv \mathcal{E} \pmod \Fall$.

\end{definition}
\end{samepage}

Congruence of probability distributions modulo $\Fall$ has several convenient properties which follow from elementary properties of total variation distance.

\begin{samepage}
\begin{proposition}
\label{prp:prob_cong_eq}

Congruence of probability distributions modulo $\Fall$ is an equivalence relation.

\end{proposition}
\end{samepage}

\begin{proof}

Obvious since $\Dtv$ is a metric.
\end{proof}

\begin{samepage}
\begin{proposition}
\label{prp:prob_cong_ev}

Consider $X$ a set, $\{\Dist^K \in \mathcal{P}(X)\}_{K \in \Nats^n}$, $\{\mathcal{E}^K \in \mathcal{P}(X)\}_{K \in \Nats^n}$ and\\ $\{f^K: X \rightarrow \Reals\}_{K \in \Nats^n}$ a uniformly bounded family of functions. Assume ${\Dist \equiv \mathcal{E} \pmod \Fall}$. Then

\begin{equation}
\E_{x \sim \Dist^K}[f^K(x)] \equiv \E_{x \sim \mathcal{E}^K}[f^K(x)] \pmod \Fall
\end{equation}

\end{proposition}
\end{samepage}

\begin{proof}

$\Abs{\E_{x \sim \Dist^K}[f^K(x)] - \E_{x \sim \mathcal{E}^K}[f^K(x)]} \leq  (\sup f - \inf f)\Dtv(\Dist^K, \mathcal{E}^K)$
\end{proof}

\begin{samepage}
\begin{proposition}
\label{prp:prob_cong_semidir}

Consider $X$, $Y$ sets, $\{\Dist^K \in \mathcal{P}(X)\}_{K \in \Nats^n}$, $\{\mathcal{E}^K \in \mathcal{P}(X)\}_{K \in \Nats^n}$ and\\ $\{f^K: X \Markov Y\}_{K \in \Nats^n}$ a family of Markov kernels. Then, $\Dist \equiv \mathcal{E} \pmod \Fall$ implies

\begin{equation}
\Dist^K \ltimes f^K \equiv \mathcal{E}^K \ltimes f^K \pmod \Fall
\end{equation}

\end{proposition}
\end{samepage}

\begin{proof}

Total variation distance is contracted by semi-direct product with a Markov kernel therefore $\Dtv(\Dist^K \ltimes f^K, \mathcal{E}^K \ltimes f^K) \leq \Dtv(\Dist^K, \mathcal{E}^K)$.
\end{proof}

\begin{samepage}
\begin{proposition}
\label{prp:prob_cong_push}

Consider $X$, $Y$ sets, $\{\Dist^K \in \mathcal{P}(X)\}_{K \in \Nats^n}$, $\{\mathcal{E}^K \in \mathcal{P}(X)\}_{K \in \Nats^n}$ and\\ $\{f^K: X \Markov Y\}_{K \in \Nats^n}$ a family of Markov kernels. Then, $\Dist \equiv \mathcal{E} \pmod \Fall$ implies

\begin{equation}
f_*^K\Dist^K \equiv f_*^K\mathcal{E}^K \pmod \Fall
\end{equation}

\end{proposition}
\end{samepage}

\begin{proof}

Total variation distance is contracted by pushforward therefore \[\Dtv(f_*^K\Dist^K, f_*^K\mathcal{E}^K) \leq \Dtv(\Dist^K, \mathcal{E}^K)\]
\end{proof}

\begin{samepage}
\begin{proposition}
\label{prp:prob_cong_dir}

Consider $X_1$, $X_2$ sets, $\{\Dist_1^K \in \mathcal{P}(X_1)\}_{K \in \Nats^n}$, $\{\mathcal{E}_1^K \in \mathcal{P}(X_1)\}_{K \in \Nats^n}$,\\ $\{\Dist_2^K \in \mathcal{P}(X_2)\}_{K \in \Nats^n}$ and $\{\mathcal{E}_2^K \in \mathcal{P}(X_2)\}_{K \in \Nats^n}$. Then, $\Dist_1 \equiv \mathcal{E}_1 \pmod \Fall$ and $\Dist_2 \equiv \mathcal{E}_2 \pmod \Fall$ imply

\begin{equation}
\Dist_1^K \times \Dist_2^K \equiv \mathcal{E}_1^K \times \mathcal{E}_2^K \pmod \Fall
\end{equation}

\end{proposition}
\end{samepage}

\begin{proof}

Total variation distance is subadditive w.r.t. direct products therefore 

\[\Dtv(\Dist_1^K \times \Dist_2^K, \mathcal{E}_1^K \times \mathcal{E}_2^K) \leq \Dtv(\Dist_1^K, \mathcal{E}_1^K) + \Dtv(\Dist_2^K, \mathcal{E}_2^K)\]
\end{proof}

\subsubsection{Polynomial-Time \texorpdfstring{$\MGrow$}{MΓ}-Schemes}

The concept of a polynomial-time $\Gamma$-scheme can be generalized in a way which allows the advice to become random in itself.

\begin{samepage}
\begin{definition}

Given encoded sets $X$ and $Y$, a \emph{polynomial-time $\MGrow$-scheme of signature ${X \rightarrow Y}$} is a triple $(S,\R_S,\M_S)$ where $S: \Nats^n \times X \times \Words \times \Words \Alg Y$, ${\R_S: \Nats^n \times \Words \Alg \Nats}$ and\\ ${\{\M_S^K \in \mathcal{P}(\Words)\}_{K \in \Nats^n}}$ are s.t.

\begin{enumerate}[(i)]

\item $\max_{x \in X} \max_{y,z \in \Words} \T_S(K,x,y,z) \in \GammaPoly^n$

\item ${\max_{z \in \Words} \T_{\R_S}(K,z) \in \GammaPoly^n}$

\item There is $r \in \GrowR$ s.t. for any $K \in \Nats^n$ and $z \in \Supp \M_S^K$, $\R_S(K,z) \leq r(K)$.

\item There is $l \in \GrowA$ s.t. for any $K \in \Nats^n$, $\Supp \M_S^K \subseteq \WordsLen{l(K)}$.

\end{enumerate}

Abusing notation, we denote the polynomial-time $\MGrow$-scheme $(S,\R_S,\M_S)$ by $S$.

$\R_S^K(z)$ will denote $\R_S(K,z)$. $\UM_S^K \in \mathcal{P}(\Words \times \Words)$ is the joint probability distribution over advice bitstrings and randomness bitstrings, given by 

\[\UM_S^K(y,z):= \M_S^K(z) \delta_{\Abs{y},\R_S^K(z)} 2^{-\R_S^K(z)}\]

$S^K(x,y,z)$ will denote $S(K,x,y,z)$. Given $w=(y,z)$, $S^K(x,w)$ will denote $S(K,x,y,z)$. $S^K(x)$ will denote the $Y$-valued random variable which equals $S(K,x,y,z)$ for $(y,z)$ sampled from $\UM_S^K$. $S_x^K$ will denote the probability distribution of this random variable i.e. $S_x^K$ is the push-forward of $\UM_S^K$ by the mapping $(y,z) \mapsto S(K,x,y,z)$.

We think of $S$ as a randomized algorithm with advice which is random in itself. In particular any polynomial-time $\Gamma$-scheme $S$ can be regarded as a polynomial-time $\MGrow$-scheme with

\[\M_S^K(z):=\delta_{z\A_S^K}\]

We will use the notation $S: X \MScheme Y$ to signify $S$ is a polynomial-time $\MGrow$-scheme of signature $X \rightarrow Y$.

\end{definition}
\end{samepage}

We introduce composition of ${\MGrow}$-schemes as well.

\begin{samepage}
\begin{definition}

Consider encoded sets $X$, $Y$, $Z$ and $S: X \MScheme Y$, $T: Y \MScheme Z$. Choose\\ ${p \in \NatPoly}$ s.t. 

\begin{align*}
\Supp \M_S^K &\subseteq \Bool^{\leq p(K)} \\
\Supp \M_T^K &\subseteq \Bool^{\leq p(K)}
\end{align*}

We can then construct $U: X \Scheme Z$ s.t. for any $K \in \Nats^n$, $a,b \in \Bool^{\leq p(K)}$, ${v \in \Bool^{\R_S(K,a)}}$, ${w \in \Bool^{\R_T(K,b)}}$ and $x \in X$

\begin{align}
\M_U^K &= \En_*^2(\M_S^K \times \M_T^K) \\
\R_U(K, \Chev{a,b}) &= \R_T(K,a)+\R_S(K,b) \\
U^K(x,vw,\Chev{a,b}) &= T^K(S^K(x,w,b),v,a)
\end{align}

Such a $U$ is called the \emph{composition} of $T$ and $S$ and denoted $U = T \circ S$.

\end{definition}
\end{samepage}

\subsubsection{Samplers and Samplability}

The concept of a \emph{samplable} word ensemble is commonly used in average-case complexity theory. Here we introduce a relaxation of this concept which allows approximate sampling with an error compatible with the given fall space. We then proceed to introduce samplable distributional estimation problems.

Samplable word ensembles can be thought of as those ensembles which can be produced by a computationally bounded process. Samplable distributional estimation problems can be thought of as those questions that can be efficiently produced together with their answers, like an exam where the examinee cannot easily find the answer but the examinator knows it (even though the examinator is also computationally bounded).

\begin{samepage}
\begin{definition}

A word ensemble $\Dist$ is called \emph{polynomial-time $\EMG$-samplable} (resp. \emph{polynomial-time $\EG$-samplable}) when there is a polynomial-time $\MGrow$-scheme (resp. polynomial-time $\Gamma$-scheme) $\sigma$ of signature ${\bm{1} \rightarrow \Words}$  s.t. $\Dist^{K} \equiv \sigma_\bullet^K \pmod \Fall$.

In this case, $\sigma$ is called a \emph{polynomial-time $\EMG$-sampler (resp. polynomial-time $\EG$-sampler) of $\Dist$}.

\end{definition}
\end{samepage}

\begin{samepage}
\begin{definition}
\label{def:smp_prob}

A distributional estimation problem $(\Dist,f)$ is called \emph{polynomial-time $\EMG$-samplable (resp. polynomial-time $\EG$-samplable)} when there is a polynomial-time $\MGrow$-scheme (resp. polynomial-time $\Gamma$-scheme) $\sigma$ of signature $\bm{1} \rightarrow \Words \times \Rats$ s.t. 

\begin{enumerate}[(i)]

\item $\sigma_0$ is a polynomial-time $\EMG$-sampler (resp. polynomial-time $\EG$-sampler) of $\Dist$.

\item For any $K \in \Nats^n$, denote $X_{\sigma}^K:=\Supp \sigma_{0\bullet}^K$. For any $x \in \Words$, denote 

$$f_\sigma^K(x):=\begin{cases}\E_{z \sim\UM_\sigma^K}[\sigma^K(z)_1 \mid \sigma^K(z)_0 = x] \text{ if } x \in X_{\sigma}^K \\ 0 \text{ if } x \not\in X_{\sigma}^K \end{cases}$$

We require that the function $\varepsilon(K):=\E_{x \sim \Dist^{K}}[\Abs{f_\sigma^K(x)-f(x)}]$ is in $\Fall$.

\end{enumerate}

This represents the requirement of being able to efficiently generate question-answer pairs, such that the distribution of questions converges to the distribution $\Dist$, and the answers converge to the true output of the function $f$.

When $\sup{\Abs{\sigma_1}} < \infty$ (since $f$ is bounded, this can always be assumed without loss of generality), $\sigma$ is called a \emph{polynomial-time $\EMG$-sampler (resp. polynomial-time $\EG$-sampler) of $(\Dist,f)$}.

\end{definition}
\end{samepage}

For sufficiently large $\GrowA$ the requirements of $\EMG$-samplability become very weak, as seen in the following propositions, which essentially say that if the bitstrings on which $\Dist$ is supported are short enough relative to the length of the advice string, then the randomized advice can just duplicate the distribution. And if there is ample advice available, then the randomized advice can also output an approximation to the true value of $f(x)$ along with $x$.
\begin{samepage}
\begin{proposition}
\label{prp:adv_mgamma_smp}

Consider a word ensemble $\Dist$ s.t. for some $l \in \GrowA$

\begin{equation}
\label{eqn:prp__adv_mgamma_smp}
\Dist^{K}(\Bool^{\leq l(K)}) \equiv 1 \pmod \Fall
\end{equation}

Denote ${I:=\{K \in \Nats^n \mid \Dist^{K}(\Bool^{\leq l(K)}) > 0\}}$. Consider ${\sigma: \bm{1} \MScheme \Words}$ s.t. for any ${K \in I}$

\begin{align*}
\M_\sigma^K&:=\Dist^{K} \mid \Bool^{\leq l(K)} \\
\sigma^K(y,z)&=z
\end{align*}

Then, $\sigma$ is a polynomial-time $\EMG$-sampler of $\Dist$. In particular, since such an $\sigma$ can always be constructed, $\Dist$ is polynomial-time $\EMG$-samplable.

\end{proposition}
\end{samepage}

\begin{proof}

$\chi_I \geq \Dist^{K}(\Bool^{\leq l(K)})$, $1 - \chi_{I} \leq 1 - \Dist^{K}(\Bool^{\leq l(K)})$ and therefore $1 - \chi_I \in \Fall$.

Given $K \in I$, ${\sigma_\bullet^K = \Dist^{K} \mid \Bool^{\leq l(K)}}$ and we get

$$\Dtv(\Dist^{K}, \sigma_\bullet^K) = \Dtv(\Dist^{K}, \Dist^{K} \mid \Bool^{\leq l(K)})$$

$$\Dtv(\Dist^{K}, \sigma_\bullet^K) = \frac{1}{2} \sum_{x \in \Words} \Abs{\Dist^{K}(x)-(\Dist^{K} \mid \WordsLen{\leq l(K)})(x)}$$

Denote $\chi^K:=\chi_{\WordsLen{\leq l(K)}}$.

$$\Dtv(\Dist^{K}, \sigma_\bullet^K) = \frac{1}{2} \sum_{x \in \Words} \Abs{\Dist^{K}(x)-\frac{\chi^K(x)\Dist^{K}(x)}{\Dist^{K}(\WordsLen{\leq l(K)})}}$$

$$\Dtv(\Dist^{K}, \sigma_\bullet^K) = \frac{1}{2} \sum_{x \in \Words} \Dist^{K}(x) \Abs{1-\frac{\chi^K(x)}{\Dist^{K}(\Bool^{\leq l(K)})}}$$

$$\Dtv(\Dist^{K}, \sigma_\bullet^K) = \frac{1}{2} \left(\sum_{x \in \Bool^{\leq l(K)}} \Dist^{K}(x) \Abs{1-\frac{\chi^K(x)}{\Dist^{K}(\Bool^{\leq l(K)})}}+\sum_{x \in \Bool^{>l(K)}} \Dist^{K}(x) \Abs{1-\frac{\chi^K(x)}{\Dist^{K}(\Bool^{\leq l(K)})}}\right)$$

$$\Dtv(\Dist^{K}, \sigma_\bullet^K) = \frac{1}{2} \left(\sum_{x \in \Bool^{\leq l(K)}} \Dist^{K}(x)\left(\frac{1}{\Dist^{K}(\Bool^{\leq l(K)})}-1\right)+\sum_{x \in \Bool^{>l(K)}} \Dist^{K}(x)\right)$$

$$\Dtv(\Dist^{K}, \sigma_\bullet^K) = \frac{1}{2} \left(\Dist^{K}(\Bool^{\leq l(K)})\left(\frac{1}{\Dist^{K}(\Bool^{\leq l(K)})}-1\right)+1 - \Dist^{K}(\Bool^{\leq l(K)})\right)$$

$$\Dtv(\Dist^{K}, \sigma_\bullet^K) = 1-\Dist^{K}(\Bool^{\leq l(K)})$$

Given arbitrary $K \in \Nats^n$,

$$\Dtv(\Dist^{K}, \sigma_\bullet^K) \leq \max(1-\Dist^{K}(\Bool^{\leq l(K)}), 1-\chi_I)$$
\end{proof}

\begin{samepage}
\begin{proposition}
\label{prp:adv_mgamma_gen}

Assume $\Fall$ is $\GrowA$-ample. Consider a distributional estimation problem $(\Dist,f)$ s.t. for some $l \in \GrowA$, \ref{eqn:prp__adv_mgamma_smp} holds. Then, $(\Dist,f)$ is polynomial-time $\EMG$-samplable. 

\end{proposition}
\end{samepage}

\begin{proof}

Consider $\zeta: \NatFun (0,\frac{1}{2}]$ s.t.  $\zeta \in \Fall$ and $\Floor{\log \frac{1}{\zeta}} \in \GrowA$. For any $K \in \Nats^n$ and ${t \in \Reals}$, let ${\rho^K(t) \in \Argmin{s \in \Rats \cap [t-\zeta(K),t+\zeta(K)]} \Abs{\En_\Rats(s)}}$. For any $K \in \Nats^n$, define ${\alpha^K: \Words \rightarrow \Words}$ by 

\[\alpha^K(x):=\Chev{x,c_\Rats(\rho^K(f(x)))}\]

Denote 

\[I:=\{K \in \Nats^n \mid \Dist^{K}(\Bool^{\leq l(K)}) > 0\}\]

Construct ${\sigma: \bm{1} \MScheme \Words \times \Rats}$ s.t. for any $K \in I$

\begin{align*}
\M_\sigma^K:=\alpha_*^K(\Dist^{K} \mid \Bool^{\leq l(K)}) \\
\sigma^K(y,\Chev{z,\En_\Rats(t)})=(z,t)
\end{align*}

By Proposition~\ref{prp:adv_mgamma_smp}, $\sigma_0$ is a polynomial-time $\EMG$-sampler of $\Dist$.

Let ${f_\sigma^K}$ be defined as in Definition~\ref{def:smp_prob}. Consider any $K \in \Nats^n$. It is easy to see that for any ${x \in \Supp \Dist^{K} \cap \Bool^{\leq l(K)}}$, ${f_\sigma^K(x)=\rho^K(f(x))}$ (for $K \not\in I$ this is vacuously true). Also, for any\\ ${x \in \Bool^{>l(K)}}$, $f_\sigma^K(x)=0$. Denote 

\[p^K:=\Dist^{K}(\Bool^{\leq l(K)})\]

We get

$$\E_{\Dist^{K}}[\Abs{f_\sigma^K(x)-f(x)}]=p^K \E_{\Dist^{K}}[\Abs{f_\sigma^K(x)-f(x)} \mid \Abs{x} \leq l(K)] + (1 - p^K)\E_{\Dist^{K}}[\Abs{f_\sigma^K(x)-f(x)} \mid \Abs{x} > l(K)]$$

$$\E_{\Dist^{K}}[\Abs{f_\sigma^K(x)-f(x)}]=p^K \E_{\Dist^{K}}[\Abs{\rho^K(f(x))-f(x)} \mid \Abs{x} \leq l(K)] + (1 - p^K)\E_{\Dist^{K}}[\Abs{f(x)} \mid \Abs{x} > l(K)]$$

$$\E_{\Dist^{K}}[\Abs{f_\sigma^K(x)-f(x)}] \leq p^K \zeta(K) + (1 - p^K)\sup \Abs{f}$$

The right hand side is obviously in $\Fall$.
\end{proof}

We now introduce the notions of samplability over a given \enquote{base space} $Y$.

\begin{definition}

Consider a word ensemble $\Dist$, an encoded set $Y$ and a family of Markov kernels ${\{\pi^K: \Supp \Dist^{K} \Markov Y\}_{K \in \Nats^n}}$. $\Dist$ is called \emph{polynomial-time $\EMG$-samplable (resp. polynomial-time $\EG$-samplable) relative to $\pi$} when there is a polynomial-time $\MGrow$-scheme (resp. polynomial-time $\Gamma$-scheme) $\sigma$ of signature $Y \rightarrow \Words$ s.t. ${\E_{y \sim \pi_*^K\Dist^{K}}[\Dtv(\Dist^K \mid (\pi^K)^{-1}(y),\sigma_y^K)] \in \Fall}$.

In this case, $\sigma$ is called a \emph{polynomial-time $\EMG$-sampler (resp. polynomial-time $\EG$-sampler) of $\Dist$ relative to $\pi$}. That is, even though the underlying distribution may not be samplable, if some evidence ($y$) is given, that permits sampling from the distribution conditional on $y$.

\end{definition}

\begin{samepage}
\begin{definition}
\label{def:smp_prob_rel}

Consider a distributional estimation problem $(\Dist,f)$, an encoded set $Y$ and a family of Markov kernels $\{\pi^K: \Supp \Dist^{K} \Markov Y\}_{K \in \Nats^n}$. $(\Dist,f)$ is called \emph{polynomial-time $\EMG$-samplable (resp. polynomial-time $\EG$-samplable) relative to $\pi$} when there is a polynomial-time $\MGrow$-scheme (resp. polynomial-time $\Gamma$-scheme) $\sigma$ of signature $Y \rightarrow \Words \times \Rats$ s.t.

\begin{enumerate}[(i)]

\item $\sigma_0$ is a polynomial-time $\EMG$-sampler (resp. polynomial-time $\EG$-sampler) of $\Dist$ relative to $\pi$.

\item For any $K \in \Nats^n$, $y \in Y$, Denote $X_{\sigma,y}^K:=\Supp \sigma_{0y}^K$. For any ${x \in \Words}$, denote 

$$f_\sigma^K(x,y):=\begin{cases}\E_{z \sim\UM_\sigma^K}[\sigma^K(y,z)_1 \mid \sigma^K(y,z)_0 = x] \text{ if } x \in X_{\sigma,y}^K \\ 0 \text{ if } x \not\in X_{\sigma,y}^K \end{cases}$$

We require that the function ${\varepsilon(K):=\E_{(x,y) \sim \Dist^{K} \ltimes \pi^K}[\Abs{f_\sigma^K(x,y)-f(x)}]}$ is in $\Fall$.

\end{enumerate}

When $\sup{\Abs{\sigma_1}} < \infty$, $\sigma$ is called a \emph{polynomial-time $\EMG$-sampler (resp. polynomial-time $\EG$-sampler) of $(\Dist,f)$ relative to $\pi$}.

\end{definition}
\end{samepage}

Note that relative samplability reduces to absolute (ordinary) samplability when $Y=\bm{1}$.

The following propositions are basic properties of samplable ensembles and problems which often come in handy. Proposition~\ref{prp:smp} states that the expectation of a function $h(x)$ remains approximately unchanged when $x$ is replaced with a sampler of $\Dist$, and Proposition~\ref{prp:gen} states that the expectation of the product of $h(x)$ and $f(x)$ remains approximately unchanged when $x$ and $f(x)$ are replaced by question/answer pairs produced by a sampler for $(\Dist,f)$.

\begin{samepage}
\begin{proposition}
\label{prp:smp}

Consider a word ensemble $\Dist$, an encoded set $Y$, a family\\ $\{\pi^K: \Supp \Dist^{K} \Markov Y\}_{K \in \Nats^n}$, a set ${I}$ and a uniformly bounded family\\ $\{h_\alpha^K: (\Supp \Dist) \times Y \rightarrow \Reals\}_{\alpha \in I, K \in \Nats^n}$. Suppose $\sigma$ is a polynomial-time $\EMG$-sampler of $\Dist$ relative to $\pi$. Then

\begin{equation}
\label{eqn:prp__smp}
\E_{(x,y) \sim \Dist^{K} \ltimes \pi^K}[h_\alpha^K(x,y)] \overset{\alpha}{\equiv} \E_{(y,z) \sim \pi_*^K\Dist^{K} \times \UM_\sigma^K}[h_\alpha^K(\sigma^K(y,z),y)] \pmod \Fall
\end{equation}

\end{proposition}
\end{samepage}

\begin{proof}

If we sample $(x,y)$ from $\Dist^{K} \ltimes \pi^K$ and then sample $x'$ from ${\Dist^{K} \mid (\pi^K)^{-1}(y)}$, $(x',y)$ will obey the distribution $\Dist^{K} \ltimes \pi^K$. Denote $\Dist_y^K:=\Dist^{K} \mid (\pi^K)^{-1}(y)$. We get

$$\E_{(x,y) \sim \Dist^{K} \ltimes \pi^K}[h_\alpha^K(x,y)] = \E_{(x,y) \sim \Dist^{K} \ltimes \pi^K}[\E_{x' \sim \Dist_y^K}[h_\alpha^K(x',y)]]$$

\begin{align*}
\E_{\Dist^{K} \ltimes \pi^K}[h_\alpha^K(x,y)] &- \E_{\pi_*^K\Dist^{K} \times \UM_\sigma^K}[h_\alpha^K(\sigma^K(y,z),y)] =\\ 
\E_{\Dist^{K} \ltimes \pi^K}[\E_{\Dist_y^K}[h_\alpha^K(x',y)]] &- \E_{\pi_*^K\Dist^{K} \times \UM_\sigma^K}[h_\alpha^K(\sigma^K(y,z),y)]
\end{align*}

$$\E_{\Dist^{K} \ltimes \pi^K}[h_\alpha^K(x,y)] - \E_{\pi_*^K\Dist^{K} \times \UM_\sigma^K}[h_\alpha^K(\sigma^K(y,z),y)] = \E_{\Dist^{K} \ltimes \pi^K}[\E_{\Dist_y^K}[h_\alpha^K(x',y)]-\E_{\UM_\sigma^K}[h_\alpha^K(\sigma^K(y,z),y)]]$$

$$\E_{\Dist^{K} \ltimes \pi^K}[h_\alpha^K(x,y)] - \E_{\pi_*^K\Dist^{K} \times \UM_\sigma^K}[h_\alpha^K(\sigma^K(y,z),y)] = \E_{\Dist^{K} \ltimes \pi^K}[\E_{\Dist_y^K}[h_\alpha^K(x',y)]-\E_{\sigma_y^K}[h_\alpha^K(x',y)]]$$

$$\Abs{\E_{\Dist^{K} \ltimes \pi^K}[h_\alpha^K(x,y)] - \E_{\pi_*^K\Dist^{K} \times \UM_\sigma^K}[h_\alpha^K(\sigma^K(y,z),y)]} \leq \E_{\Dist^{K} \ltimes \pi^K}[\Abs{\E_{\Dist_y^K}[h_\alpha^K(x',y)]-\E_{\sigma_y^K}[h_\alpha^K(x',y)]}]$$

$$\Abs{\E_{\Dist^{K} \ltimes \pi^K}[h_\alpha^K(x,y)] - \E_{\pi_*^K\Dist^{K} \times \UM_\sigma^K}[h_\alpha^K(\sigma^K(y,z),y)]} \leq (\sup h - \inf h) \E_{\Dist^{K} \ltimes \pi^K}[\Dtv(\Dist_y^K,\sigma_y^K)]$$

Using the defining property of $\sigma$, we get the desired result.
\end{proof}

\begin{proposition}
\label{prp:gen}

Consider a distributional estimation problem $(\Dist,f)$, an encoded set $Y$, a family ${\{\pi^K: \Supp \Dist^{K} \Markov Y\}_{K \in \Nats^n}}$, a set ${I}$ and a uniformly bounded family

\[\{h_\alpha^K: (\Supp \Dist) \times Y \rightarrow \Reals\}_{\alpha \in I, K \in \Nats^n}\]

Denote $\Dist_\pi^K:=\Dist^{K} \ltimes \pi^K$. Suppose $\sigma$ is a polynomial-time $\EMG$-sampler of $(\Dist,f)$ relative to $\pi$. Then

\begin{equation}
\E_{\Dist_\pi^K}[h_\alpha^K(x,y)f(x)] \overset{\alpha}{\equiv} \E_{\pi_*^K\Dist^{K} \times \UM_\sigma^K}[h_\alpha^K(\sigma^K(y,z)_0,y)\sigma^K(y,z)_1] \pmod \Fall
\end{equation}

\end{proposition}

\begin{proof}

Let ${f_\sigma^K}$ be defined as in Definition~\ref{def:smp_prob_rel}.

$$\E_{\Dist_\pi^K}[h_\alpha^K(x,y)f(x)]-\E_{\Dist_\pi^K}[h_\alpha^K(x,y)f_\sigma^K(x,y)]=\E_{\Dist_\pi^K}[h_\alpha^K(x,y)(f(x)-f_\sigma^K(x,y))]$$

$$\Abs{\E_{\Dist_\pi^K}[h_\alpha^K(x,y)f(x)]-\E_{\Dist_\pi^K}[h_\alpha^K(x,y)f_\sigma^K(x,y)]} \leq \E_{\Dist_\pi^K}[\Abs{h_\alpha^K(x,y)} \cdot \Abs{f(x)-f_\sigma^K(x,y)}]$$

$$\Abs{\E_{\Dist_\pi^K}[h_\alpha^K(x,y)f(x)]-\E_{\Dist_\pi^K}[h_\alpha^K(x,y)f_\sigma^K(x,y)]} \leq (\sup \Abs{h}) \E_{\Dist_\pi^K}[\Abs{f(x)-f_\sigma^K(x,y)}]$$

By property (ii) of Definition~\ref{def:smp_prob_rel}

$$\E_{\Dist_\pi^K}[h_\alpha^K(x,y)f(x)] \overset{\alpha}{\equiv} \E_{\Dist_\pi^K}[h_\alpha^K(x,y)f_\sigma^K(x,y)] \pmod \Fall$$

Using property (i) of Definition~\ref{def:smp_prob_rel} we can apply Proposition~\ref{prp:smp} to the right hand side and get

$$\E_{\Dist_\pi^K}[h_\alpha^K(x,y)f(x)] \overset{\alpha}{\equiv} \E_{\pi_*^K\Dist^{K} \times \UM_\sigma^K}[h_\alpha^K(\sigma^K(y,z)_0,y) f_\sigma^K(\sigma^K(y,z)_0,y)] \pmod \Fall$$

\begin{align*}
&\E_{\Dist_\pi^K}[h_\alpha^K(x,y)f(x)] \overset{\alpha}{\equiv}\\ &\E_{\pi_*^K\Dist^{K} \times \UM_\sigma^K}[h_\alpha^K(\sigma^K(y,z)_0,y) \E_{z' \sim\UM_\sigma^K}[\sigma^K(y,z')_1 \mid \sigma^K(y,z')_0 = \sigma^K(y,z)_0]] \pmod \Fall
\end{align*}

$$\E_{\Dist_\pi^K}[h_\alpha^K(x,y)f(x)] \overset{\alpha}{\equiv} \E_{\pi_*^K\Dist^{K} \times \UM_\sigma^K}[h_\alpha^K(\sigma^K(y,z)_0,y) \sigma^K(y,z)_1] \pmod \Fall$$
\end{proof}

\subsection{Independent Variables}
\label{subsec:indep_var}

Independent random variables $F_1, F_2$ satisfy 

\begin{equation}
\label{eqn:ev_mult}
\E[F_1 F_2] = \E[F_1] \E[F_2]
\end{equation}

To formulate an analogous property for optimal polynomial-time estimators, we need a notion of independence for distributional decision problems which doesn't make the identity tautologous. Consider distributional decision problems $(\Dist, f_1)$, $(\Dist, f_2)$. Informally, $f_1$ is \enquote{independent} of $f_2$ when learning the value of $f_2(x)$ provides no efficiently accessible information about $f_1(x)$. In the present work, we won't try to formalise this in full generality. Instead, we will construct a specific scenario in which the independence assumption is justifiable.

We start with an informal description. Suppose that $f_1(x)$ depends only on part $\pi(x)$ of the information in $x$ i.e. $f_1(x) = g(\pi(x))$. Suppose further that given $y=\pi(x)$ it is possible to efficiently produce samples $x'$ of $\Dist \mid \pi^{-1}(y)$ for which $f_2(x')$ is known. Then, the knowledge of $f_2(x)$ doesn't provide new information about $g(\pi(x))$ since equivalent information can be efficiently produced without this knowledge, by observing $y$.
 Moreover, if we can only efficiently produce samples $x'$ of $\Dist \mid \pi^{-1}(y)$ together with $\tilde{f}_2(x')$ an \emph{unbiased estimate} of $f_2(x')$, we still expect the analogue of \ref{eqn:ev_mult} to hold since the expected value of $\tilde{f}_2(x') - f_2(x')$ vanishes for any given $x'$ so it is uncorrelated with $f_1(x)$.
 
The following theorem formalises this setting.

\begin{samepage}
\begin{theorem}
\label{thm:mult}

Consider $\Dist$ a word ensemble, $f_1, f_2: \Supp \Dist \rightarrow \Reals$ bounded, $(\mathcal{E},g)$ a distributional estimation problem and $\pi: \Words \Scheme \Words$. Assume the following conditions:

\begin{enumerate}[(i)]

\item\label{con:thm__mult__dist} $\pi_*^{K}(\Dist^{K}) \equiv \mathcal{E}^{K} \pmod \Fall$

\item\label{con:thm__mult__fun} Denote ${\bar{g}: \Words \rightarrow \Reals}$ the extension of $g$ by $0$.  We require

\[\E_{(x,z) \sim \Dist^{K} \times \Un_\pi^{K}}[\Abs{f_1(x)-\bar{g}(\pi^{K}(x,z))}] \in \Fall\]

\item\label{con:thm__mult__smp} $(\Dist, f_2)$ is polynomial-time $\EMG$-samplable relative to $\pi$.

\end{enumerate}

Suppose $P_1$ is an $\ESG$-optimal estimator for $(g,\mathcal{E})$ and $P_2$ is an $\ESG$-optimal estimator for $(\Dist,f_2)$. Denote $P_\pi := P_1 \circ \pi$. Construct ${P: \Words \Scheme \Rats}$ s.t. $\R_P=\R_{P_\pi}+\R_{P_2}$ and for any ${x \in \Words}$, ${z_1 \in \WordsLen{\R_{P_\pi}(K)}}$ and $z_2 \in \WordsLen{\R_{P_2}(K)}$

\begin{equation}
P^{K}(x,z_1 z_2)=P_\pi^{K}(x,z_1) P_2^{K}(x,z_2)
\end{equation}

Then, $P$ is an $\ESG$-optimal estimator for $(\Dist,f_1 f_2)$.

\end{theorem}
\end{samepage}

In order to prove Theorem~\ref{thm:mult} we will need the following proposition, which takes the defining inexploitability property of an $\ESG$-optimal estimator, and extends it to the adversary $S$ having access to randomized advice, which can be done because deterministic advice can copy the "luckiest possible advice string" drawn from the distribution over advice strings.

\begin{samepage}
\begin{proposition}
\label{prp:mixed_ort}

Consider $(\Dist,f)$ a distributional estimation problem, $P$ an $\ESG$-optimal estimator for $(\Dist,f)$ and $S: \Words \times \Rats \MScheme \Rats$ bounded. Then

\begin{equation}
\label{eqn:prp__mixed_ort}
\E_{\Dist^{K} \times \Un_P^{K} \times \UM_S^{K}}[(P^{K}(x,y) - f(x))S^{K}(x,P^{K}(x,y),z,w)] \equiv 0 \pmod \Fall
\end{equation}

\end{proposition}
\end{samepage}

\begin{proof}

For any $K \in \Nats^n$, choose 

$$w^{K} \in \Argmax{w \in \Supp \M_S^{K}} \Abs{\E_{\Dist^{K} \times \Un_P^{K} \times \Un^{\R_S^{K}(w)}}[(P^{K}(x,y) - f(x))S^{K}(x,P^{K}(x,y),z,w)]}$$

Construct $\bar{S}: \Words \times \Rats \Scheme \Rats$ s.t. 

\begin{align*}
\R_{\bar{S}}(K)&=\R_S^{K}(w^{K}) \\
\bar{S}^{K}(x,t,z)&=S^{K}(x,t,z,w^{K})
\end{align*}

$P$ is an $\ESG$-optimal estimator for $(\Dist,f)$, therefore

$$\E_{\Dist^{K} \times \Un_P^{K} \times \Un_{\bar{S}}^{K}}[(P^{K}(x,y) - f(x))\bar{S}^{K}(x,P^{K}(x,y),z)] \equiv 0 \pmod \Fall$$

$$\E_{\Dist^{K} \times \Un_P^{K} \times \Un^{\R_S^{K}(w)}}[(P^{K}(x,y) - f(x))S^{K}(x,P^{K}(x,y),z,w^{K})] \equiv 0 \pmod \Fall$$

By construction of $w^{K}$, the absolute value of the left hand side is no less than the absolute value of the left hand side of \ref{eqn:prp__mixed_ort}.
\end{proof}

\begin{proof}[Proof of Theorem \ref{thm:mult}]

Consider $K \in \Nats^n$, $x \in \Supp \Dist^{K}$, $z_1 \in \WordsLen{\R_{P_1}(K)}$, ${z_2 \in \WordsLen{\R_{P_2}(K)}}$ and\\ ${z_3 \in \WordsLen{\R_\pi(K)}}$.

\[P^{K}(x,z_1 z_3 z_2)-f_1(x)f_2(x)=P_\pi^{K}(x, z_1 z_3) P_2^{K}(x,z_2) - f_1(x) f_2(x)\]

Adding and subtracting $P_\pi^{K}(x, z_1 z_3) f_2(x)$ from the right hand side and grouping variables, we get

\[P^{K}(x,z_1 z_3 z_2)-f_1(x)f_2(x)=P_\pi^{K}(x, z_1 z_3)(P_2^{K}(x,z_2)-f_2(x))+(P_\pi^{K}(x, z_1 z_3)-f_1(x))f_2(x)\]

For any bounded $S: \Words \times \Rats \Scheme \Rats$ we get

$$\Abs{\E[(P^{K}-f_1 f_2)S^{K}]} \leq \Abs{\E[(P_2^{K}-f_2) P_\pi^{K} S^{K}]} + \Abs{\E[(P_\pi^{K}-f_1)f_2 S^{K}]}$$

$P_2$ is an $\ESG$-optimal estimator for $(\Dist,f_2)$ therefore the first term on the right hand side is in $\Fall$.

$$\Abs{\E[(P^{K}-f_1 f_2)S^{K}]} \leq \Abs{\E[(P_\pi^{K}-f_1)f_2 S^{K}]} \pmod \Fall$$

$$\Abs{\E[(P^{K}-f_1 f_2)S^{K}]} \leq \Abs{\E[(P_\pi^{K}-f_1)f_2 S^{K}] - \E[(P_\pi^{K}-\bar{g} \circ \pi^{K})f_2 S^{K}] + \E[(P_\pi^{K}-\bar{g} \circ \pi^{K})f_2 S^{K}]} \pmod \Fall$$

$$\Abs{\E[(P^{K}-f_1 f_2)S^{K}]} \leq \Abs{\E[(P_\pi^{K}-f_1)f_2 S^{K}] - \E[(P_\pi^{K}-\bar{g} \circ \pi^{K})f_2 S^{K}]} + \Abs{\E[(P_\pi^{K}-\bar{g} \circ \pi^{K})f_2 S^{K}]} \pmod \Fall$$

$$\Abs{\E[(P^{K}-f_1 f_2)S^{K}]} \leq \Abs{\E[(\bar{g} \circ \pi^{K}-f_1)f_2 S^{K}]} + \Abs{\E[(P_\pi^{K}-\bar{g} \circ \pi^{K})f_2 S^{K}]} \pmod \Fall$$

$$\Abs{\E[(P^{K}-f_1 f_2)S^{K}]} \leq (\sup \Abs{f_2}) (\sup \Abs{S}) \E[\Abs{\bar{g} \circ \pi^{K} - f_1}] + \Abs{\E[(P_\pi^{K}-\bar{g} \circ \pi^{K})f_2 S^{K}]} \pmod \Fall$$

Condition~\ref{con:thm__mult__fun} implies the first term on the right hand side is in $\Fall$.

$$\Abs{\E[(P^{K}-f_1 f_2)S^{K}]} \leq \Abs{\E[(P_\pi^{K}-\bar{g} \circ \pi^{K})f_2 S^{K}]} \pmod \Fall$$

Denote $\Un_{\text{tot}}^{K}:= \Un_{P_1}^{K} \times \Un_{P_2}^{K} \times \Un_S^{K}$. We change variables inside the expected value on the right hand side by $y:=\pi^{K}(x,z_3)$. Observing that $(x,y)$ obeys the distribution $\Dist^{K} \ltimes \pi^{K}$ we get

$$\Abs{\E[(P^{K}-f_1 f_2)S^{K}]} \leq \Abs{\E_{\Dist^{K} \ltimes \pi^{K} \times \Un_{\text{tot}}^{K}}[(P_1^{K}(y,z_1)-\bar{g}(y))f_2(x) S^{K}(x,P_1^{K}(y,z_1)P_2^{K}(x,z_2), z_4)]} \pmod \Fall$$

$$\Abs{\E[(P^{K}-f_1 f_2)S^{K}]} \leq \Abs{\E_{\Dist^{K} \ltimes \pi^{K}}[\E_{\Un_{\text{tot}}^{K}}[(P_1^{K}(y,z_1)-\bar{g}(y))S^{K}(x,P_1^{K}(y,z_1)P_2^{K}(x,z_2), z_4)]f_2(x)]} \pmod \Fall$$

Let $\sigma$ be a polynomial-time $\EMG$-sampler of $(\Dist,f_2)$ relative to $\pi$. Applying Proposition~\ref{prp:gen} to the right hand side we get

$$\Abs{\E[(P^{K}-f_1 f_2)S^{K}]} \leq \Abs{\E_{\pi_*^K\Dist^{K} \times \UM_\sigma^{K}}[\E[(P_1^{K}(y)-\bar{g}(y))S^{K}(\sigma^{K}(y)_0,P_1^{K}(y)P_2^{K}(\sigma^{K}(y)_0))]\sigma^{K}(y)_1]} \pmod \Fall$$

Using condition~\ref{con:thm__mult__dist} we conclude that

$$\Abs{\E[(P^{K}-f_1 f_2)S^{K}]} \leq \Abs{\E_{\mathcal{E}^k \times \UM_\sigma^{K}}[\E[(P_1^{K}(y)-g(y))S^{K}(\sigma^{K}(y)_0,P_1^{K}(y)P_2^{K}(\sigma^{K}(y)_0))]\sigma^{K}(y)_1]} \pmod \Fall$$

$$\Abs{\E[(P^{K}-f_1 f_2)S^{K}]} \leq \Abs{\E_{\mathcal{E}^k \times \Un_{\text{tot}}^{K} \times \UM_\sigma^{K}}[(P_1^{K}(y)-g(y))S^{K}(\sigma^{K}(y)_0,P_1^{K}(y)P_2^{K}(\sigma^{K}(y)_0))\sigma^{K}(y)_1]} \pmod \Fall$$

By Proposition~\ref{prp:mixed_ort}, this implies

$$\Abs{\E[(P^{K}-f_1 f_2)S^{K}]} \equiv 0 \pmod \Fall$$
\end{proof}

The following corollary demonstrates one natural scenario in which the conditions of Theorem~\ref{thm:mult} hold. The scenario is one where the distribution is $\Dist_{1}\times\Dist_{2}$, and the task is to estimate $f_{1}(x_{1})f_{2}(x_{2})$. By Theorem~\ref{thm:mult}, this can be done if there is a sampler for $(\Dist_{2},f_{2})$, and a sampler for $\Dist_{1}$.

\begin{samepage}
\begin{corollary}
\label{crl:dir_prod}

Consider $(\Dist_1,f_1)$, $(\Dist_2,f_2)$ distributional estimation problems. Suppose $P_1$ is an $\ESG$-optimal estimator for $(\Dist_1,f_1)$, $P_2$ is an $\ESG$-optimal estimator for $(\Dist_2,f_2)$, $\sigma_1$ is a polynomial-time $\EMG$-sampler for $\Dist_1$ and $\sigma_2$ is a polynomial-time $\EMG$-sampler for $(\Dist_2,f_2)$. Define ${\Dist^{K}:=\En_*^2(\Dist_1^k \times \Dist_2^k)}$. Define ${f: \Supp \Dist \rightarrow \Reals}$ by ${f(\Chev{x_1,x_2}):=f_1(x_1)f_2(x_2)}$. Then, there is $P$, an $\ESG$-optimal estimator for $(\Dist,f)$, s.t. $\R_P=\R_{P_1}+\R_{P_2}$ and for any $K \in \Nats^n$, $x_1 \in \Supp \sigma_{1\bullet}^{K}$, $x_2 \in \Words$, $z_1 \in \WordsLen{\R_{P_1}(K)}$ and $z_2 \in \WordsLen{\R_{P_2}(K)}$

\begin{equation}
P^{K}(\Chev{x_1,x_2}, z_1 z_2)=P_1^{K}(x_1,z_1) P_2^{K}(x_2,z_2)
\end{equation}

\end{corollary}
\end{samepage}

In order to prove Corollary~\ref{crl:dir_prod}, we'll need to prove several minor propositions first.

\begin{samepage}
\begin{proposition}
\label{prp:thm__mult__cond1}

Consider $\Dist_1$, $\Dist_2$ word ensembles and $\sigma_1$, $\sigma_2$ which are polynomial-time $\EMG$-samplers for $\Dist_1$ and $\Dist_2$ respectively. Define ${\Dist^k:=\En_*^2(\Dist_1^k \times \Dist_2^k)}$. Suppose $\pi: \Words \Scheme \Words$ is s.t. for any $K \in \Nats^n$, $x_1 \in \Supp \sigma_{1\bullet}^{K}$, ${x_2 \in \Supp \sigma_{2\bullet}^{K}}$ and $z \in \Bool^{\R_\pi(K)}$, $\pi^{K}(\Chev{x_1,x_2},z)=x_1$. Then $\pi_*^K\Dist^{K} \equiv \Dist_1^{K} \pmod \Fall$

\end{proposition}
\end{samepage}

\begin{proof}

$\sigma_{1\bullet}^{K} \equiv \Dist_1^{K} \pmod \Fall$ and $\sigma_{2\bullet}^{K} \equiv \Dist_2^{K} \pmod \Fall$.  By Proposition~\ref{prp:prob_cong_dir},

\[\sigma_{1\bullet}^{K} \times \sigma_{2\bullet}^{K} \equiv \Dist_1^{K} \times \Dist_2^{K} \pmod \Fall\]

Denote $\Dist_\sigma^{K}:=\En_*^2(\sigma_{1\bullet}^{K} \times \sigma_{2\bullet}^{K})$. We get ${\Dist_\sigma^{K} \equiv \Dist^{K} \pmod \Fall}$ and therefore ${\pi_*^K\Dist_\sigma^{K} \equiv \pi_*^K\Dist^{K} \pmod \Fall}$ (by Proposition~\ref{prp:prob_cong_push}). Obviously $\pi_*^K\Dist_\sigma^{K}=\sigma_{1\bullet}^{K}$. We conclude that ${\pi_*^K\Dist^{K} \equiv \sigma_{1\bullet}^{K} \pmod \Fall}$ and therefore ${\pi_*^K\Dist^{K} \equiv \Dist_1 \pmod \Fall}$ (by Proposition~\ref{prp:prob_cong_eq}).
\end{proof}

\begin{samepage}
\begin{proposition}
\label{prp:thm__mult__cond2}

Consider $\Dist_1$, $\Dist_2$ word ensembles and $\sigma_1$, $\sigma_2$ which are polynomial-time $\EMG$-samplers for $\Dist_1$ and $\Dist_2$ respectively. Suppose $\pi: \Words \Scheme \Words$ is s.t. for any $K \in \Nats^n$,\\ $x_1 \in \Supp \sigma_{1\bullet}^{K}$, ${x_2 \in \Supp \sigma_{2\bullet}^{K}}$ and $z \in \Bool^{\R_\pi(K)}$, $\pi^{K}(\Chev{x_1,x_2},z)=x_1$. Then, for any\\ $g: \Supp \Dist_1 \rightarrow \Reals$ bounded and ${\bar{g}: \Words \rightarrow \Reals}$ its extension by ${0}$, we have 

$$\E_{(x_1,x_2,z) \sim\Dist_1^{K} \times \Dist_2^{K} \times \Un_\pi^{K}}[\Abs{g(x_1)-\bar{g}(\pi^{K}(\Chev{x_1,x_2},z))}] \in \Fall$$

\end{proposition}
\end{samepage}

\begin{proof}

Denote $M:= \sup g - \inf g$.

$$\E[\Abs{g(x_1)-\bar{g}(\pi^{K}(\Chev{x_1,x_2}))}] \leq M \Prb_{ \Dist_1^{K} \times \Dist_2^{K}}[(x_1,x_2) \not\in \Supp \sigma_{1\bullet}^{K} \times \Supp \sigma_{2\bullet}^{K}]$$

$$\E[\Abs{g(x_1)-\bar{g}(\pi^{K}(\Chev{x_1,x_2}))}] \leq M\Prb_{ \sigma_{1\bullet}^{K} \times \sigma_{2\bullet}^{K}}[(x_1,x_2) \not\in \Supp \sigma_{1\bullet}^{K} \times \Supp \sigma_{2\bullet}^{K}] \pmod \Fall$$

$$\E[\Abs{g(x_1)-\bar{g}(\pi^{K}(\Chev{x_1,x_2}))}] \equiv 0 \pmod \Fall$$
\end{proof}
\begin{samepage}
\begin{proposition}
\label{prp:smp_base_change}

Consider word ensembles $\Dist_1$ and $\Dist_2$ with polynomial-time $\EMG$-samplers $\sigma_1$ and $\sigma_2$ respectively. Define ${\Dist^k:=\En_*^2(\Dist_1^k \times \Dist_2^k)}$. Suppose ${\pi: \Words \Scheme \Words}$ is s.t. for any $K \in \Nats^n$, ${x_1 \in \Supp \sigma_{1\bullet}^{K}}$, ${x_2 \in \Words}$ and $z \in \Bool^{\R_\pi(K)}$, ${\pi^{K}(\Chev{x_1,x_2},z)=x_1}$ and, conversely, if ${x \in \Words}$ is s.t. ${\pi^K(x,z)=x_1}$ then ${x}$ is of the form ${\Chev{x_1,x_2'}}$ for some ${x_2' \in \Words}$. Consider ${\sigma: \Words \MScheme \Words}$ s.t. $\UM_\sigma^{K}=\UM_{\sigma_2}^{K}$ and for any $x \in \Supp \sigma_{1\bullet}^{K}$, ${\sigma^{K}(x,z,w)=\Chev{x,\sigma_2^{K}(z,w)}}$. Then, $\sigma$ is a polynomial-time $\EMG$-sampler of $\Dist$ relative to $\pi$. In particular, since such an $\sigma$ can always be constructed, $\Dist$ is polynomial-time $\EMG$-samplable relative to $\pi$.

\end{proposition}
\end{samepage}

\begin{proof}

\[\Dist^{K} \equiv \En_*^2(\sigma_{1\bullet}^K \times \sigma_{2\bullet}^K)\pmod \Fall\]

\[\pi_*^K\Dist^{K} \equiv \pi_*^K\En_*^2(\sigma_{1\bullet}^K \times \sigma_{2\bullet}^K) \pmod \Fall\] 

\[\pi_*^K\Dist^{K} \equiv  \sigma_{1\bullet}^K \pmod \Fall\]

Denote $\Dist_x^K:=\Dist \mid (\pi^K)^{-1}(x)$.

\[\E_{x \sim \pi_*^K\Dist^{K}}[\Dtv(\Dist_x^K,\sigma_x^K)] \equiv \E_{x \sim \sigma_{1\bullet}^K}[\Dtv(\Dist_x^K,\sigma_x^K)] \pmod \Fall\]

For any $x \in \Supp \sigma_{1\bullet}^{K}$, $\Dist_x^K = \En_*^2(\delta_x \times \Dist_2^{K})$ and $\sigma_x^K=\En_*^2(\delta_x \times \sigma_{2\bullet}^K)$.

\[\E_{x \sim \pi_*^K\Dist^{K}}[\Dtv(\Dist_x^K,\sigma_x^K)] \equiv \E_{x \sim \sigma_{1\bullet}^K}[\Dtv(\En_*^2(\delta_x \times \Dist_2^{K}),\En_*^2(\delta_x \times \sigma_{2\bullet}^K))] \pmod \Fall\]

\[\E_{x \sim \pi_*^K\Dist^{K}}[\Dtv(\Dist_x^K,\sigma_x^K)] \equiv \E_{x \sim \sigma_{1\bullet}^K}[\Dtv(\Dist_2^{K},\sigma_{2\bullet}^K)] \pmod \Fall\]

\[\E_{x \sim \pi_*^K\Dist^{K}}[\Dtv(\Dist_x^K,\sigma_x^K)] \equiv \Dtv(\Dist_2^{K},\sigma_{2\bullet}^K) \pmod \Fall\]

\[\E_{x \sim \pi_*^K\Dist^{K}}[\Dtv(\Dist_x^K,\sigma_x^K)] \equiv 0 \pmod \Fall\]
\end{proof}

\begin{samepage}
\begin{proposition}
\label{prp:thm__mult__cond3}

Consider $\Dist_1$ a word ensemble with polynomial-time $\EMG$-sampler $\sigma$ and $(\Dist_2, f)$ a distributional estimation problem with polynomial-time $\EMG$-sampler $\tau$. Define the distributional estimation problem $(\Dist,\bar{f})$ by 

\begin{align*}
\Dist^k:=\En_*^2(\Dist_1^k \times \Dist_2^k)\\
\bar{f}(\Chev{x_1,x_2})=f(x_2)
\end{align*}

Suppose $\pi: \Words \Scheme \Words$ is s.t. for any $K \in \Nats^n$, ${x_1 \in \Supp \sigma_{\bullet}^{K}}$, ${x_2 \in \Words}$ and\\ $z \in \Bool^{\R_\pi(K)}$, $\pi^{K}(\Chev{x_1,x_2},z)=x_1$ and, conversely, if ${x \in \Words}$ is s.t. ${\pi^K(x,z)=x_1}$ then ${x}$ is of the form ${\Chev{x_1,x_2'}}$ for some ${x_2' \in \Words}$. Then, $(\Dist,\bar{f})$ is polynomial-time $\EMG$-samplable relative to $\pi$.

\end{proposition}
\end{samepage}

\begin{proof}

Construct $\bar{\tau}: \Words \MScheme \Words \times \Rats$ s.t. $\UM_{\bar{\tau}}^K=\UM_{\tau}^K$ and for any ${x \in \Supp \sigma_{\bullet}^K}$

\[\bar{\tau}^K(x,y,z)=(\Chev{x,\tau^K(y,z,w)_0},\tau^K(y,z,w)_1)\] 

By Proposition~\ref{prp:smp_base_change}, $\bar{\tau}_0$ is a polynomial-time $\EMG$-sampler of $\Dist$ relative to $\pi$.

\[\Dist^{K} \equiv \En_*^2(\sigma_\bullet^K \times \tau_{0\bullet}^K)\pmod \Fall\]

\[\Dist^{K} \ltimes \pi^K \equiv \En_*^2(\sigma_\bullet^K \times \tau_{0\bullet}^K) \ltimes \pi^K \pmod \Fall\]

Let ${f_\tau^K}$ and ${f_{\bar{\tau}}^K}$ be defined as in Definition~\ref{def:smp_prob_rel}.

\[\E_{(x,y) \sim \Dist^{K} \ltimes \pi^K}[\Abs{f_{\bar{\tau}}^K(x,y)-\bar{f}(x)}] \equiv \E_{(x,y) \sim \En_*^2(\sigma_\bullet^K \times \tau_{0\bullet}^K) \ltimes \pi^K}[\Abs{f_{\bar{\tau}}^K(x,y)-\bar{f}(x)}] \pmod \Fall\]

\[\E_{ \Dist^{K} \ltimes \pi^K}[\Abs{f_{\bar{\tau}}^K-\bar{f}}] \equiv \E_{(x_1,x_2) \sim \sigma_\bullet^K \times \tau_{0\bullet}^K}[\Abs{f_{\bar{\tau}}^K(\Chev{x_1,x_2},x_1)-\bar{f}(\Chev{x_1,x_2})}] \pmod \Fall\]

\[\E_{ \Dist^{K} \ltimes \pi^K}[\Abs{f_{\bar{\tau}}^K-\bar{f}}] \equiv \E_{(x_1,x_2) \sim \sigma_\bullet^K \times \tau_{0\bullet}^K}[\Abs{\E_{\UM_{\bar{\tau}}^K}[\bar{\tau}_1^K(x_1) \mid \bar{\tau}^K(x_1)_0 = \Chev{x_1,x_2}]-f(x_2)}] \pmod \Fall\]

\[\E_{\Dist^{K} \ltimes \pi^K}[\Abs{f_{\bar{\tau}}^K-\bar{f}}] \equiv \E_{(x_1,x_2) \sim \sigma_\bullet^K \times \tau_{0\bullet}^K}[\Abs{\E_{\UM_\tau^K}[\tau_1^K \mid \Chev{x_1,\tau_0^K} = \Chev{x_1,x_2}]-f(x_2)}] \pmod \Fall\]

\[\E_{\Dist^{K} \ltimes \pi^K}[\Abs{f_{\bar{\tau}}^K-\bar{f}}] \equiv \E_{(x_1,x_2) \sim \sigma_\bullet^K \times \tau_{0\bullet}^K}[\Abs{\E_{\UM_\tau^K}[\tau_1^K \mid \tau_0^K = x_2]-f(x_2)}] \pmod \Fall\]

\[\E_{\Dist^{K} \ltimes \pi^K}[\Abs{f_{\bar{\tau}}^K-\bar{f}}] \equiv \E_{x_2 \sim \tau_{0\bullet}^K}[\Abs{f_\tau^K(x_2)-f(x_2)}] \pmod \Fall\]

\[\E_{\Dist^{K} \ltimes \pi^K}[\Abs{f_{\bar{\tau}}^K-\bar{f}}] \equiv \E_{x_2 \sim \Dist_2^{K}}[\Abs{f_\tau^K(x_2)-f(x_2)}] \pmod \Fall\]

\[\E_{\Dist^{K} \ltimes \pi^K}[\Abs{f_{\bar{\tau}}^K-\bar{f}}] \equiv 0 \pmod \Fall\]
\end{proof}

\begin{samepage}
\begin{proposition}
\label{prp:thm__mult__op}

Consider word ensemble $\Dist_1$  with polynomial-time $\EMG$-sampler $\sigma$ and $(\Dist_2,f)$ a distributional estimation problem. Define the distributional estimation problem $(\Dist,\bar{f})$ by 

\begin{align*}
\Dist^k:=\En_*^2(\Dist_1^k \times \Dist_2^k)\\
\bar{f}(\Chev{x_1,x_2})=f(x_2)
\end{align*}

Suppose $P$ is an $\ESG$-optimal estimator for $(\Dist_2,f)$. Let ${\bar{P}: \Words \Scheme \Rats}$ be s.t. $\R_{\bar{P}}=\R_P$ and for any $K \in \Nats^n$, $x_1 \in \Supp \sigma_\bullet^K$, $x_2 \in \Supp \Dist_2^K$ and $z \in \Bool^{\R_P(K)}$, $\bar{P}^K(\Chev{x_1,x_2},z)=P^K(x_2,z)$. Then, $\bar{P}$ is an $\ESG$-optimal estimator for $(\Dist,\bar{f})$.

\end{proposition}
\end{samepage}

\begin{proof}

Consider any $S: \Words \times \Rats \Scheme \Rats$ bounded. Denote $\Un_{PS}^K:=\Un_P^K \times \Un_S^K$, ${\Dist_{PS}^K:=\Dist^{K} \times \Un_{PS}^K}$.

\begin{align*}
&\E_{\Dist_{PS}^K}[(\bar{P}^{K}(x) - \bar{f}(x))S^{K}(x,\bar{P}^{K}(x))]=\\&\E_{\Dist_1^{K} \times \Dist_2^{K} \times \Un_{PS}^K}[(\bar{P}^K(\Chev{x_1,x_2}) - \bar{f}(\Chev{x_1,x_2}))S^K(\Chev{x_1,x_2},\bar{P}^K(\Chev{x_1,x_2}))]
\end{align*}

\[\E_{\Dist_{PS}^K}[(\bar{P}^{K}(x) - \bar{f}(x))S^{K}(x,\bar{P}^{K}(x))]=\E_{\Dist_1^{K} \times \Dist_2^{K} \times \Un_{PS}^K}[(\bar{P}^K(\Chev{x_1,x_2}) - f(x_2))S^K(\Chev{x_1,x_2},\bar{P}^K(\Chev{x_1,x_2}))]\]

\[\E_{\Dist_{PS}^K}[(\bar{P}^{K}(x) - \bar{f}(x))S^{K}(x,\bar{P}^{K}(x))]=\E_{\Dist_1^{K}}[\E_{\Dist_2^{K} \times \Un_{PS}^K}[(\bar{P}^K(\Chev{x_1,x_2}) - f(x_2))S^K(\Chev{x_1,x_2},\bar{P}^K(\Chev{x_1,x_2}))]]\]

Applying Proposition~\ref{prp:smp} (with $Y = \bm{1}$) to the right hand side, we get

\[\E_{\Dist_{PS}^K}[(\bar{P}^{K} - \bar{f})S^{K}] \equiv \E_{\UM_\sigma^K}[\E_{\Dist_2^{K} \times \Un_{PS}^K}[(\bar{P}^K(\Chev{\sigma^K,x_2}) - f(x_2))S^K(\Chev{\sigma^K,x_2},\bar{P}^K(\Chev{\sigma^K,x_2}))]] \pmod \Fall\]

\[\E_{\Dist_{PS}^K}[(\bar{P}^{K} - \bar{f})S^{K}] \equiv \E_{\UM_\sigma^K}[\E_{\Dist_2^{K} \times \Un_{PS}^K}[(P^K(x_2) - f(x_2))S^K(\Chev{\sigma^K,x_2},P^K(x_2))]] \pmod \Fall\]

\[\E_{\Dist_{PS}^K}[(\bar{P}^{K} - \bar{f})S^{K}] \equiv \E_{\Dist_2^{K} \times \Un_{PS}^K \times \UM_\sigma^K}[(P^K(x_2) - f(x_2))S^K(\Chev{\sigma^K,x_2},P^K(x_2))] \pmod \Fall\]

Using the fact that $P$ is an $\ESG$-optimal estimator for $(\Dist_2,f)$, we conclude

\[\E_{\Dist_{PS}^K}[(\bar{P}^{K} - \bar{f})S^{K}] \equiv 0 \pmod \Fall\]
\end{proof}

\begin{proof}[Proof of Corollary \ref{crl:dir_prod}]

Define $\bar{f}_1, \bar{f}_2: \Supp \Dist \rightarrow \Reals$ by $\bar{f}_1(\Chev{x_1,x_2})=f_1(x_1)$, $\bar{f}_2(\Chev{x_1,x_2})=f_2(x_2)$. 

Construct $\pi: \Words \Scheme \Words$ s.t. $\R_\pi \equiv 0$, for any $K \in \Nats^n$, ${x_1 \in \Supp \sigma_{1\bullet}^{K}}$ and ${x_2 \in \Words}$, ${\pi^{K}(\Chev{x_1,x_2})=x_1}$ and, conversely, if ${x \in \Words}$ is s.t. ${\pi^K(x)=x_1}$ then ${x}$ is of the form ${\Chev{x_1,x_2'}}$ for some ${x_2' \in \Words}$. This is possible because the runtime of $\sigma_1^K$ is bounded by a polynomial in $K$ so the length of $\sigma_{1}^K$'s output is also bounded by a polynomial in $K$, implying $\pi^K$ only has to read a polynomial size prefix of its input in order to output $x_1$. On the other hand, if the input is not of the form ${\Chev{x_1,x_2}}$ for ${x_1}$ sufficiently short to be in ${\Supp \sigma_{1\bullet}^{K}}$, ${\pi}$ may output a string too long to be in ${\Supp \sigma_{1\bullet}^{K}}$.

Construct $\bar{P}: \Words \Scheme \Rats$ s.t. $\R_{\bar{P}}=\R_{P_2}$ and for any $x_1 \in \Supp \sigma_{1\bullet}^K$, $x_2 \in \Words$ and\\ $z \in \Bool^{\R_{P_2}(K)}$, $\bar{P}^K(\Chev{x_1,x_2},z)=P_2^K(x_2,z)$. This is possible for the same reason as above: $\bar{P}$ skips the polynomial size prefix corresponding to $x_1$ and then executes a simulation of running $P_2$ on $x_2$, even if $x_2$ is too long to read in full. By Proposition~\ref{prp:thm__mult__op}, $\bar{P}$ is an $\ESG$-optimal estimator for $(\Dist,\bar{f}_2)$. 

We apply Theorem~\ref{thm:mult} where $\bar{f}_1$, $\bar{f}_2$ play the roles of $f_1$, $f_2$ and $(\Dist_1, f_1)$ plays the role of $(\mathcal{E},g)$: condition~\ref{con:thm__mult__dist} holds due to Proposition~\ref{prp:thm__mult__cond1}, condition~\ref{con:thm__mult__fun} holds due to Proposition~\ref{prp:thm__mult__cond2} and condition~\ref{con:thm__mult__smp} holds due to Proposition~\ref{prp:thm__mult__cond3}. This gives us $P$, an optimal polynomial-time estimator for $(\Dist, f)$ s.t. ${\R_P=\R_{P_1}+\R_{P_2}}$ and for any ${z_1 \in \Bool^{\R_{P_1}(K)}}$ and $z_2 \in \Bool^{\R_{P_1}(K)}$ 

\[P^K(x, z_1 z_2) = P_1^K(\pi^K(x),z_1) \bar{P}^K(x,z_2)\]

In particular, for any ${x_1 \in \Supp \sigma_{1\bullet}^K}$ and $x_2 \in \Words$

\[P^K(\Chev{x_1,x_2}, z_1 z_2)=P_1^K(x_1,z_2)P_2^K(x_2,z_2)\]
\end{proof}

\section{Reductions and Completeness}
\label{sec:reductions}

In this section we study notions of Karp reduction between distributional estimation problems such that the pull-back of an optimal polynomial-time estimator is an optimal polynomial-time estimator. It is also interesting to study Cook reductions but we avoid it in the present work.

First, we demonstrate that the notion of Karp reduction used in average-case complexity theory is insufficiently strong for our purpose.

Consider the setting of Corollary~\ref{crl:one_way}. Denote $\Dist^k:=\Un^{2k}$ and define ${\chi: \Supp \Dist \rightarrow \Bool}$ s.t. for any $x,y \in \Bool^k$, $\chi(xy)$ = $x \cdot y$. Construct $\pi_f: \Words \Scheme \Words$ s.t. for any\\ $x,y \in \Bool^k$, ${\pi_f^k(xy) = \Chev{f(x),y}}$. $\pi_f$ can be regarded as a Karp reduction of $(\Dist, \chi)$ to $(\Dist_{(f)},\chi_f)$ since for any ${z \in \Supp \Dist^k}$ we have $\chi_f(\pi_f^k(z))=\chi(z)$ and $(\pi_f)_*\Dist=\Dist_{(f)}$\footnote{This is a much stronger condition than what is needed for a reduction to preserve average-case complexity. See \cite{Bogdanov_2006} for details.}. However, the pullback of $P$ is \emph{not} an $\Fall_{\text{neg}}(\Gamma)$-optimal estimator for $(\Dist,\chi)$ since its error is $\E_{z \sim \Dist^k}[(\frac{1}{2}-\chi(z))^2]=\frac{1}{4}$ whereas we can construct $Q: \Words \Scheme \Rats$ s.t. for any $z \in \Supp \Dist^k$, $Q^k(z)=\chi(z)$ and therefore $\E_{z \sim \Dist^k}[(Q^k(z)-\chi(z))^2]=0$.

We will describe several types of reductions that preserve optimal polynomial-time estimators. After that, we will characterize reductions that can be constructed by composing those types and prove a completeness theorem.

\subsection{Strict Pseudo-Invertible Reductions}

\begin{samepage}
\begin{definition}
\label{def:psp_reduce}

Consider $(\Dist,f)$, $(\mathcal{E},g)$ distributional estimation problems and ${\pi: \Words \Scheme \Words}$. $\pi$ is called a \emph{precise strict pseudo-invertible $\EG$-reduction of $(\Dist,f)$ to $(\mathcal{E},g)$} when

\begin{enumerate}[(i)]

\item\label{con:def__psp_reduce__dist} $\pi_*^K\Dist^{K} \equiv \mathcal{E}^{K} \pmod \Fall$

\item\label{con:def__psp_reduce__fun} Denote ${\bar{g}: \Words \rightarrow \Reals}$ the extension of $g$ by 0. We require

\[\E_{(x,z) \sim \Dist^{K} \times \Un_\pi^{K}}[\Abs{f(x)-\bar{g}(\pi^{K}(x,z))}] \equiv 0 \pmod \Fall\]

\item\label{con:def__psp_reduce__smp} $\Dist$ is polynomial-time $\EMG$-samplable relative to $\pi$.

\end{enumerate}

\end{definition}
\end{samepage}

Note that condition~\ref{con:def__psp_reduce__smp} is violated in the one-way function example above, and in particular it ensures that the problem doesn't become significantly more difficult after applying $\pi$.

Also, notice the similarity of condition ii to a randomized Karp reduction (page 189 in \cite{Goldreich_2008}). Reexpressing that definition in our terminology, it is $\forall x:\E_{z\sim\Un_\pi^{K}}[\Abs{f(x)-g(\pi^{K}(x,z))}]\le\mu(K)$, where $\mu$ is a negligible function.

Precise strict pseudo-invertible $\EG$-reductions preserve $\ESG$-optimal estimators as a simple corollary of Theorem~\ref{thm:mult}:

\begin{samepage}
\begin{corollary}
\label{crl:psp_reduce_sharp}

Consider $(\Dist,f)$, $(\mathcal{E},g)$ distributional estimation problems and $\pi$ a precise strict pseudo-invertible $\EG$-reduction of $(\Dist, f)$ to $(\mathcal{E}, g)$. Suppose $P$ is an $\ESG$-optimal estimator for $(\mathcal{E}, g)$. Then, $P \circ \pi$ is an $\ESG$-optimal estimator for $(\Dist, f)$.

\end{corollary}
\end{samepage}

\begin{proof}

Follows directly from Theorem~\ref{thm:mult} for $f_1 = f$, ${f_2 \equiv 1}$, $P_2 \equiv 1$. This relies on the trivial observation that $(\Dist, 1)$ is samplable relative to $\pi$ iff $\Dist$ is samplable relative to $\pi$.
\end{proof}

$\EG$-optimal estimators are also preserved.

\begin{samepage}
\begin{theorem}
\label{thm:psp_reduce}

Consider $(\Dist,f)$, $(\mathcal{E},g)$ distributional estimation problems and $\pi$ a precise strict pseudo-invertible $\EG$-reduction of $(\Dist, f)$ to $(\mathcal{E}, g)$. Suppose $P$ is an $\EG$-optimal estimator for $(\mathcal{E}, g)$. Then, $P \circ \pi$ is an $\EG$-optimal estimator for $(\Dist, f)$.

\end{theorem}
\end{samepage}

\begin{samepage}
\begin{proposition}
\label{prp:mixed_opt}

Consider ${(\Dist,f)}$ a distributional estimation problem and ${P}$ an ${\EG}$-optimal estimator for ${(\Dist,f)}$. Then, for any $Q: \Words \MScheme \Rats$ bounded

\begin{equation}
\label{eqn:prp__mixed_ort}
\E_{(x,y) \sim \Dist^{K} \times \Un_P^K}[(P^K(x,y) - f(x))^2] \leq \E_{(x,y) \sim \Dist^{K} \times \UM_Q^K}[(Q^K(x,y)-f(x))^2] \pmod \Fall
\end{equation}

\end{proposition}
\end{samepage}

\begin{proof}

For any ${K \in \Nats^n}$, choose 

\[w^K \in \Argmax{w \in \Supp \M_Q^K} \E_{(x,z) \sim \Dist^{K} \times \Un^{\R_Q^K(w)}}[(Q^K(x,z,w)-f(x))^2]\]

Construct ${\bar{Q}: \Words \Scheme \Rats}$ s.t.

\begin{align*}
\R_{\bar{Q}}(K) &= \R_Q^K(w^K) \\
\bar{Q}^K(x,z) &= \bar{Q}^K(x,z,w)
\end{align*}

Equation \ref{eqn:op} for ${\bar{Q}}$ implies \ref{eqn:prp__mixed_ort}.
\end{proof}

\begin{samepage}
\begin{proposition}
\label{prp:sq_diff_cong}

Consider $\{F^K\}_{K \in \Nats^n}$, $\{G_1^K\}_{K \in \Nats^n}$, $\{G_2^K\}_{K \in \Nats^n}$ uniformly bounded families of random variables and suppose ${\E[\Abs{G_1^K - G_2^K}] \in \Fall}$. Then

\begin{equation}
\E[(F^K + G_1^K)^2] \equiv \E[(F^K + G_2^K)^2] \pmod \Fall
\end{equation}

\end{proposition}
\end{samepage}

\begin{proof}

\[\E[(F^K + G_1^K)^2] - \E[(F^K + G_2^K)^2] = \E[(2 F^K + G_1^K + G_2^K)(G_1^K - G_2^K)]\]

\[\Abs{\E[(F^K + G_1^K)^2] - \E[(F^K + G_2^K)^2]} \leq (2 \sup F + \sup G_1 + \sup G_2) \E[\Abs{G_1^K - G_2^K}]\]
\end{proof}

\begin{proof}[Proof of Theorem \ref{thm:psp_reduce}]

Let ${\sigma}$ be an ${\EMG}$-sampler of ${\Dist}$ relative to ${\pi}$. Consider any ${Q: \Words \Scheme \Rats}$ bounded. Applying Proposition~\ref{prp:mixed_opt} for ${P}$ and ${Q \circ \sigma}$, we get

\[\E_{\mathcal{E}^{K} \times \Un_P^K}[(P^K-g)^2] \leq \E_{\mathcal{E}^{K} \times \Un_Q^K \times \UM_\sigma^K}[((Q \circ \sigma)^K - g)^2] \pmod \Fall\]

Using condition~\ref{con:def__psp_reduce__dist} of Definition~\ref{def:psp_reduce}

\[\E_{\pi_*^K\Dist^{K} \times \Un_P^K}[(P^K-\bar{g})^2] \leq \E_{\pi_*^K\Dist^{K} \times \Un_Q^K \times \UM_\sigma^K}[((Q \circ \sigma)^K - \bar{g})^2] \pmod \Fall\]

\[\E_{\pi_*^K\Dist^{K} \times \Un_P^K}[(P^K-\bar{g})^2] \leq \E_{\pi_*^K\Dist^{K} \times \UM_\sigma^K}[\E_{\Un_Q^K}[((Q \circ \sigma)^K - \bar{g})^2]] \pmod \Fall\]

The right hand side has the form of the right hand side in \ref{eqn:prp__smp} enabling us to apply Proposition~\ref{prp:smp} and get

\[\E_{\pi_*^K\Dist^{K} \times \Un_P^K}[(P^K-\bar{g})^2] \leq \E_{\Dist^{K} \times \Un_\pi^K}[\E_{\Un_Q^K}[(Q^K - \bar{g} \circ \pi^K)^2]] \pmod \Fall\]

\[\E_{\Dist^{K} \times \Un_\pi^K \times \Un_P^K}[((P \circ \pi)^K-\bar{g} \circ \pi^K)^2] \leq \E_{\Dist^{K} \times \Un_\pi^K \times \Un_Q^K}[(Q^K - \bar{g} \circ \pi^K)^2] \pmod \Fall\]

By Proposition~\ref{prp:sq_diff_cong} and condition~\ref{con:def__psp_reduce__fun} of Definition~\ref{def:psp_reduce}

\[\E_{\Dist^{K} \times \Un_\pi^K \times \Un_P^K}[((P \circ \pi)^K-f)^2] \leq \E_{\Dist^{K} \times \Un_Q^K}[(Q^K - f)^2] \pmod \Fall\]
\end{proof}

We now consider a more general type of reduction which only preserves the function on average (the only difference is in condition~\ref{con:def__sp_reduce__fun}):

\begin{samepage}
\begin{definition}
\label{def:sp_reduce}

Consider $(\Dist,f)$, $(\mathcal{E},g)$ distributional estimation problems and ${\pi: \Words \Scheme \Words}$. $\pi$ is called a \emph{strict pseudo-invertible $\EG$-reduction of $(\Dist,f)$ to $(\mathcal{E},g)$} when

\begin{enumerate}[(i)]

\item\label{con:def__sp_reduce__dist} $\pi_*^K\Dist^{K} \equiv \mathcal{E}^{K} \pmod \Fall$

\item\label{con:def__sp_reduce__fun} Denote ${\bar{g}: \Words \rightarrow \Reals}$ the extension of $g$ by 0. We require

\[\E_{(x,z) \sim \Dist^{K}}[\Abs{f(x)-\E_{\Un_\pi^{K}}[g(\pi^{K}(x,z))]}] \equiv 0 \pmod \Fall\]

\item\label{con:def__sp_reduce__smp} $\Dist$ is polynomial-time $\EMG$-samplable relative to $\pi$.

\end{enumerate}

\end{definition}
\end{samepage}

\begin{samepage}
\begin{theorem}
\label{thm:sp_reduce_sharp}

Suppose $\gamma \in \GammaPoly^n$ is s.t. $\gamma^{-\frac{1}{2}} \in \Fall$. Consider $(\Dist,f)$, $(\mathcal{E},g)$ distributional estimation problems, $\pi$ a strict pseudo-invertible $\EG$-reduction of $(\Dist, f)$ to $(\mathcal{E}, g)$ and $P_g$ an $\ESG$-optimal estimator for $(\mathcal{E}, g)$. Assume $\gamma (\R_P + \R_\pi) \in \GrowR$. Construct ${P_f}$ s.t. for any ${\{z_i \in \BoolR{\pi}\}_{i \in [\gamma(K)]}}$ and ${\{w_i \in \BoolR{P_g}\}_{i \in [\gamma(K)]}}$

\begin{align}
\label{eqn:thm__sp_reduce__rpf}\R_{P_f}(K) &= \gamma(K) (\R_{P_g}(K) + \R_\pi(K)) \\
\label{eqn:thm__sp_reduce__pf}P_f^K\left(x, \prod_{i \in [\gamma(K)]} w_i z_i\right) &= \frac{1}{\gamma(K)}\sum_{i \in [\gamma(K)]} P_g^K(\pi^K(x,z_i),w_i)
\end{align}

Then, $P_f$ is an $\ESG$-optimal estimator for $(\Dist, f)$.

\end{theorem}
\end{samepage}

\begin{samepage}
\begin{proposition}
\label{prp:ev_equiv_mean}

Consider $\gamma \in \GammaPoly^n$, $\Dist$ a word ensemble and $\bar{g}: \Words \rightarrow \Reals$ bounded. Then,

\begin{equation}
\E_{(x,z) \sim \Dist^{K} \times \prod_{i \in [\gamma(K)]} \Un_\pi^K}[\Abs{\E_{z \sim \Un_\pi^{K}}[\bar{g}(\pi^{K}(x,z))]-\frac{1}{\gamma(K)} \sum_{i \in [\gamma(K)]} \bar{g}(\pi^K(x,z_i))}] \leq \frac{\sup \Abs{\bar{g}}}{\gamma(K)^{\frac{1}{2}}}
\end{equation}

\end{proposition}
\end{samepage}

\begin{proof}

Denote $\Un_\gamma^K:=\prod_{i \in [\gamma(K)]} \Un_\pi^K$. Using $\Abs{X}=\sqrt{X^2}$, applying Jensen's inequality to move the square root outside the second expectation, and partially pulling the $\frac{1}{\gamma(K)}$ out,

\begin{align*}
&\E[\Abs{\E[\bar{g}(\pi^{K}(x,z))]-\frac{1}{\gamma(K)} \sum_{i \in [\gamma(K)]} \bar{g}(\pi^K(x,z_i))}] \leq\\ &\frac{1}{\sqrt{\gamma(K)}}\E_{\Dist^{K}}\left[\sqrt{\frac{1}{\gamma(K)}\E_{\Un_\gamma^K}\left[\left(\sum_{i\in[\gamma(K)]}\E_{\Un_\pi^K}[\bar{g}(\pi^{K}(x,z))]- \bar{g}(\pi^K(x,z_i))\right)^2\right]}\right]
\end{align*}

Because the $z_{i}$ are i.i.d, the sum of the variances is the variance of the sum, so

\[\E_{\Un_\gamma^K}\left[\left(\sum_{i\in[\gamma(K)]}\E_{\Un_\pi^K}[\bar{g}(\pi^{K}(x,z))]- \bar{g}(\pi^K(x,z_i))\right)^2\right]=\gamma(K)\Var_{\Un_\pi^K}[\bar{g}(\pi^{K}(x,z))]\]

Substituting this into the previous equation, canceling $\gamma(K)$, and using the fact that $\sqrt{\Var(X)}\le\sup\Abs{X}$, we get

\[\E[\Abs{\E[\bar{g}(\pi^{K}(x,z))]-\frac{1}{\gamma(K)} \sum_{i \in [\gamma(K)]} \bar{g}(\pi^K(x,z_i))}] \leq \frac{\sup\Abs{\overline{g}}}{\gamma(K)^{\frac{1}{2}}}\]
\end{proof}

\begin{proof}[Proof of Theorem \ref{thm:sp_reduce_sharp}]

Consider any $S: \Words \times \Rats \Scheme \Rats$ bounded. Denote ${\Un_{PS}^K:=\Un_{P_f}^K \times \Un_S^K}$. Using condition~\ref{con:def__sp_reduce__fun} of Definition~\ref{def:sp_reduce}

\[\E_{\Dist^{K} \times \Un_{PS}^K}[(P_f^K(x) - f(x))S(x,P_f^K(x))] \equiv \E_{\Dist^{K} \times \Un_{PS}^K}[(P_f^K(x) - \E_{\Un_\pi^{K}}[g(\pi^{K}(x))])S(x,P_f^K(x))] \pmod \Fall\]

Using the construction of $P_f$, the assumption on $\gamma$ and Proposition~\ref{prp:ev_equiv_mean}, we get

\begin{align*}
&\E[(P_f^K - f)S] \equiv\\ 
&\E_{\Dist^{K} \times \Un_{PS}^K}\left[\left(\frac{1}{\gamma(K)}\sum_{i \in [\gamma(K)]} P_g^K(\pi^K(x,z_i),w_i) - \frac{1}{\gamma(K)} \sum_{i \in [\gamma(K)]} \bar{g}(\pi^K(x,z_i))\right)S(x,P_f^K(x))\right] \pmod \Fall
\end{align*}

\[\E[(P_f^K - f)S] \equiv \frac{1}{\gamma(K)} \sum_{i \in [\gamma(K)]} \E_{\Dist^{K} \times \Un_{PS}^K}[(P_g^K(\pi^K(x,z_i),w_i) - \bar{g}(\pi^K(x,z_i)))S(x,P_f^K(x))] \pmod \Fall\]

All the terms in the sum are equal, therefore

\[\E[(P_f^K - f)S] \equiv \E_{\Dist^{K} \times \Un_{PS}^K}[(P_g^K(\pi^K(x,z_0),w_0) - \bar{g}(\pi^K(x,z_0)))S(x,P_f^K(x))] \pmod \Fall\]

Let $\sigma$ be a polynomial-time $\EMG$-sampler of $\Dist$ relative to ${\pi}$. Denote 

\begin{align*}
\Dist_\pi^K &:= \pi_*^K\Dist^{K} \\
\Un_0^K&:=\left(\prod_{i \in [\gamma(K)]} \Un_{P_g}^K\right) \times \left(\prod_{i \in [\gamma(K)] \setminus 0} \Un_{\pi}^K\right) \times \Un_S^K \times \UM_\sigma^K 
\end{align*}

Applying Proposition~\ref{prp:smp} we get

\[\E[(P_f^K - f)S] \equiv\E_{\Dist_\pi^K \times \Un_0^K}\left[(P_g^K - \bar{g})S\left(\sigma^K,\frac{1}{\gamma(K)}(P_g^K+\sum_{i \in [\gamma(K)] \setminus 0} P_g^K(\pi^K(\sigma^K,z_i)))\right)\right] \pmod \Fall\]

Using condition~\ref{con:def__sp_reduce__dist} of Definition~\ref{def:sp_reduce}, we get

\[\E[(P_f^K - f)S] \equiv \E_{\mathcal{E}^{K} \times \Un_0^K}\left[(P_g^K - g)S\left(\sigma^K,\frac{1}{\gamma(K)}(P_g^K+\sum_{i \in [\gamma(K)] \setminus 0} P_g^K(\pi^K(\sigma^K,z_j)))\right)\right] \pmod \Fall\]

$P_g$ is a $\ESG$-optimal estimator for $(\mathcal{E},g)$, therefore

\[\E[(P_f^K - f)S] \equiv 0 \pmod \Fall\]
\end{proof}

\begin{samepage}
Above we showed that strict pseudo-invertible reductions preserve $\ESG$-optimal estimators. We will now see that they preserve $\EG$-optimal estimators as well, as Theorem~\ref{thm:sp_reduce} states.
\begin{theorem}
\label{thm:sp_reduce}

Suppose $\gamma \in \GammaPoly^n$ is s.t. $\gamma^{-\frac{1}{2}} \in \Fall$. Consider $(\Dist,f)$, $(\mathcal{E},g)$ distributional estimation problems, $\pi$ a strict pseudo-invertible $\EG$-reduction of $(\Dist, f)$ to $(\mathcal{E}, g)$ and $P_g$ an $\EG$-optimal estimator for $(\mathcal{E}, g)$. Assume $\R_P + \gamma \R_\pi \in \GrowR$. Construct ${P_f}$  s.t. for any ${\{z_i \in \BoolR{\pi}\}_{i \in [\gamma(K)]}}$ and ${w \in \BoolR{P_g}}$

\begin{align}
\label{eqn:thm__sp_reduce__rpf}\R_{P_f}(K) &= \R_{P_g}(K) + \gamma(K) \R_\pi(K) \\
\label{eqn:thm__sp_reduce__pf}P_f^K\left(x, w \prod_{i \in [\gamma(K)]} z_i\right) &= \frac{1}{\gamma(K)}\sum_{i \in [\gamma(K)]} P_g^K(\pi^K(x,z_i),w)
\end{align}

Then, $P_f$ is an $\EG$-optimal estimator for ${(\Dist,g)}$.

\end{theorem}
\end{samepage}

\begin{samepage}
\begin{proposition}
\label{prp:ev_diff_sq}

Consider ${F}$ a bounded random variable and ${s,t \in \Reals}$. Then

\begin{equation}
\E[(F - s)^2 - (F - t)^2] = (\E[F] - s)^2 - (\E[F] - t)^2
\end{equation}

\end{proposition}
\end{samepage}

\begin{proof}

\[\E[(F - s)^2 - (F - t)^2] = \E[(2F - s - t)(t-s)]\]

\[\E[(F - s)^2 - (F - t)^2] = (2\E[F] - s - t)(t-s)\]

\[\E[(F - s)^2 - (F - t)^2] = (\E[F] - s)^2 - (\E[F] - t)^2\]
\end{proof}

\begin{proof}[Proof of Theorem \ref{thm:sp_reduce}]

Let ${\sigma}$ be an ${\EMG}$-sampler of ${\Dist}$ relative to ${\pi}$. Consider any\\ $Q_f: \Words \Scheme \Rats$ bounded. Construct ${Q_g: \Words \MScheme \Rats}$ s.t. for any ${z_\sigma \in \UM_\sigma^K}$, ${z_Q \in \BoolR{Q_f}}$, ${z_\pi \in \Bool^{\gamma(K) \R_\pi(K)}}$ and ${z_g \in \BoolR{P_g}}$

\begin{align*}
\M_{Q_g}^K &= c_*^4(\M_\sigma^K \times \M_{Q_f}^K \times \M_{\pi}^K \times \M_{P_g}^K) \\
\R_{Q_g}^K(\Chev{z_{\sigma1}, \A_{Q_f}(K),\A_{\pi}(K),\A_{P_g}(K)}) &= \R_\sigma^K(z_{\sigma1}) + \R_{Q_f}(K) + \gamma(K)\R_{\pi}(K) + \R_{P_g}(K) \\
Q_g^K(x,z_{\sigma0} z_{Q} z_{\pi} z_{g}, \Chev{z_{\sigma1}, \A_{Q_f}(K),\A_{\pi}(K),\A_{P_g}(K)}) &= Q_f^K(\sigma^K(x,z_\sigma),z_{Q})-P_f^K(\sigma^K(x,z_\sigma),z_g z_\pi)+P_g^K(x,z_g)
\end{align*}

Applying Proposition~\ref{prp:mixed_opt} for ${P_g}$ and ${Q_g}$, we get

\[\E_{\mathcal{E}^{K} \times \Un_{P_g}^K}[(P_g^K - g)^2] \leq \E_{\mathcal{E}^{K} \times \UM_{Q_g}^K}[(Q_g^K - g)^2] \pmod \Fall\]

Using condition~\ref{con:def__sp_reduce__dist} of Definition~\ref{def:sp_reduce}

\[\E_{\pi_*^K\Dist^{K} \times \Un_{P_g}^K}[(P_g^K-\bar{g})^2] \leq \E_{\pi_*^K\Dist^{K} \times \UM_{Q_g}^K}[(Q_g^K - \bar{g})^2] \pmod \Fall\]

\[\E_{\pi_*^K\Dist^{K} \times \Un_{P_g}^K}[(P_g^K-\bar{g})^2] \leq \E_{\pi_*^K\Dist^{K} \times \UM_{Q_g}^K}[((Q_f \circ \sigma)^K - (P_f \circ \sigma)^K + P_g^K - \bar{g})^2] \pmod \Fall\]

The right hand side has the form of the right hand side in \ref{eqn:prp__smp} enabling us to apply Proposition~\ref{prp:smp} and get

\[\E_{\Dist^{K} \times \Un_\pi^K \times \Un_{P_g}^K}[((P_g \circ \pi)^K-\bar{g} \circ \pi^K)^2] \leq \E_{\Dist^{K} \times \Un_\pi^K \times \Un_{Q_f}^K \times \Un_{P_f}^K}[(Q_f ^K - P_f^K+(P_g \circ \pi)^K - \bar{g} \circ \pi^K)^2] \pmod \Fall\]

We can consider the expressions within the expected values on both sides as random variables w.r.t. $\Un_\pi^K$ while fixing the other components of the distribution. This allows us applying Proposition~\ref{prp:ev_diff_sq} to the difference between the right hand side and the left hand side (with the terms that don't depend on $\Un_\pi^K$ playing the role of the constants), which results in moving the expected value over $\Un_\pi^K$ inside the squares. Let $\Un^K_{PQ}:=\Un^K_{Q_f}\times \Un^K_{P_f}$.

\[\E_{\Dist^{K} \times \Un_{P_g}^K}[\E_{\Un_\pi^K}[(P_g \circ \pi)^K-\bar{g} \circ \pi^K]^2] \leq \E_{\Dist^{K} \times \Un_{PQ}^K} [(Q_f ^K - P_f^K+\E_{\Un_\pi^K}[(P_g \circ \pi)^K - \bar{g} \circ \pi^K])^2] \pmod \Fall\]

\begin{align*}
&\E_{\Dist^{K} \times \Un_{P_g}^K}[(\E_{\Un_\pi^K}[(P_g \circ \pi)^K]-\E_{\Un_\pi^K}[\bar{g} \circ \pi^K])^2] \leq\\ 
&\E_{\Dist^{K} \times \Un_{PQ}^K} [(Q_f ^K - P_f^K+\E_{\Un_\pi^K}[(P_g \circ \pi)^K] - \E_{\Un_\pi^K}[\bar{g} \circ \pi^K])^2] \pmod \Fall 
\end{align*}

We now apply Proposition~\ref{prp:sq_diff_cong} via condition \ref{con:def__sp_reduce__fun} of Definition~\ref{def:sp_reduce}

\[\E_{\Dist^{K} \times \Un_{P_g}^K}[(\E_{\Un_\pi^K}[(P_g \circ \pi)^K]-f)^2] \leq \E_{\Dist^{K} \times \Un_{PQ}^K}[(Q_f ^K - P_f^K+\E_{\Un_\pi^K}[(P_g \circ \pi)^K] - f)^2] \pmod \Fall\]

Denote $y_i:=\pi^K(x,z_i)$ where the ${z_i}$ are sampled independently from ${\Un_\pi^K}$. Applying Proposition~\ref{prp:sq_diff_cong} via Proposition~\ref{prp:ev_equiv_mean} and the assumption on $\gamma$, we get

\begin{align*}
&\E_{\Dist^{K} \times \Un_{P_f}^K}\left[\left(\frac{1}{\gamma(K)}\sum_{i \in [\gamma(K)]}P_g^K(y_i)-f\right)^2\right] \leq \\ 
&\E_{\Dist^{K} \times \Un_{PQ}^K}\left[\left(Q_f ^K - P_f^K+\frac{1}{\gamma(K)}\sum_{i \in [\gamma(K)]}P_g^K(y_i) - f\right)^2\right] \pmod \Fall 
\end{align*}

\[\E_{\Dist^{K} \times \Un_{P_f}^K}[(P_f^K-f)^2] \leq \E_{\Dist^{K} \times \Un_{PQ}^K}[(Q_f ^K - P_f^K+P_f^K - f)^2] \pmod \Fall\]

\[\E_{\Dist^{K} \times \Un_{P_f}^K}[(P_f^K-f)^2] \leq \E_{\Dist^{K} \times \Un_{Q_f}^K}[(Q_f ^K - f)^2] \pmod \Fall\]
\end{proof}

\subsection{Dominance}

Next, we consider a scenario in which the identity mapping can be regarded as a valid reduction between distributional estimation problems that have the same function but different word ensembles.

\begin{samepage}
\begin{definition}

Consider ${\Dist}$, ${\mathcal{E}}$ word ensembles. ${\Dist}$ is said to be \emph{${\EG}$-dominated by ${\mathcal{E}}$} when there is ${W: \Words \Scheme \Rats^{\geq 0}}$ bounded s.t.

\begin{equation}
\sum_{x \in \Words} \Abs{\mathcal{E}^{K}(x)\E_{\Un_W^K}[W^K(x)]-\Dist^{K}(x)} \in \Fall
\end{equation}

In this case, ${W}$ is called a \emph{Radon-Nikodym ${\EG}$-derivative of ${\Dist}$ w.r.t. ${\mathcal{E}}$}.


\end{definition}
\end{samepage}

\begin{samepage}
\begin{proposition}
\label{prp:dom_reduce_sharp}

Consider ${\Dist}$, ${\mathcal{E}}$ word ensembles, ${f: \Supp \Dist \cup \Supp \mathcal{E} \rightarrow \Reals}$ bounded and ${P}$ an ${\ESG}$-optimal estimator for ${(\mathcal{E},f)}$. Suppose ${\Dist}$ is ${\EG}$-dominated by ${\mathcal{E}}$. Then, ${P}$ is an ${\ESG}$-optimal estimator for ${(\Dist,f)}$.

\end{proposition}
\end{samepage}

\begin{proof}

Let ${W}$ be a Radon-Nikodym ${\EG}$-derivative of ${\Dist}$ w.r.t. ${\mathcal{E}}$. Consider any ${S: \Words \times \Rats \Scheme \Rats}$ bounded.

\[\E_{\mathcal{E}^{K} \times \Un_P^K \times \Un_W^K \times \Un_S^K}[(P^K(x)-f(x))W^K(x)S^K(x,P^K(x))] \equiv 0 \pmod \Fall\]

\[\sum_{x \in \Words} \mathcal{E}^{K}(x) \E_{\Un_W^K}[W^K(x)] \E_{\Un_P^K \times \Un_S^K}[(P^K(x)-f(x))S^K(x,P^K(x))] \equiv 0 \pmod \Fall\]

\[\sum_{x \in \Words} (\mathcal{E}^{K}(x) \E_{\Un_W^K}[W^K(x)] - \Dist^{K}(x) + \Dist^{K}(x)) \E_{\Un_P^K \times \Un_S^K}[(P^K(x)-f(x))S^K(x,P^K(x))] \equiv 0 \pmod \Fall\]

\begin{align*}
\sum_{x \in \Words} (\mathcal{E}^{K}(x) \E_{\Un_W^K}[W^K(x)] - \Dist^{K}(x)) \E_{\Un_P^K \times \Un_S^K}[(P^K-f)S^K] +\\
\sum_{x \in \Words} \ \Dist^{K}(x) \E_{\Un_P^K \times \Un_S^K}[(P^K-f)S] \equiv 0 \pmod \Fall 
\end{align*}

\[\E_{\Dist^{K} \times \Un_P^K \times \Un_S^K}[(P^K-f)S] \equiv -\sum_{x \in \Words} (\mathcal{E}^{K}(x) \E_{\Un_W^K}[W^K(x)] - \Dist^{K}(x)) \E_{\Un_P^K \times \Un_S^K}[(P^K-f)S^K] \pmod \Fall\]

\[\Abs{\E_{\Dist^{K} \times \Un_P^K \times \Un_S^K}[(P^K-f)S]} \leq (\sup \Abs{P} + \sup \Abs{f}) \sup \Abs{S} \sum_{x \in \Words} \Abs{\mathcal{E}^{K}(x) \E_{\Un_W^K}[W^K(x)] - \Dist^{K}(x)} \pmod \Fall\]

\[\E_{\Dist^{K} \times \Un_P^K \times \Un_S^K}[(P^K-f)S] \equiv 0 \pmod \Fall\]
\end{proof}

The corresponding statement for ${\EG}$-optimal estimators may be regarded as a generalization of Corollary~\ref{crl:weight}.

\begin{samepage}
\begin{proposition}
\label{prp:dom_reduce}

Assume ${\Fall}$ is ${\GrowA}$-ample. Consider ${\Dist}$, ${\mathcal{E}}$ word ensembles, ${f: \Supp \Dist \cup \Supp \mathcal{E} \rightarrow \Reals}$ bounded and ${P}$ an ${\EG}$-optimal estimator for ${(\mathcal{E},f)}$. Suppose ${\Dist}$ is ${\EG}$-dominated by ${\mathcal{E}}$. Then, ${P}$ is an ${\EG}$-optimal estimator for ${(\Dist,f)}$.

\end{proposition}
\end{samepage}

\begin{proof}

Let ${W}$ be a Radon-Nikodym ${\EG}$-derivative of ${\Dist}$ w.r.t. ${\mathcal{E}}$. Consider any ${Q: \Words \Scheme \Rats}$ bounded. According to Proposition~\ref{prp:weight}

\[\E_{\mathcal{E}^{K} \times \Un_W^K \times \Un_P^K}[W^K(x)(P^K(x)-f(x))^2] \leq \E_{\mathcal{E}^{K} \times \Un_W^K \times \Un_Q^K}[W^K(x)(Q^K(x)-f(x))^2] \pmod \Fall\]

\begin{align*}
&\sum_{x \in \Words} \mathcal{E}^{K}(x) \E_{\Un_W^K}[W^K(x)] \E_{\Un_P^K}[(P^K(x)-f(x))^2] \leq\\
&\sum_{x \in \Words} \mathcal{E}^{K}(x) \E_{\Un_W^K}[W^K(x)] \E_{\Un_Q^K}[(Q^K(x)-f(x))^2] \pmod \Fall
\end{align*}

Using the assumption on ${W}$

\[\sum_{x \in \Words} \Dist^{K}(x) \E_{\Un_P^K}[(P^K(x)-f(x))^2] \leq \sum_{x \in \Words} \Dist^{K}(x) \E_{\Un_Q^K}[(Q^K(x)-f(x))^2] \pmod \Fall\]

\[\E_{\Dist^{K} \times \Un_P^K}[(P^K(x)-f(x))^2] \leq \E_{\Dist^{K} \times \Un_Q^K}[(Q^K(x)-f(x))^2] \pmod \Fall\]
\end{proof}

\subsection{Ensemble Pullbacks}

Finally, we consider another scenario in which the identity mapping is a valid reduction. This scenario is a simple re-indexing of the word ensemble (redefinition of the security parameters). For the remainder of section~\ref{sec:reductions}, we fix some ${m \in \Nats}$. Note that is important that the growth spaces for the resources and fall space for the error, after reindexing, lie in the growth spaces and fall space of the new problem.

\begin{samepage}
\begin{definition}

We denote ${\GammaPoly^{mn}:=\{\gamma: \Nats^m \rightarrow \Nats^n \mid \forall i \in [n]: \gamma_i \in \GammaPoly^m \}}$.

\end{definition}
\end{samepage}

\begin{samepage}
\begin{definition}

Consider ${\Gamma_*}$ a growth space of rank ${n}$ and ${\alpha \in \GammaPoly^{mn}}$. We introduce the notation

\begin{equation}
\label{eqn:tbd}
\Gamma_*\alpha:=\{\gamma_\alpha: \Nats^m \rightarrow \Reals^{\geq 0} \mid \exists \gamma \in \Gamma_*: \gamma_\alpha \leq \gamma \circ \alpha\}
\end{equation}

Obviously ${\Gamma_*\alpha}$ is a growth space of rank ${m}$.

We also denote ${\Gamma \alpha := (\GrowR \alpha, \GrowA \alpha)}$.

\end{definition}
\end{samepage}

\begin{samepage}
\begin{definition}

Consider ${\alpha \in \GammaPoly^{mn}}$. We introduce the notation

\begin{equation}
\label{eqn:tbd}
\Fall \alpha:=\{\varepsilon_\alpha: \Nats^m \rightarrow \Reals^{\geq 0} \textnormal{ bounded} \mid \exists \varepsilon \in \Fall: \varepsilon_\alpha \leq \varepsilon \circ \alpha\}
\end{equation}

\end{definition}
\end{samepage}

\begin{samepage}
\begin{proposition}
\label{prp:tbd}

For any ${\alpha \in \GammaPoly^{mn}}$, ${\Fall \alpha}$ is a fall space.

\end{proposition}
\end{samepage}

\begin{proof}

Conditions \ref{con:def__fall__add} and \ref{con:def__fall__ineq} are obvious. To verify condition \ref{con:def__fall__pol}, consider ${h \in \NatPoly}$ s.t. ${2^{-h} \in \Fall}$. Note that since the coefficients of ${h}$ are non-negative it is non-decreasing in all arguments. Consider ${p: \Nats^m \rightarrow \Nats^n}$ a polynomial map s.t. for any ${i \in [n]}$, ${\alpha_i \leq p_i}$. We have ${2^{-h \circ p} \leq 2^{-h \circ \alpha}}$ and therefore ${2^{-h \circ p} \in \Fall \alpha}$.
\end{proof}

\begin{samepage}
\begin{definition}

Consider ${\Dist}$ a word ensemble of rank ${n}$ and ${\alpha: \Nats^m \rightarrow \Nats^n}$. The \emph{pullback of ${\Dist}$ by ${\alpha}$}, denoted ${\Dist^\alpha}$, is the word ensemble of rank ${m}$ given by ${(\Dist^\alpha)^k:=\Dist^{\alpha(k)}}$.

\end{definition}
\end{samepage}

\begin{samepage}
\begin{definition}

Consider ${X}$, ${Y}$ encoded sets, ${S: X \Scheme Y}$ and ${\alpha: \Nats^m \Alg \Nats^n}$ s.t. ${\alpha \in \GammaPoly^{mn}}$ as a function and ${\T_\alpha \in \GammaPoly^m}$. We define ${S^\alpha: X \xrightarrow{\Gamma \alpha} Y}$ by requiring that for any ${L \in \Nats^m}$,\\ $\R_{S^\alpha}(L)=\R_S(\alpha(L))$ and ${(S^\alpha)^L(x,y)=S^{\alpha(L)}(x,y)}$.

\end{definition}
\end{samepage}

\begin{samepage}
\begin{proposition}
\label{prp:rev_sch_idx}

Consider  ${X}$, ${Y}$ encoded sets, ${\alpha: \Nats^m \Alg \Nats^n}$ and ${\beta \in \GammaPoly^{nm}}$. Assume that\\ $\T_\alpha \in \GammaPoly^m$ and ${\forall L \in \Nats^m: \beta(\alpha(L))=L}$. Then, for any  ${S: X \xrightarrow{\Gamma \alpha} Y}$ there is ${\tilde{S}: X \Scheme Y}$ s.t. for all ${K \in \Nats^n}$ that satisfy ${\alpha(\beta(K))=K}$, ${x \in X}$ and ${y,z \in \Words}$

\begin{align}
\A_{\tilde{S}}(K)&=\A_S(\beta(K)) \\
\R_{\tilde{S}}^K(z)&=\R_S^{\beta(K)}(z) \\
\tilde{S}^K(x,y,z)&=S^{\beta(K)}(x,y,z)
\end{align}

\end{proposition}
\end{samepage}

\begin{proof}

To see there is no obstruction of time complexity, note that ${\beta}$ can be computed by some\\ ${\beta^*: \Nats^n \Alg \Nats^m}$ s.t. ${\T_{\beta^*} \in \GammaPoly^n}$. Given input ${K}$, ${\beta^*}$ works by iterating over all ${L}$ within some polynomial size range (thanks to the assumption ${\beta \in \GammaPoly^{nm}}$) and checking the condition ${\alpha(L)=K}$.

To see there are no obstructions of random or advice complexity, note there is ${\gamma_{\mathfrak{R}} \in \GrowR}$ s.t.\\ $\R_S(L) \leq \gamma_{\mathfrak{R}}(\alpha(L))$ and ${\gamma_{\mathfrak{A}} \in \GrowA}$ s.t. ${\Abs{\A_S(L)} \leq \gamma_{\mathfrak{A}}(\alpha(L))}$. In particular, if ${K \in \Nats^n}$ is s.t.\\ $\alpha(\beta(K))=K$ then ${\R_S(\beta(K)) \leq \gamma_{\mathfrak{R}}(K)}$ and ${\Abs{\A_S(\beta(K))} \leq \gamma_{\mathfrak{A}}(K)}$.
\end{proof}

\begin{samepage}
\begin{definition}

${\alpha: \Nats^m \Alg \Nats^n}$ is called an \emph{efficient injection} when ${\alpha \in \GammaPoly^{mn}}$ as a function,\\ $\T_\alpha \in \GammaPoly^m$ and there is ${\beta \in \GammaPoly^{nm}}$ s.t. ${\forall L \in \Nats^m: \beta(\alpha(L))=L}$.

\end{definition}
\end{samepage}

\begin{samepage}
\begin{proposition}
\label{prp:idx_reduce_sharp}

Consider $(\Dist,f)$ a distributional estimation problem of rank ${n}$, ${P}$ an ${\ESG}$-optimal estimator for ${(\Dist,f)}$ and ${\alpha: \Nats^m \Alg \Nats^n}$ an efficient injection. Then, ${P^\alpha}$ is an ${\Fall \alpha^\sharp(\Gamma \alpha)}$-optimal estimator for ${(\Dist^\alpha,f)}$.

\end{proposition}
\end{samepage}

\begin{proof}

Consider any ${S: \Words \times \Rats \xrightarrow{\Gamma \alpha} \Rats}$ bounded. Construct ${\tilde{S}: \Words \times \Rats \Scheme \Rats}$ by applying Proposition~\ref{prp:rev_sch_idx} to ${S}$. There is ${\varepsilon \in \Fall}$ s.t. for any ${K \in \Nats^n}$

\[\Abs{\E_{\Dist^{K} \times \Un_P^K \times \Un_{\tilde{S}}^K}[(P^K(x,y) - f(x))\tilde{S}^K(x,P^K(x,y),z)]}=\varepsilon(K)\]

Substituting ${\alpha(L)}$ for ${K}$, we get

\[\Abs{\E_{\Dist^{\alpha(L)} \times \Un_P^{\alpha(L)} \times \Un_{\tilde{S}}^{\alpha(L)}}[(P^{\alpha(L)}(x,y) - f(x))\tilde{S}^{\alpha(L)}(x,P^{\alpha(L)}(x,y),z)]}=\varepsilon(\alpha(L))\]

\[\Abs{\E_{(\Dist^\alpha)^{L} \times \Un_{P^\alpha}^{L} \times \Un_{\tilde{S}}^{\alpha(L)}}[((P^\alpha)^{L}(x,y) - f(x))\tilde{S}^{\alpha(L)}(x,(P^\alpha)^{L}(x,y),z)]}=\varepsilon(\alpha(L))\]

We have ${\alpha(\beta(\alpha(L))=\alpha(L)}$, therefore

\[\Abs{\E_{(\Dist^\alpha)^{L} \times \Un_{P^\alpha}^{L} \times \Un_{S}^{\beta(\alpha(L))}}[((P^\alpha)^{L}(x,y) - f(x))S^{\beta(\alpha(L))}(x,(P^\alpha)^{L}(x,y),z)]}=\varepsilon(\alpha(L))\]

\[\Abs{\E_{(\Dist^\alpha)^{L} \times \Un_{P^\alpha}^{L} \times \Un_{S}^{L}}[((P^\alpha)^{L}(x,y) - f(x))S^{L}(x,(P^\alpha)^{L}(x,y),z)]}=\varepsilon(\alpha(L))\]
\end{proof}

\begin{samepage}
\begin{proposition}
\label{prp:idx_reduce}

Consider $(\Dist,f)$ a distributional estimation problem of rank ${n}$, ${P}$ an ${\EG}$-optimal estimator for ${(\Dist,f)}$ and ${\alpha: \Nats^m \Alg \Nats^n}$ an efficient injection. Then, ${P^\alpha}$ is an ${\Fall \alpha(\Gamma \alpha)}$-optimal estimator for ${(\Dist^\alpha,f)}$.

\end{proposition}
\end{samepage}

\begin{proof}

Consider any ${Q: \Words \xrightarrow{\Gamma \alpha} \Rats}$ bounded. Construct ${\tilde{Q}: \Words \Scheme \Rats}$ by applying Proposition~\ref{prp:rev_sch_idx} to ${Q}$. There is ${\varepsilon \in \Fall}$ s.t.

\[\E_{\Dist^{K} \times \Un_P^K}[(P^K(x,y)-f(x))^2] \leq \E_{\Dist^{K} \times \Un_{\tilde{Q}}^K}[(\tilde{Q}^K(x,y)-f(x))^2] + \varepsilon(K)\]

Substituting ${\alpha(L)}$ for ${K}$, we get

\[\E_{\Dist^{\alpha(L)} \times \Un_P^{\alpha(L)}}[(P^{\alpha(L)}(x,y)-f(x))^2] \leq \E_{\Dist^{\alpha(L)} \times \Un_{\tilde{Q}}^{\alpha(L)}}[(\tilde{Q}^{\alpha(L)}(x,y)-f(x))^2] + \varepsilon({\alpha(L)})\]

\[\E_{(\Dist^\alpha)^{L} \times \Un_{P^\alpha}^L}[((P^\alpha)^L(x,y)-f(x))^2] \leq \E_{(\Dist^\alpha)^{L} \times \Un_{\tilde{Q}}^{\alpha(L)}}[(\tilde{Q}^{\alpha(L)}(x,y)-f(x))^2] + \varepsilon(\alpha(L))\]

We have ${\alpha(\beta(\alpha(L))=\alpha(L)}$, therefore

\[\E_{(\Dist^\alpha)^{L} \times \Un_{P^\alpha}^L}[((P^\alpha)^L(x,y)-f(x))^2] \leq \E_{(\Dist^\alpha)^{L} \times \Un_Q^{\beta(\alpha(L))}}[(Q^{\beta(\alpha(L))}(x,y)-f(x))^2] + \varepsilon(\alpha(L))\]

\[\E_{(\Dist^\alpha)^{L} \times \Un_{P^\alpha}^L}[((P^\alpha)^L(x,y)-f(x))^2] \leq \E_{(\Dist^\alpha)^{L} \times \Un_Q^L}[(Q^L(x,y)-f(x))^2] + \varepsilon(\alpha(L))\]
\end{proof}

\subsection{Lax Pseudo-Invertible Reductions}

We now consider compositions of reductions of different types. For the remainder of the section, we fix ${\mathcal{G}}$, a fall space of rank ${m}$. 

\begin{samepage}
\begin{definition}
\label{def:pp_reduce}

Consider $(\Dist,f)$ a distributional estimation problem of rank ${m}$, $(\mathcal{E},g)$ a distributional estimation problem of rank ${n}$, ${\alpha: \Nats^m \Alg \Nats^n}$ an efficient injection and ${\pi: \Words \xrightarrow{\Gamma \alpha} \Words}$. $\pi$ is called a \emph{precise pseudo-invertible $\mathcal{G}(\Gamma)$-reduction of $(\Dist,f)$ to $(\mathcal{E},g)$ over ${\alpha}$} when

\begin{enumerate}[(i)]

\item\label{con:def__pp_reduce__dist} ${\pi_*\Dist}$ is ${\mathcal{G}(\Gamma \alpha)}$-dominated by ${\mathcal{E}^\alpha}$.

\item\label{con:def__pp_reduce__fun} Denote ${\bar{g}: \Words \rightarrow \Reals}$ the extension of $g$ by 0. We require

\[\E_{(x,z) \sim \Dist^{K} \times \Un_\pi^{K}}[\Abs{f(x)-\bar{g}(\pi^{K}(x,z))}] \equiv 0 \pmod {\mathcal{G}}\]

\item\label{con:def__pp_reduce__smp} $\Dist$ is $\mathcal{G}(\MGrow \alpha)$-samplable relative to $\pi$.

\end{enumerate}

\end{definition}
\end{samepage}

\begin{samepage}


\begin{corollary}
\label{crl:pp_reduce_sharp}

Consider $(\Dist,f)$ a distributional estimation problem of rank ${m}$, $(\mathcal{E},g)$ distributional estimation problem of rank ${n}$, ${\alpha: \Nats^m \Alg \Nats^n}$ an efficient injection and $\pi$ a precise pseudo-invertible $\mathcal{G}(\Gamma)$-reduction of $(\Dist, f)$ to $(\mathcal{E}, g)$ over ${\alpha}$. Assume ${\Fall\alpha \subseteq \mathcal{G}}$. Suppose $P$ is an $\ESG$-optimal estimator for $(\mathcal{E}, g)$. Then, $P^\alpha \circ \pi$ is a $\mathcal{G}^\sharp (\Gamma \alpha)$-optimal estimator for $(\Dist, f)$.

\end{corollary}
\end{samepage}

\begin{proof}

By Proposition~\ref{prp:idx_reduce_sharp}, ${P^\alpha}$ is an ${\Fall \alpha^\sharp(\Gamma \alpha)}$-optimal estimator (and in particular a ${\mathcal{G}^\sharp (\Gamma \alpha)}$-optimal estimator) for ${(\mathcal{E}^\alpha, g)}$. By Proposition~\ref{prp:dom_reduce_sharp} and condition \ref{con:def__pp_reduce__dist} of Definition~\ref{def:pp_reduce}, ${P^\alpha}$ is also a ${\mathcal{G}^\sharp (\Gamma \alpha)}$-optimal estimator for ${(\pi_* \Dist, g)}$. By Corollary~\ref{crl:psp_reduce_sharp} and conditions \ref{con:def__pp_reduce__fun} and \ref{con:def__pp_reduce__smp} of Definition~\ref{def:pp_reduce}, ${P^\alpha \circ \pi}$ is a ${\mathcal{G}^\sharp (\Gamma \alpha)}$-optimal estimator for ${(\Dist, f)}$.
\end{proof}

\begin{samepage}
\begin{corollary}

Consider $(\Dist,f)$ a distributional estimation problem of rank ${m}$, $(\mathcal{E},g)$ distributional estimation problem of rank ${n}$, ${\alpha: \Nats^m \Alg \Nats^n}$ an efficient injection and $\pi$ a precise pseudo-invertible $\mathcal{G}(\Gamma)$-reduction of $(\Dist, f)$ to $(\mathcal{E}, g)$ over ${\alpha}$. Assume ${\Fall\alpha \subseteq \mathcal{G}}$ and ${\mathcal{G}}$ is ${\GrowA \alpha}$-ample. Suppose $P$ is an $\EG$-optimal estimator for $(\mathcal{E}, g)$. Then, $P^\alpha \circ \pi$ is a $\mathcal{G} (\Gamma \alpha)$-optimal estimator for $(\Dist, f)$.

\end{corollary}
\end{samepage}

\begin{proof}

Completely analogous to proof of Corollary~\ref{crl:pp_reduce_sharp}.
\end{proof}

\begin{samepage}
\begin{definition}
\label{def:p_reduce}

Consider $(\Dist,f)$ a distributional estimation problem of rank ${m}$, $(\mathcal{E},g)$ a distributional estimation problem of rank ${n}$, ${\alpha: \Nats^m \Alg \Nats^n}$ an efficient injection and ${\pi: \Words \xrightarrow{\Gamma \alpha} \Words}$. $\pi$ is called a \emph{pseudo-invertible $\mathcal{G}(\Gamma)$-reduction of $(\Dist,f)$ to $(\mathcal{E},g)$ over ${\alpha}$} when

\begin{enumerate}[(i)]

\item\label{con:def__p_reduce__dist} ${\pi_*\Dist}$ is ${\mathcal{G}(\Gamma \alpha)}$-dominated by ${\mathcal{E}^\alpha}$.

\item\label{con:def__p_reduce__fun} Denote ${\bar{g}: \Words \rightarrow \Reals}$ the extension of $g$ by 0. We require \[\E_{(x,z) \sim \Dist^{K}}[\Abs{f(x)-\E_{\Un_\pi^{K}}[g(\pi^{K}(x,z))]}] \equiv 0 \pmod{\mathcal{G}}\]

\item\label{con:def__p_reduce__smp} $\Dist$ is $\mathcal{G}(\MGrow \alpha)$-samplable relative to $\pi$.

\end{enumerate}

\end{definition}
\end{samepage}

The following corollaries are completely analogous to Corollary~\ref{crl:pp_reduce_sharp} and therefore given without proof. We also drop the explicit constructions of the optimal polynomial-time estimators which are obviously modeled on Theorem~\ref{thm:sp_reduce_sharp} and Theorem~\ref{thm:sp_reduce}.

\begin{samepage}
\begin{corollary}
\label{crl:p_reduce_sharp}

Consider $(\Dist,f)$ a distributional estimation problem of rank ${m}$, $(\mathcal{E},g)$ distributional estimation problem of rank ${n}$, ${\alpha: \Nats^m \Alg \Nats^n}$ an efficient injection and $\pi$ a pseudo-invertible $\mathcal{G}(\Gamma)$-reduction of $(\Dist, f)$ to $(\mathcal{E}, g)$ over ${\alpha}$. Assume ${\Fall\alpha \subseteq \mathcal{G}}$. Suppose there exist ${P}$ an $\ESG$-optimal estimator for $(\mathcal{E}, g)$ and ${\gamma \in \GammaPoly^m}$ s.t. ${\gamma^{-\frac{1}{2}} \in \mathcal{G}}$ and ${\gamma(\R_P \circ \alpha + \R_\pi) \in \GrowR \alpha}$. Then, there exists a $\mathcal{G}^\sharp (\Gamma \alpha)$-optimal estimator for $(\Dist, f)$.


\end{corollary}
\end{samepage}

\begin{samepage}
\begin{corollary}
\label{crl:p_reduce}

Consider $(\Dist,f)$ a distributional estimation problem of rank ${m}$, $(\mathcal{E},g)$ distributional estimation problem of rank ${n}$, ${\alpha: \Nats^m \Alg \Nats^n}$ an efficient injection and $\pi$ a pseudo-invertible $\mathcal{G}(\Gamma)$-reduction of $(\Dist, f)$ to $(\mathcal{E}, g)$ over ${\alpha}$. Assume ${\Fall\alpha \subseteq \mathcal{G}}$ and ${\mathcal{G}}$ is ${\GrowA \alpha}$-ample. Suppose there exist ${P}$ an $\EG$-optimal estimator for $(\mathcal{E}, g)$ and ${\gamma \in \GammaPoly^m}$ s.t. ${\gamma^{-\frac{1}{2}} \in \mathcal{G}}$ and ${\R_P \circ \alpha + \gamma\R_\pi \in \GrowR \alpha}$. Then, there exists a $\mathcal{G} (\Gamma \alpha)$-optimal estimator for $(\Dist, f)$.


\end{corollary}
\end{samepage}

Note that the last results involved passing from fall space ${\Fall}$ and growth spaces ${\Gamma}$ to fall space ${\mathcal{G}}$ and growth spaces ${\Gamma \alpha}$, however in many natural examples ${m = n}$, ${\mathcal{G} = \Fall}$ and ${\Gamma \alpha = \Gamma}$. In particular, the following propositions are often applicable.

\begin{samepage}
\begin{proposition}
\label{prp:stable_growth_space}

Assume ${\Gamma_*}$ is a growth space of rank ${n}$ s.t. for any ${\gamma \in \Gamma_*}$ and ${\alpha \in \GammaPoly^{nn}}$, ${\gamma \circ \alpha \in \Gamma_*}$. Let ${\alpha^*, \beta^* \in \GammaPoly^{nn}}$ be s.t. ${\beta^*(\alpha^*(K))=K}$. Then, ${\Gamma_* \alpha^* = \Gamma_*}$.

\end{proposition}
\end{samepage}

\begin{proof}

For any ${\gamma_\alpha \in \Gamma_* \alpha^*}$ there is ${\gamma \in \Gamma_*}$ s.t. ${\gamma_\alpha \leq \gamma \circ \alpha \in \Gamma_*}$. Conversely, for any ${\gamma \in \Gamma_*}$ we have ${\gamma = \gamma \circ \beta \circ \alpha \in \Gamma_* \alpha^*}$.
\end{proof}

\begin{samepage}
\begin{proposition}
\label{prp:gamma_r_alpha_p}

Consider ${r: \NatFun \Nats}$ steadily growing and ${p \in \NatPoly}$ increasing in the last argument. Define ${\alpha_p: \NatFun \Nats}$ by ${\forall J \in \Nats^{n-1}, k \in \Nats: \alpha_p(J,k)=(J,p(J,k))}$. Then, ${\Gamma_r \alpha_p = \Gamma_r}$.

\end{proposition}
\end{samepage}

\begin{proof}

Consider ${\gamma_\alpha \in \Gamma_r \alpha_p}$. There is ${\gamma \in \Gamma_r}$ s.t. ${\gamma_\alpha \leq \gamma \circ \alpha_p}$. There is ${q \in \NatPoly}$ s.t. ${\gamma(J,k) \leq r(J,q(J,k))}$. We get ${\gamma_\alpha(J,k) \leq \gamma(J,p(J,k)) \leq r(J,q(J,p(J,k)))}$ and therefore ${\gamma_\alpha \in \Gamma_r}$. Conversely, consider ${\gamma' \in \Gamma_r}$. There is ${q' \in \NatPoly}$ s..t ${\gamma'(J,k) \leq r(J,q'(J,k))}$.\\ $p(J,k) \geq k$ and ${r}$ is non-decreasing in the last argument, implying that ${r \leq r \circ \alpha_p}$. We conclude that ${\gamma'(J,k) \leq r(J,p(J,q'(J,k)))}$ and therefore ${\gamma' \in \Gamma_r \alpha_p}$.
\end{proof}

\subsection{Completeness}


Fix ${r,s: \Nats^n \Alg \Nats}$ s.t.

\begin{enumerate}[(i)]

\item ${\T_r, \T_s \in \GammaPoly^n}$

\item ${r}$ and ${s}$ are steadily growing.

\item ${\forall K \in \Nats^n: 1 \leq r(K) \leq s(K)}$

\end{enumerate}

Denote ${\Gamma_\text{det}:=(\Gamma_0^n,\Gamma_0^n)}$, ${\Gamma_{\textnormal{red}}:=(\Gamma_{r},\Gamma_0^n)}$, ${\Gamma_{\textnormal{smp}}:=(\Gamma_{s},\Gamma_0^n)}$.

We will show that certain classes of functions paired with ${\Fall(\Gamma_{\text{smp}})}$-samplable word ensembles have a distributional estimation problem which is complete w.r.t. precise pseudo-invertible ${\Fall(\Gamma_{\text{red}})}$-reductions. This construction is an adaption of the standard construction of a complete problem for ${\mathsf{SampNP}}$, as provided in Theorem 10.25 of \cite{Goldreich_2008}.

Due to the large number of variables and functions in the following theorem, some intuitive exposition of the result seems helpful. A universal function $\mathfrak{F}$ will be considered, which takes three inputs. There is an element $\phi$ of some encoded set $E$ that tells $\mathfrak{F}$ which (possibly hard-to-compute) function $f$ to emulate, a time parameter $k$ which controls the computational resources used in emulating $f$, and a bit string $x$ which is just the input to $f$. When $k$ is sufficiently large, and $b$ is chosen appropriately, $\mathfrak{F}(b,k,x)=f(x)$. The distribution $\Dist_{\mathfrak{F}}$ is over 4-element tuples of a bit string (which dictates what $f$ is), the last coordinate of $K$ ($k$), which serves as a time parameter, a bit string $a$ (which can be interpreted as a sampler), and a bit string $x$ (which is the output of the sampler when run for $k$ steps.

For samplable distributions $\Dist$, and functions $f$ which have a corresponding $b$ that makes $\mathfrak{F}$ emulate them, there is a reduction to this universal problem. Observe that if $\Dist$ is samplable, there is some bit string $a$ which encodes a turing machine that samples from $\Dist$. The reduction maps $x$ to the tuple $(b,p(K),a,x)$, and then reindexes. (In particular, to ensure that the time parameter is large enough to fully run the sampler.) 

\begin{samepage}
\begin{theorem}
\label{thm:complete}

Consider an encoded set ${E}$ which is prefix-free, i.e. for all ${\phi,\psi \in E}$ and ${z \in \Bool^{>0}}$, ${\En_E(\phi) \ne \En_E(\psi)z}$. Consider ${\mathfrak{F}: E \times \Nats \times \Words \rightarrow \Reals}$ bounded. For any ${K \in \Nats^n}$, define\\ $\zeta^K: \Words^2 \rightarrow \Words^2$ by 

\begin{equation}
\zeta^K(a,w)=(a,\Ev^{K_{n-1}}(a;\En_{\Nats^n}(K),w))
\end{equation}

Define the distributional estimation problem $({\Dist_{\mathfrak{F}}},f_{\mathfrak{F}})$ by 

\begin{align}
\Dist_{\mathfrak{F}}^K&:=\En_*^4 (\Un^{r(K)} \times \En_{\Nats*} \delta_{K_{n-1}} \times \zeta_*^k(\Un^{r(K)} \times \Un^{s(K)})) \\
f_{\mathfrak{F}}(\Chev{b,\En_\Nats(k),a,x}) &:= \begin{cases}\mathfrak{F}(\phi,k,x) \textnormal{ if } \exists z \in \Words: b=\En_E(\phi)z \\ 0 \textnormal{ if } \forall \phi \in E, z \in \Words: b \ne \En_E(\phi)z\end{cases}
\end{align}

For any ${p \in \NatPoly}$, define ${\alpha_p: \NatFun \Nats^n}$ by

\begin{equation}
\forall J \in \Nats^{n-1}, k \in \Nats: \alpha_p(J,k) = (J,p(J,k))
\end{equation}

Consider a distributional estimation problem ${(\Dist,f)}$ s.t. ${\Dist}$ is ${\Fall(\Gamma_{\textnormal{smp}})}$-samplable and there are ${\phi \in E}$ and ${q \in \Nats[k]}$ s.t. for any ${x \in \Supp \Dist}$ and ${k \geq q(\Abs{x})}$, ${f(x)=\mathfrak{F}(\phi,k,x)}$. Then, there is a precise pseudo-invertible ${\Fall(\Gamma_{\textnormal{red}})}$-reduction from ${(\Dist,f)}$ to ${(\Dist_{\mathfrak{F}}, f_{\mathfrak{F}})}$ over ${\alpha_p}$ for some\\ $p \in \NatPoly$ increasing in the last argument (it is easy to see that any such ${\alpha_p}$ is an efficient injection).

\end{theorem}
\end{samepage}

\begin{proof}

Let ${\sigma}$ be an ${\Fall(\Gamma_{\text{smp}})}$-sampler of ${\Dist}$. Denote ${b=\En_E(\phi)}$. Choose ${p \in \NatPoly}$ increasing in the last argument and ${a \in \Words}$ s.t. for any ${K \in \Nats^{n}}$, $z \in \Words$, ${w_1 \in \Bool^{\R_\sigma(K)}}$ and ${w_2 \in \Words}$: ${p(K) \geq q(\max_{x \in \Supp \sigma_\bullet^K} \Abs{x})}$, ${r(\alpha_p(K)) \geq \Abs{b}}$, ${r(\alpha_p(K)) \geq \Abs{a}}$, ${s(\alpha_p(K)) \geq \R_\sigma(K)}$ and 

\[\Ev^{p(K)}(az;\En_{\Nats^n}(\alpha_p(K)),w_1 w_2)=\sigma^{K}(w_1)\]

The latter is possible because ${\alpha_p}$ can be efficiently inverted using binary search over ${K_{n-1}}$.

Denote ${r_p := r \circ \alpha_p}$. Note that ${\Gamma_{\textnormal{red}} \alpha_p = \Gamma_{\textnormal{red}}}$ by Proposition~\ref{prp:gamma_r_alpha_p}. We construct\\ $\pi: \Words \xrightarrow{\Gamma_{\textnormal{red}}} \Words$ s.t. for any ${K \in \Nats^{n}}$, ${x \in \Supp \sigma_{\bullet}^{K}}$, ${z_b \in \Bool^{r_p(K)-\Abs{b}}}$ and ${z_a \in \Bool^{r_p(K)-\Abs{a}}}$

\begin{align}
\R_\pi(K) &= 2r_p(K)-\Abs{a}-\Abs{b} \\
\label{eqn:thm__complete__red}\pi^{K}(x,z_b z_a)&=\Chev{bz_b,\En_\Nats(p(K)),a z_a, x} 
\end{align}

We also ensure that for any ${K \in \Nats^n}$, ${x \in \Words}$ and ${z_b,z_a}$ as above, either \ref{eqn:thm__complete__red} holds or 

\[\pi^K(x,z_b z_a)=\Estr\]

To verify condition~\ref{con:def__pp_reduce__dist} of Definition~\ref{def:pp_reduce} (with ${\alpha_p}$ playing the role of the efficient injection), fix ${h \in \NatPoly}$ s.t. $h \geq r_p$ and $\Supp \sigma_\bullet^K \subseteq \Bool^{h(K)}$. Construct ${W: \Words \xrightarrow{\Gamma_{\textnormal{det}}} \Rats^{\geq 0}}$ s.t. 

\[W^K(y)=\begin{cases}2^{\Abs{a}+\Abs{b}} \text{ if } \exists z_b,z_a,x \in \Bool^{\leq h(K)}: y=\Chev{bz_b,\En_\Nats(p(K)), az_a,x}\\0 \text{ otherwise}\end{cases}\]

${\Dist^K \equiv \sigma_\bullet^K \pmod{\Fall}}$ since ${\sigma}$ is an ${\Fall(\Gamma_{\textnormal{smp}})}$-sampler of ${\Dist}$. By Proposition~\ref{prp:prob_cong_push}

\[\pi_*^K\Dist^K \equiv \pi_*^K\sigma_\bullet^K \pmod{\Fall}\]

It follows that

\[\sum_{y \in \Words} \Abs{\Dist_{\mathfrak{F}}^{\alpha_p(K)}(y)W^K(y)-(\pi_*^K\Dist^K)(y)} \equiv \sum_{y \in \Words} \Abs{\Dist_{\mathfrak{F}}^{\alpha_p(K)}(y)W^K(y)-(\pi_*^K\sigma_\bullet^K)(y)} \pmod{\Fall}\]

For any ${y \in \Words}$, if ${W^K(y)=0}$ then ${(\pi_*^K\sigma_\bullet^K)(y) = 0}$, so the corresponding terms contribute nothing to the sum on the right hand side. Denote ${\bar{\pi}^K(x,z_b,z_a):=\Chev{bz_b,\En_\Nats(p(K)),az_a,x}}$.

\[\sum_{\Words} \Abs{\Dist_{\mathfrak{F}}^{\alpha_p(K)}W^K-(\pi_*^K\Dist^K)} \equiv \sum_{\substack{z_b \in \Bool^{\leq h(K)}\\z_a \in \Bool^{\leq h(K)}\\x \in \Bool^{\leq h(K)}}} \Abs{\Dist_{\mathfrak{F}}^{\alpha_p(K)}(\bar{\pi}^K(x,z_b,z_a))2^{\Abs{a}+\Abs{b}}-(\pi_*^K\sigma_\bullet^K)(\bar{\pi}^K(x,z_b,z_a))} \pmod{\Fall}\]

\[\sum_{\Words} \Abs{\Dist_{\mathfrak{F}}^{\alpha_p(K)}W^K-\pi_*^K\Dist^K} \equiv \sum_{\substack{z_b \in \Bool^{r_p(K)-\Abs{b}}\\z_a \in \Bool^{r_p(K)-\Abs{a}}\\x \in \Bool^{\leq h(K)}}} \Abs{2^{-r_p(K)}2^{-r_p(K)}\sigma_\bullet^K(x) 2^{\Abs{a}+\Abs{b}}-(\pi_*^K\sigma_\bullet^K)(\bar{\pi}^K(x,z_b,z_a))} \pmod{\Fall}\]

\[\sum_{\Words} \Abs{\Dist_{\mathfrak{F}}^{\alpha_p(K)}W^K-\pi_*^K\Dist^K} \equiv \sum_{\substack{z_1 \in \Bool^{r_p(K)-\Abs{a}}\\z_2 \in \Bool^{r_p(K)-\Abs{b}}\\x \in \Bool^{\leq h(K)}}} \Abs{2^{-2r_p(K)+\Abs{a}+\Abs{b}}\sigma_\bullet^K(x) -(\pi_*^K\sigma_\bullet^K)(\bar{\pi}^K(x,z_b,z_a))} \pmod{\Fall}\]

\begin{align*}
&\sum_{\Words} \Abs{\Dist_{\mathfrak{F}}^{\alpha_p(K)}W^K-\pi_*^K\Dist^K} \equiv\\
&\sum_{\substack{z_1 \in \Bool^{r_p(K)-\Abs{a}}\\z_2 \in \Bool^{r_p(K)-\Abs{b}}\\x \in \Bool^{\leq h(K)}}} \Abs{2^{-2r_p(K)+\Abs{a}+\Abs{b}}\sigma_\bullet^K(x) -2^{-(r_p(K)-\Abs{a})}2^{-(r_p(K)-\Abs{b})}\sigma_\bullet^K(x)} \pmod{\Fall} 
\end{align*}

\[\sum_{\Words} \Abs{\Dist_{\mathfrak{F}}^{p(K)}W^K-\pi_*^K\Dist^K} \equiv 0 \pmod{\Fall}\]

To verify condition~\ref{con:def__pp_reduce__fun} of Definition~\ref{def:pp_reduce}, use Proposition~\ref{prp:prob_cong_ev} to get

\[\E_{\Dist^K \times \Un_\pi^K}[\Abs{f(x)-f_{\mathfrak{F}}(\pi^K(x,z))}] \equiv \E_{\sigma_\bullet^K \times \Un_\pi^K}[\Abs{f(x)-f_{\mathfrak{F}}(\pi^K(x,z))}] \pmod{\Fall}\]

\[\E_{\Dist^K \times \Un_\pi^K}[\Abs{f(x)-f_{\mathfrak{F}}(\pi^K(x,z))}] \equiv \E_{\sigma_\bullet^K \times \Un_\pi^K}[\Abs{f(x)-f_{\mathfrak{F}}(\Chev{bz_b,\En_\Nats(p(K)),az_a,x})}] \pmod{\Fall}\]

\[\E_{\Dist^K \times \Un_\pi^K}[\Abs{f(x)-f_{\mathfrak{F}}(\pi^K(x,z))}] \equiv \E_{\sigma_\bullet^K \times \Un_\pi^K}[\Abs{\mathfrak{F}(\phi,p(K),x)-\mathfrak{F}(\phi,p(K),x)}] \pmod{\Fall}\]

\[\E_{\Dist^K \times \Un_\pi^K}[\Abs{f(x)-f_{\mathfrak{F}}(\pi^K(x,z))}] \equiv 0 \pmod{\Fall}\]

To verify condition~\ref{con:def__pp_reduce__smp} of Definition~\ref{def:pp_reduce}, construct ${\tau: \Words \xrightarrow{\Gamma_{\textnormal{det}}} \Words}$ s.t. for any\\ $z_1,z_2 \in \Bool^{r_p(K)}$ and ${x \in \Supp \sigma_\bullet^K}$, ${\tau^K(\Chev{z_1,\En_\Nats(p(K)),z_2,x})=x}$. By Proposition~\ref{prp:prob_cong_push} and Proposition~\ref{prp:prob_cong_ev}

\[\E_{y \sim \pi_*^K\Dist^K}[\Dtv(\Dist^K \mid (\pi^K)^{-1}(y),\tau_y^K)] \equiv \E_{y \sim \pi_*^K\sigma_\bullet^K}[\Dtv(\Dist^K \mid (\pi^K)^{-1}(y),\tau_y^K)] \pmod{\Fall}\]

Denoting ${\Un_{ba}^K:=\Un^{r_p(K)-\Abs{b}} \times \Un^{r_p(K)-\Abs{a}}}$

\[\E[\Dtv(\Dist^K \mid (\pi^K)^{-1}(y),\tau_y^K)] \equiv \E_{(z_b,z_a,x) \sim \Un_{ba}^K \times \sigma_\bullet^K}[\Dtv(\Dist^K \mid (\pi^K)^{-1}(\bar{\pi}^K(x,z_b,z_a)),\tau_{\bar{\pi}^K(x,z_b,z_a)}^K)] \pmod{\Fall}\]

\[\E[\Dtv(\Dist^K \mid (\pi^K)^{-1}(y),\tau_y^K)] \equiv \E_{(z_b,z_a,x) \sim \Un_{ba}^K \times \sigma_\bullet^K}[\Dtv(\delta_x,\delta_x)] \pmod{\Fall}\]

\[\E[\Dtv(\Dist^K \mid (\pi^K)^{-1}(y),\tau_y^K)] \equiv 0 \pmod{\Fall}\]
\end{proof}

Denote ${\textsc{X}_\mathfrak{F}}$ the set of bounded functions ${f: D \rightarrow \Reals}$ (where ${D \subseteq \Words}$) satisfying the conditions of Theorem~\ref{thm:complete}, and ${\textsc{SampX}_\mathfrak{F}[\Fall(\Gamma_{\text{smp}})]}$ the set of distributional estimation problems of the form ${(\Dist,f)}$ for ${\Fall(\Gamma_{\text{smp}})}$-samplable ${\Dist}$ and ${f \in \textsc{X}_\mathfrak{F}}$. Obviously ${\Dist_{\mathfrak{F}}}$ is ${\Fall(\Gamma_{\text{smp}})}$-samplable. Therefore, if ${f_{\mathfrak{F}} \in \textsc{X}_\mathfrak{F}}$ then ${(\Dist_{\mathfrak{F}},f_{\mathfrak{F}})}$ is complete for ${\textsc{SampX}_\mathfrak{F}[\Fall(\Gamma_{\text{smp}})]}$ w.r.t. precise pseudo-invertible ${\Fall(\Gamma_{\text{red}})}$-reductions over efficient injections of the form ${\alpha_p}$.

\begin{samepage}
\begin{example}

${n = 1}$. ${E_{\textsc{NP}} \subseteq \Words}$ is the set of valid programs for the universal machine ${\mathcal{U}_2}$. ${\mathfrak{F}_{\textsc{NP}}}$ is given by 

\begin{equation}
\mathfrak{F}_{\textsc{NP}}(\phi,k,x):=\begin{cases}1 \text{ if } \exists y\in \Bool^k: \Ev^k(\phi;x,y)=1 \\ 0 \text { otherwise} \end{cases}
\end{equation}

\end{example}
\end{samepage}

\begin{samepage}
\begin{example}

${n = 1}$. ${E_{\textsc{EXP}} \subseteq \Words}$ is the set of valid programs for the universal machine ${\mathcal{U}_1}$. ${\mathfrak{F}_{\textsc{EXP}}}$ is given by 

\begin{equation}
\mathfrak{F}_{\textsc{EXP}}(\phi,k,x):=\begin{cases}1 \text{ if } \Ev^{2^k}(\phi;x)=1 \\ 0 \text { otherwise} \end{cases}
\end{equation}

\end{example}
\end{samepage}

This completeness property implies that, under certain assumptions, optimal polynomial-time estimators exist for all problems in ${\textsc{SampX}_\mathfrak{F}[\Fall(\Gamma_{\text{smp}})]}$ if an optimal polynomial-time estimator exists for ${(\Dist_{\mathfrak{F}},f_{\mathfrak{F}})}$. More precisely and slightly more generally, we have the following corollaries. For the remainder of the section, fix ${m \in \Nats}$ s.t. ${m \geq n}$. For any ${p \in \NatPoly}$, define ${\beta_p: \Nats^m \rightarrow \Nats^m}$ by 

\begin{equation}
\forall J \in \Nats^{n-1}, k \in \Nats, L \in \Nats^{m-n}: \beta_p(J,k,L)=(J,p(J,k),L)
\end{equation}

Define ${\eta: \Nats^{m} \rightarrow \Nats^n}$ by

\begin{equation}
\forall K \in \Nats^n, L \in \Nats^{m-n}: \eta(K,L)=K
\end{equation}

\begin{samepage}
\begin{corollary}
\label{crl:complete_sharp}

Fix ${\Fall^{(m)}}$ a fall space of rank ${m}$ and ${\Gamma^m=(\GrowR^m, \GrowA^m)}$ growth spaces of rank ${m}$. Assume that ${\Fall \eta \subseteq \Fall^{(m)}}$, ${\Gamma_r \eta \subseteq \GrowR^{m}}$ and for any ${p \in \NatPoly}$ increasing in the last argument, ${\Fall^{(m)} \beta_p \subseteq \Fall^{(m)}}$, ${\GrowR^m \beta_p = \GrowR^m}$ and ${\GrowA^m \beta_p = \GrowA^m}$. In the setting of Theorem~\ref{thm:complete}, assume there is an ${\Fall^{(m)\sharp}(\Gamma^m)}$-optimal estimator for ${(\Dist_{\mathfrak{F}}^\eta,f_{\mathfrak{F}})}$. Then, for any ${(\Dist,f) \in \textsc{SampX}_\mathfrak{F}[\Fall(\Gamma_{\text{smp}})]}$ there is an\\ $\Fall^{(m)\sharp}(\Gamma^m)$-optimal estimator for ${(\Dist^\eta,f)}$.

\end{corollary}
\end{samepage}

\begin{proof}

According to Theorem~\ref{thm:complete}, there is ${\pi}$ a precise pseudo-invertible ${\Fall(\Gamma_{\textnormal{red}})}$-reduction of ${(\Dist,f)}$ to ${(\Dist_{\mathfrak{F}},f_{\mathfrak{F}})}$ over ${\alpha_p}$ for some ${p \in \NatPoly}$  increasing in the last argument. This implies ${\pi^\eta}$ is a precise pseudo-invertible ${\Fall^{(m)}(\Gamma^m)}$-reduction of ${(\Dist^\eta,f)}$ to ${(\Dist_{\mathfrak{F}}^\eta,f_{\mathfrak{F}})}$ over ${\beta_p}$. Applying Corollary~\ref{crl:pp_reduce_sharp}, we get the desired result. 
\end{proof}

\begin{samepage}
\begin{corollary}
\label{crl:complete}

Fix ${\Fall^{(m)}}$ a fall space of rank ${m}$ and ${\Gamma^m=(\GrowR^m, \GrowA^m)}$ growth spaces of rank ${m}$ s.t. ${\Fall^{(m)}}$ is ${\GrowA^m}$-ample. Assume that ${\Fall \eta \subseteq \Fall^{(m)}}$, ${\Gamma_r \eta \subseteq \GrowR^{m}}$ and for any ${p \in \NatPoly}$ increasing in the last argument, ${\Fall^{(m)} \beta_p \subseteq \Fall^{(m)}}$, ${\GrowR^m \beta_p = \GrowR^m}$ and ${\GrowA^m \beta_p = \GrowA^m}$. In the setting of Theorem~\ref{thm:complete}, assume there is an ${\Fall^{(m)}(\Gamma^m)}$-optimal estimator for ${(\Dist_{\mathfrak{F}}^\eta,f_{\mathfrak{F}})}$. Then, for any ${(\Dist,f) \in \textsc{SampX}_\mathfrak{F}[\Fall(\Gamma_{\text{smp}})]}$ there is an ${\Fall^{(m)}(\Gamma^m)}$-optimal estimator for ${(\Dist^\eta,f_\phi)}$.
\end{corollary}
\end{samepage}

\begin{proof}

Completely analogous to proof of Corollary~\ref{crl:complete_sharp}.
\end{proof}

In particular, the conditions of Corollary~\ref{crl:complete_sharp} and Corollary~\ref{crl:complete} can hold for ${\Fall=\Fall_\zeta}$ (the fall space of functions which are $O(\zeta)$; see Example~\ref{exm:fall_zeta}) and ${\Fall^{(m)}=\FallUt{\varphi}}$ (see Example~\ref{exm:e_uni}):

\begin{samepage}
\begin{proposition}

Consider ${\varphi: \Nats^n \rightarrow \Nats}$ non-decreasing in the last argument s.t. ${\varphi \geq 3}$. Define ${\zeta: \Nats^n \rightarrow \Reals}$ by

\begin{equation}
\zeta(K):=\frac{\log \log (3+\sum_{i \in [n]} K_i)}{\log \log \varphi(K)}
\end{equation}

Assume ${\zeta}$ is bounded and there is ${h \in \NatPoly}$ s.t. ${\zeta \geq 2^{-h}}$. Let ${m = n + 1}$. Then, ${\Fall_\zeta \eta \subseteq \FallUt{\varphi}}$ and for any ${p \in \NatPoly}$ increasing in the last argument, ${\FallUt{\varphi} \beta_p \subseteq \FallUt{\varphi}}$.

\end{proposition}
\end{samepage}

\begin{proof}

Consider any ${\varepsilon_0 \in \Fall_\zeta}$.

\[\sum_{l=2}^{\varphi(K)-1} \frac{\varepsilon_0(K)}{l \log l} \leq \frac{3}{2} (\log 3) \varepsilon_0(K) \int_2^{\varphi(K)} \frac{\dif t}{t \log t}\]

\[\sum_{l=2}^{\varphi(K)-1} \frac{\varepsilon_0(K)}{l \log l} \leq \frac{3}{2} (\log 3) (\ln 2)^2 \varepsilon_0(K) \log \log \varphi(K)\]

For some ${M_0 \in \Reals^{>0}}$, ${\varepsilon_0 \leq M_0 \zeta}$, therefore

\[\sum_{l=2}^{\varphi(K)-1} \frac{\varepsilon_0(K)}{l \log l} \leq \frac{3}{2} (\log 3)(\ln 2)^2 M_0 \zeta(K) \log \log \varphi(K)\]

\[\sum_{l=2}^{\varphi(K)-1} \frac{\varepsilon_0(K)}{l \log l} \leq \frac{3}{2} (\log 3)(\ln 2)^2 M_0 \log \log (3+\sum_{i \in [n]} K_i)\]

We got ${\varepsilon_0 \circ \eta \in \FallUt{\varphi}}$. Now, consider any ${\varepsilon_1 \in \FallUt{\varphi}}$ and ${p \in \NatPoly}$ increasing in the last argument. Clearly, ${p(K) \geq K_{n-1}}$.

\[\sum_{l=2}^{\varphi(J,k)-1} \frac{\varepsilon_1(J,p(J,k),l)}{l \log l} \leq \sum_{l=2}^{\varphi(J,p(J,k))-1} \frac{\varepsilon_1(J,p(J,k),l)}{l \log l}\]

For some ${M_1 \in \Reals^{>0}}$ and ${q \in \NatPoly}$

\[\sum_{l=2}^{\varphi(J,k)-1} \frac{\varepsilon_1(J,p(J,k),l)}{l \log l} \leq M_1 \log \log q(J,p(J,k))\]

We got ${\varepsilon_1 \circ \beta_p \in \FallUt{\varphi}}$.
\end{proof}

\section{Existence and Uniqueness}
\label{sec:e_and_u}

\subsection{Existence}

\subsubsection{Positive Results}

We give two existence theorems for ${\FallU(\Gamma)}$-optimal estimators (the fall space $\FallU$ was defined in Example~\ref{exm:e_uni}). Theorem~\ref{thm:exists_all} shows that, for appropriate steadily growing functions ${r}$ and ${l}$, \emph{all} distributional estimation problems of rank ${n-1}$ admit ${\FallU(\Gamma_r,\Gamma_l)}$-optimal estimators when trivially extended to rank ${n}$. The extra parameter serves to control the resources available to the estimator. To illustrate its significance using the informal\footnote{Strictly speaking, this example cannot be formalized in the framework as presented here since the set of prime numbers is in $\textsc{P}$. We can tackle it by e.g. taking ${\textsc{NC}}$ instead of ${\textsc{P}}$ as the permissible time complexity for our estimators, but we don't explore this variant in the present work.} example from the introduction, observe that the question \enquote{what is the probability 7614829 is prime?} should depend on the amount of available time. For example, we can use additional time to test for divisibility by additional smaller primes (or in some more clever way) until eventually we are able to test primality and assign a probability in $\{0,1\}$. 

However, in general the estimators constructed in Theorem~\ref{thm:exists_all} are non-uniform because they rely on the advice string to emulate the $Q$ with the lowest Brier score. Theorem~\ref{thm:exists_smp} shows that, under certain stronger assumptions on ${r}$ and ${l}$, for \emph{samplable} distributional estimation problems there is an estimator which requires only as much advice as the sampler. In particular, the existence of a uniform sampler implies the existence of a uniform ${\FallU(\Gamma_r,\Gamma_l)}$-optimal estimator.

We will use the notation ${\eta: \NatFun \Nats^{n-1}}$ defined by

\[\forall J \in \Nats^{n-1}, k \in \Nats: \eta(J,k)=J\]

\begin{samepage}
\begin{theorem}
\label{thm:exists_all}

Fix ${l: \NatFun \Nats^{>0}}$ steadily growing. Denote ${\Gamma_{\textnormal{adv}}^n:=(\Gamma_0^n,\Gamma_l)}$. Fix ${r: \bm{1} \xrightarrow{\Gamma_{\textnormal{adv}}}\Nats}$ steadily growing. Assume ${\GrowR=\Gamma_r}$, ${\GrowA=\Gamma_l}$. Consider ${(\Dist,f)}$ a distributional estimation problem of rank ${n-1}$. Then, there exists an ${\FallU(\Gamma)}$-optimal estimator for ${(\Dist^\eta,f)}$.

\end{theorem}
\end{samepage}

\begin{samepage}

The following two propositions approximately state that, given an arbitrary function $\zeta(J,k)$ for which polynomial increases in $k$ lead to a decrease in $\zeta$, the difference between $\zeta(J,k)$ and some average of values after $\zeta(J,q(J,k))$ lies in $\FallU$. Roughly, this occurs because either $\zeta$ falls quickly enough that, in the asymptotic tail, the values approximately vanish, or $\zeta$ falls slowly enough that, going polynomially further out doesn't change $\zeta$ very much.

\begin{proposition}
\label{prp:fall_uni_weak}

For any ${q \in \NatPolyJ}$ s.t. ${q \geq 2}$ there are ${\{\omega_q^K \in \mathcal{P}(\Nats)\}_{K \in \Nats^n}}$ s.t. for any ${\zeta: \NatFun \Reals}$ bounded, if there is a function $\varepsilon\in\FallU$ s.t.

\begin{equation}
\label{eqn:prp__fall_uni_weak__prem}
\forall J \in \Nats^{n-1}, k,k' \in \Nats: k' \geq (k+2)^{\Floor{\log q(J)}}-2 \implies \zeta(J,k') \leq \zeta(J,k)+\varepsilon(J,k)
\end{equation}

then

\begin{equation}
\label{eqn:prp__fall_uni_weak__conc}
\zeta(J,k) \equiv \E_{i \sim \omega_p^{Jk}}[\zeta(J,(k+2)^{\Floor{\log q(J)}}-2+i)] \pmod \FallU
\end{equation}

\end{proposition}
\end{samepage}

\begin{proof}

Take any ${a \in \Nats}$ s.t. ${a \geq 5}$.

\[\int_{t = a}^{a^{\Floor{\log q(J)}}} \dif (\log \log t) = \log \log a^{\Floor{\log q(J)}} - \log \log a\]

\[\int_{t = a}^{a^{\Floor{\log q(J)}}} \dif (\log \log t) = \log (\Floor{\log q(J)} \log a) - \log \log a\]

\[\int_{t = a}^{a^{\Floor{\log q(J)}}} \dif (\log \log t) = \log \Floor{\log q(J)} + \log \log a - \log \log a\]

\[\int_{t = a}^{a^{\Floor{\log q(J)}}} \dif (\log \log t) = \log \Floor{\log q(J)}\]

Consider any ${\zeta: \NatFun \Reals}$ bounded.

\[\Abs{\int_{t = a}^{a^{\Floor{\log q(J)}}} \zeta(J,\Floor{t}-2) \dif (\log \log t)} \leq (\sup \Abs{\zeta}) \log \Floor{\log q(J)}\]

In particular

\[\Abs{\int_{t = 2}^{2^{\Floor{\log q(J)}}} \zeta(J,\Floor{t}-2) \dif (\log \log t)} \leq (\sup \Abs{\zeta}) \log \Floor{\log q(J)}\]

Adding the last two inequalities

\[\Abs{\int_{t = 2}^{2^{\Floor{\log q(J)}}} \zeta(J,\Floor{t}-2) \dif (\log \log t)} + \Abs{\int_{t = a}^{a^{\Floor{\log q(J)}}} \zeta(J,\Floor{t}-2) \dif (\log \log t)} \leq 2(\sup \Abs{\zeta}) \log \Floor{\log q(J)}\]

\[\int_{t = 2}^{2^{\Floor{\log q(J)}}} \zeta(J,\Floor{t}-2) \dif (\log \log t) - \int_{t = a}^{a^{\Floor{\log q(J)}}} \zeta(J,\Floor{t}-2) \dif (\log \log t) \leq 2(\sup \Abs{\zeta}) \log \Floor{\log q(J)}\]

\[\int_{t = 2}^{a} \zeta(J,\Floor{t}-2) \dif (\log \log t) - \int_{t=2^{\Floor{\log q(J)}}}^{a^{\Floor{\log q(J)}}} \zeta(J,\Floor{t}-2) \dif (\log \log t) \leq 2(\sup \Abs{\zeta}) \log \Floor{\log q(J)}\]

We have ${\dif (\log \log t^{\Floor{\log q(J)}}) = \dif(\log \log t)}$ therefore we can substitute in the second term on the left hand side and get

\[\int_{t = 2}^{a} \zeta(J,\Floor{t}-2) \dif (\log \log t) - \int_{t=2}^{a} \zeta(J,\Floor{t^{\Floor{\log q(J)}}}-2) \dif (\log \log t) \leq 2(\sup \Abs{\zeta}) \log \Floor{\log q(J)}\]

\[\int_{t = 2}^{a} (\zeta(J,\Floor{t}-2) - \zeta(J,\Floor{t^{\Floor{\log q(J)}}}-2)) \dif (\log \log t) \leq 2(\sup \Abs{\zeta}) \log \Floor{\log q(J)}\]



\[\int_2^a (\zeta(J,\Floor{t}-2) - \zeta(J,\Floor{t^{\Floor{\log q(J)}}}-2)) \frac{\dif t}{(\ln 2)^2 t \log t} \leq 2(\sup \Abs{\zeta}) \log \Floor{\log q(J)}\]

\[\sum_{k=0}^{a-3} \int_{k+2}^{k+3} \frac{\zeta(J,\Floor{t}-2) - \zeta(J,\Floor{t^{\Floor{\log q(J)}}}-2)}{t \log t} \dif t \leq 2(\ln 2)^2(\sup \Abs{\zeta}) \log \Floor{\log q(J)}\]

\[\sum_{k=0}^{a-3} \int_0^1 \frac{\zeta(J,k) - \zeta(J,\Floor{(k+t+2)^{\Floor{\log q(J)}}}-2)}{(k+t+2) \log (k+t+2)}\dif t \leq 2(\ln 2)^2(\sup \Abs{\zeta}) \log \Floor{\log q(J)}\]

For ${k \geq 2}$ we have ${(k + 3) \log (k + 3) \leq \frac{5}{2}k \log \frac{5}{2}k \leq \frac{5}{2}k \log k^{\log 5} = \frac{5}{2} (\log 5) k \log k}$.

\[\sum_{k=2}^{a-3} \frac{\zeta(J,k) - \int_0^1 \zeta(J,\Floor{(k+t+2)^{\Floor{\log q(J)}}}-2)\dif t}{\frac{5}{2} (\log 5) k \log k} \leq 2(\ln 2)^2(\sup \Abs{\zeta}) \log \Floor{\log q(J)}\]

Define

\[I_q^{Jk}(i):=\{t \in [0,1] \mid (k+t+2)^{\Floor{\log q(J)}} - (k+2)^{\Floor{\log q(J)}} \in [i,i+1)\}\]

\[\omega_q^K(i):=\begin{cases}\sup I_q^K-\inf I_q^K \text{ if } I_q^K \ne \varnothing\\0 \text{ otherwise}\end{cases}\]

We get

\[\sum_{k=2}^{a-3} \frac{\zeta(J,k) - \sum_{i=0}^\infty \zeta(J,(k+2)^{\Floor{\log q(J)}}-2+i)\omega_q^{Jk}(i) }{k \log k}  \leq \frac{4}{5}(\ln 2)(\ln 5)(\sup \Abs{\zeta}) \log \Floor{\log q(J)}\]

Denote $M:=\frac{4}{5}(\ln 2)(\ln 5)(\sup \Abs{\zeta})$ and $\bar{\zeta}(J,k):=\sum_{i=0}^\infty \zeta(J,(k+2)^{\Floor{\log q(J)}}-2+i)\omega_q^{Jk}(i)$. Using~\ref{eqn:prp__fall_uni_weak__prem}

\[\zeta(J,k) - \bar{\zeta}(J,k) \geq -\epsilon(J,k)\]

\[\Abs{\zeta(J,k) - \bar{\zeta}(J,k)} \leq \zeta(J,k) - \bar{\zeta}(J,k) + 2\epsilon(J,k)\]

\[\sum_{k=2}^{a-3} \frac{\Abs{\zeta(J,k) - \bar{\zeta}(J,k)}}{k \log k}  \leq M \log \Floor{\log q(J)}+2\sum_{k=2}^{a-3}\frac{\epsilon(J,k)}{k\log{k}} \]

Taking $a$ to infinity and using the fact that $\epsilon\in\FallU$, we get the desired result.
\end{proof}

\begin{samepage}
\begin{proposition}
\label{prp:fall_uni}

For any ${p \in \NatPoly}$ there are ${\{\omega_p^K \in \mathcal{P}(\Nats)\}_{K \in \Nats^n}}$ s.t. for any\\ $\zeta: \NatFun \Reals$ bounded, if there is a function $\varepsilon\in\FallU$ s.t.

\begin{equation}
\label{eqn:prp__fall_uni__prem}
\forall J \in \Nats^{n-1}, k,k' \in \Nats: k' \geq p(J,k) \implies \zeta(J,k') \leq \zeta(J,k) + \varepsilon(J,k)
\end{equation}

then

\begin{equation}
\label{eqn:prp__fall_uni__conc}
\zeta(J,k) \equiv \E_{i \sim \omega_p^{Jk}}[\zeta(J,p(J,k)+i)] \pmod \FallU
\end{equation}

\end{proposition}
\end{samepage}

\begin{proof}

Fix ${p \in \NatPoly}$. Choose ${q \in \NatPolyJ}$ s.t. ${p(J,k) \leq (k+2)^{\Floor{\log q(J)}}}-2$. Let ${\{\omega_q^K \in \mathcal{P}(\Nats)\}_{K \in \Nats^n}}$ be as in Proposition~\ref{prp:fall_uni_weak}. Define ${\{\omega_p^K \in \mathcal{P}(\Nats)\}_{K \in \Nats^n}}$ by

\[\Prb_{i \sim \omega_p^{Jk}}[i \geq k] = \Prb_{i \sim \omega_q^{Jk}}[i + (k+2)^{\Floor{\log q(J)}} - 2 - p(J,k) \geq k]\]

Suppose ${\zeta: \NatFun \Reals}$ is bounded and s.t. \ref{eqn:prp__fall_uni__prem} holds. In particular, \ref{eqn:prp__fall_uni_weak__prem} also holds. Therefore, we have \ref{eqn:prp__fall_uni_weak__conc}. We rewrite it as follows

\[\zeta(J,k) \equiv \E_{i \sim \omega_q^{Jk}}[\zeta(J, p(J,k) + i+ (k+2)^{\Floor{\log q(J)}} - 2 - p(J,k))] \pmod \FallU\]

By definition of ${\omega_p}$, \ref{eqn:prp__fall_uni__conc} follows.
\end{proof}

In the following, we use the notation ${\alpha_{p}(J,k):=(J,p(J,k))}$.

\begin{samepage}
\begin{proposition}
\label{prp:fall_uni_amp}

Consider ${p \in \NatPoly}$, ${(\Dist,f)}$ a distributional estimation problem and\\ ${P,Q: \Words \Scheme \Rats}$ bounded. Suppose that

\begin{align}
\label{eqn:prp__fall_uni_amp__p}\sup_{i \in \Nats} \E_{ \Dist^{\alpha_{p+i}(K)} \times \Un_P^{\alpha_{p+i}(K)}}[(P^{\alpha_{p+i}(K)}-f)^2] &\leq \E_{ \Dist^{K} \times \Un_P^{K}}[(P^{K}-f)^2] \pmod\FallU \\
\label{eqn:prp__fall_uni_amp__q}\sup_{i \in \Nats} \E_{\Dist^{\alpha_{p+i}(K)} \times \Un_P^{\alpha_{p+i}(K)}}[(P^{\alpha_{p+i}(K)}-f)^2] &\leq \E_{\Dist^{K} \times \Un_Q^{K}}[(Q^{K}-f)^2] \pmod\FallU
\end{align}

Then

\begin{equation}
\label{eqn:prp__fall_uni_amp}
\E_{\Dist^{K} \times \Un_P^{K}}[(P^{K}-f)^2] \leq \E_{\Dist^{K} \times \Un_Q^{K}}[(Q^{K}-f)^2] \pmod \FallU
\end{equation}

\end{proposition}
\end{samepage}

\begin{proof}

Define ${\zeta(K):=\E_{\Dist^{K} \times \Un_P^{K}}[(P^{K}(x,y)-f(x))^2]}$ and observe that \ref{eqn:prp__fall_uni_amp__p} implies \ref{eqn:prp__fall_uni__prem}, allowing us to apply Proposition~\ref{prp:fall_uni} and get

\[\E_{\Dist^{K} \times \Un_P^{K}}[(P^{K}(x,y)-f(x))^2] \equiv \E_{\omega_p^K}[\E_{ \Dist^{\alpha_{p+i}(K)} \times \Un_P^{\alpha_{p+i}(K)}}[(P^{\alpha_{p+i}(K)}(x,y)-f(x))^2]] \pmod \FallU\]

Applying \ref{eqn:prp__fall_uni_amp__q} to the right hand side, we get \ref{eqn:prp__fall_uni_amp}.
\end{proof}

\begin{proof}[Proof of Theorem \ref{thm:exists_all}]

Fix $M \geq \sup \Abs{f}$ and construct ${D: \Words \Alg \Rats}$ s.t.

\[D(x)=\begin{cases}D(x)=\max(\min(t,+M),-M) \text{ if ${x = \En_\Rats(t)}$}\\D(x)=0 \text{ if } x \not\in \Img \En_\Rats\end{cases}\]

Choose ${a^*: \Nats^n \rightarrow \Words}$ s.t.

\begin{equation}
\label{eqn:thm__exists_all__astar}
a^*(K) \in \Argmin{a \in \Bool^{\leq l(K)}} \E_{\Dist^{\eta(K)} \times \Un^{r(K)}}[(D(\Ev^{K_{n-1}}(a;x,y))-f(x))^2]
\end{equation} 

Construct ${P: \Words \Scheme \Rats}$ s.t. for any ${K \in \Nats^n}$, ${x,y,b_0 \in \Words}$ and ${a_0 \in \Bool^{\leq l(K)}}$

\begin{align}
\A_P(K) &= \Chev{a^*(K),\A_r(K)}\\
\R_P(K,\Chev{a_0,b_0}) &= r(K,b_0) \\
P^K(x,y,\Chev{a_0,b_0}) &= D(\Ev^{K_{n-1}}(a_0;x,y))
\end{align}

Consider ${Q: \Words \Scheme \Rats}$ bounded. Without loss of generality we can assume $\sup \Abs{Q} \leq M$ (otherwise we can replace ${Q}$ by ${\tilde{Q}:=\max(\min(Q,+M),-M)}$ and have ${\E[(\tilde{Q}-f)^2]\leq\E[(Q-f)^2]}$). Choose ${q \in \NatPoly}$ s.t. for any ${K \in \Nats^n}$ there exists ${a_Q^K \in \Bool^{l(\alpha_q(K))}}$ for which

\begin{align}
\label{eqn:thm__exists_all__q_rnd}\R_Q(K) &\leq r(\alpha_q(K)) \\
\label{eqn:thm__exists_all__q_adv}\forall i \in \Nats, x,z \in \Words, y \in \BoolR{Q}: D(\Ev^{q(K)+i}(a_Q^K;x,yz))&=Q^K(x,y)
\end{align}

Take any ${i \in \Nats}$.

\[\E_{ \Dist^{\eta(K)} \times \Un_P^{\alpha_{q+i}(K)}}[(P^{\alpha_{q+i}(K)}(x,y)-f(x))^2]=\E_{ \Dist^{\eta(K)} \times \Un^{r(\alpha_{q+i}(K))}}[(D(\Ev^{q(K)+i}(a^*(\alpha_{q+i}(K));x,y))-f(x))^2]\]

Using \ref{eqn:thm__exists_all__astar}

\[\E_{ \Dist^{\eta(K)} \times \Un_P^{\alpha_{q+i}(K)}}[(P^{\alpha_{q+i}(K)}(x,y)-f(x))^2] \leq \E_{\Dist^{\eta(K)} \times \Un^{r(\alpha_{q+i}(K))}}[(D(\Ev^{q(K)+i}(a_Q^K;x,y))-f(x))^2]\]

\[\E_{ \Dist^{\eta(K)} \times \Un_P^{\alpha_{q+i}(K)}}[(P^{\alpha_{q+i}(K)}(x,y)-f(x))^2] \leq \E_{\Dist^{\eta(K)} \times \Un_Q^K}[(Q^K(x,y)-f(x))^2]\]

By the same reasoning we can choose ${p \in \NatPoly}$ s.t. ${p \geq q}$ and

\[\E_{ \Dist^{\eta(K)} \times \Un_P^{\alpha_{p+i}(K)}}[(P^{\alpha_{p+i}(K)}(x,y)-f(x))^2] \leq \E_{\Dist^{\eta(K)} \times \Un_P^K}[(P^K(x,y)-f(x))^2]\]

Applying Proposition~\ref{prp:fall_uni_amp}, we conclude that ${P}$ is an ${\FallU(\Gamma)}$-optimal estimator for ${(\Dist^\eta,f)}$.
\end{proof}

We now proceed to study the special case of samplable problems. These problems admit an optimal polynomial-time estimator which is essentially a brute-force implementation of the empirical risk minimization principle in statistical learning. In particular, the optimality of this algorithm can be regarded as a manifestation of the fundamental theorem of agnostic PAC learning (see e.g. Theorem 6.7 in \cite{Shalev-Shwartz_2014}). In our case the hypothesis space of the space of programs, so this algorithm can also be regarded as a variation of Levin's universal search. The advantage of this optimal polynomial-time estimator on the fully general construction of Theorem~\ref{thm:exists_all} is that the required advice is only the advice of the sampler. The notation $\FallM$ below refers to the fall space defined in Example~\ref{exm:e_mon}.

\begin{samepage}
\begin{theorem}
\label{thm:exists_smp}

Fix ${r: \Nats^n \Alg \Nats}$ s.t.

\begin{enumerate}[(i)]

\item ${\T_r \in \GammaPoly^n}$

\item As a function, ${r \in \GammaPoly^n}$.

\item ${r}$ is non-decreasing in the last argument.

\item There is ${s \in \NatPoly}$ s.t. ${\forall K \in \Nats^n: \log(K_{n-1}+4)r(K) \leq r(\alpha_s(K))}$.

\end{enumerate}

In particular, ${r}$ is steadily growing. Assume ${\GrowR=\Gamma_r}$ and ${\GrowA=\Gamma_\textnormal{log}^n}$. Consider ${(\Dist,f)}$ an distributional estimation problem of rank ${n-1}$ and ${\sigma}$ an ${\FallM(\Gamma)}$-sampler of ${(\Dist^\eta,f)}$. Then, there exists ${P}$ an ${\FallU(\Gamma)}$-optimal estimator for ${(\Dist^\eta,f)}$ s.t. ${\A_P=\A_\sigma}$. In particular, if ${\sigma}$ is uniform (i.e. ${\A_\sigma \equiv \Estr}$) then so is ${P}$.


\end{theorem}
\end{samepage}

\begin{samepage}
\begin{proposition}
\label{prp:std_grow_log}

Fix ${r \in \GammaPoly^n}$ s.t.

\begin{enumerate}[(i)]

\item ${r}$ is non-decreasing in the last argument.

\item There is ${s \in \NatPoly}$ s.t. ${\forall K \in \Nats^n: \log(K_{n-1}+4)r(K) \leq r(\alpha_s(K))}$.

\end{enumerate}

In particular, ${r}$ is steadily growing. Consider any ${\gamma \in \Gamma_r}$ and define ${\gamma': \Nats \rightarrow \Nats}$ by 

\[\gamma'(K):=\Floor{\log(K_{n-1}+2)}\gamma(K)\]

Then, ${\gamma' \in \Gamma_r}$

\end{proposition}
\end{samepage}

\begin{proof}

Choose ${p \in \NatPoly}$ s.t. ${p(K) \geq K_{n-1}}$ and ${r(\alpha_p(K)) \geq \gamma(K)}$. We get 

\[\gamma'(K) \leq \Floor{\log(K_{n-1}+2)} r(\alpha_p(K))\]

\[\gamma'(K) \leq \Floor{\log(p(K)+4)} r(\alpha_p(K))\]

\[\gamma'(K) \leq r(\alpha_s(\alpha_p(K)))\]
\end{proof}

\begin{samepage}
\begin{proposition}
\label{prp:est_err_smp}

Consider ${(\Dist,f)}$ a distributional estimation problem, ${\sigma}$ an ${\EG}$-sampler of ${(\Dist,f)}$, ${I}$ a set and ${\{h_\alpha^K: \Words \Markov \Reals\}_{\alpha \in I, K \in \Nats^n}}$ uniformly bounded. Then

\begin{equation}
\label{eqn:prp__est_err_smp}
\E_{\Un_\sigma^K}[\E[(h_\alpha^K \circ \sigma^K_0-\sigma^K_1)^2]] \overset{\alpha}{\equiv} \E_{\Dist^K}[\E[(h_\alpha^K-f)^2]] + \E_{\Un_\sigma^K}[(f \circ \sigma^K_0-\sigma^K_1)^2] \pmod \Fall
\end{equation}

\end{proposition}
\end{samepage}

\begin{proof}

Denote ${h_{\sigma\alpha }^K := h_\alpha^K \circ \sigma^K_0}$, ${f_\sigma^K := f \circ \sigma^K_0}$. Proposition~\ref{prp:smp} implies

\[\E_{\Un_\sigma^K}[(\E[h_{\sigma\alpha}^K]-f_\sigma^K)f_\sigma^K] \overset{\alpha}{\equiv} \E_{\Dist^K}[(\E[h_\alpha^K]-f)f] \pmod \Fall\]

Applying Proposition~\ref{prp:gen} to the right hand side

\[\E_{\Un_\sigma^K}[(\E[h_{\sigma\alpha}^K]-f_\sigma^K)f_\sigma^K]] \overset{\alpha}{\equiv} \E_{\Un_\sigma^K}[(\E[h_{\sigma\alpha}^K]-f_\sigma^K)\sigma^K_1] \pmod \Fall\]

\begin{equation}
\label{eqn:prp__est_err_smp__prf}
\E_{\Un_\sigma^K}[(\E[h_{\sigma\alpha}^K]-f_\sigma^K)(f_\sigma^K-\sigma_1^K)]] \overset{\alpha}{\equiv} 0 \pmod \Fall
\end{equation}

On the other hand

\[\E_{\Un_\sigma^K}[\E[(h_{\sigma\alpha}^K-\sigma^K_1)^2]] = \E_{\Un_\sigma^K}[\E[(h_{\sigma\alpha}^K-f_\sigma^K+f_\sigma^K-\sigma^K_1)^2]]\]

\[\E_{\Un_\sigma^K}[\E[(h_{\sigma\alpha}^K-\sigma^K_1)^2]] = \E_{\Un_\sigma^K}[\E[(h_{\sigma\alpha}^K-f_\sigma^K)^2]]+2\E_{\Un_\sigma^K}[(\E[h_{\sigma\alpha}^K]-f_\sigma^K)(f_\sigma^K-\sigma^K_1)]]+\E_{\Un_\sigma^K}[\E[(f_\sigma^K-\sigma^K_1)^2]]\]

Applying Proposition~\ref{prp:smp} to the first term on the right hand side and \ref{eqn:prp__est_err_smp__prf} to the second term on the right hand side, we get \ref{eqn:prp__est_err_smp}.
\end{proof}

\begin{proof}[Proof of Theorem \ref{thm:exists_smp}]

Fix $M \geq \sup \Abs{f}$ and construct ${D: \Words \Alg \Rats}$ s.t.

\[D(x)=\begin{cases}D(x)=\max(\min(t,M),-M) \text{ if ${x = \En_\Rats(t)}$}\\D(x)=0 \text{ if } x \not\in \Img \En_\Rats\end{cases}\]

Denote ${l(K):=\Floor{\log (K_{n-1} + 2)}}$. Denote $s(K):=2\Ceil{M^2}l(K)^2$. Construct ${R: \Words \Scheme \Rats}$ s.t. for any ${K \in \Nats^n}$, $w \in \Words$, ${a \in \Bool^{l(K)}}$, ${\{y_i \in \Bool^{r_\sigma(K,w)}\}_{i \in [s(K)]}}$ and ${\{z_i \in \Bool^{r(K)}\}_{i \in [s(K)]}}$

\begin{align}
\label{eqn:thm__exists_smp__l_adv}\A_R(K) &= \A_\sigma(K) \\
\label{eqn:thm__exists_smp__l_rnd}\R_R(K,w) &= s(K) (\R_\sigma(K,w) + r(K)) \\
\label{eqn:thm__exists_smp__l_alg}R^K\left(a,\prod_{i \in [s(K)]}y_i z_i,w\right) &= \frac{1}{s(K)}\sum_{i \in [s(K)]}(D(\Ev^{K_{n-1}}(a;\sigma^K(y_i,w)_0,z_i))-\sigma^K(y_i,w)_1)^2 
\end{align}

That is, ${R}$ generates ${2\Ceil{M^2}l(K)^2}$ estimates of ${f}$ using ${\sigma}$ and computes the \enquote{empirical risk} of the program ${a}$ w.r.t. these estimates. Here, \ref{eqn:thm__exists_smp__l_rnd} is legitimate due to Proposition~\ref{prp:std_grow_log}. 

Construct ${A: \bm{1} \Scheme \Words}$ s.t. for any ${K \in \Nats^n}$, ${w \in \Words}$, ${\{y_i \in \Bool^{r_\sigma(K,w)}\}_{i \in [s(K)]}}$ and\\ ${\{z_i \in \Bool^{r(K)}\}_{i \in [s(K)]}}$

\begin{align}
\label{eqn:thm__exists_smp__a_adv}\A_A(K) &= \A_\sigma(K) \\
\label{eqn:thm__exists_smp__a_rnd}\R_A(K,w) &= \R_R(K,w) \\
\label{eqn:thm__exists_smp__a_alg}A^K\left(\prod_{i \in [s(K)]}y_i z_i,w\right) &\in \Argmin{a \in \Bool^{\leq l(K)}} R^K\left(a,\prod_{i \in [s(K)]}y_i z_i,w\right)
\end{align}

Finally, construct ${P: \Words \Scheme \Rats}$ s.t. for any ${K \in \Nats^n}$, ${w \in \Words}$, ${\{y_i \in \Bool^{r_\sigma(K,w)}\}_{i \in [s(K)]}}$,\\ ${\{z_i \in \Bool^{r(K)}\}_{i \in [s(K)]}}$ and ${z_* \in \Bool^{r(K)}}$

\begin{align}
\label{eqn:thm__exists_smp__p_adv}\A_P(K) &= \A_\sigma(K) \\
\label{eqn:thm__exists_smp__p_rnd}\R_P(K,w) &= \R_R(K,w) + r(K) \\
\label{eqn:thm__exists_smp__p_alg}P^K(x,\left(\prod_{i \in [s(K)]}y_i z_i\right)z_{*},w) &= D(\Ev^{K_{n-1}}(A^K\left(\prod_{i \in [s(K)]}y_i z_i,w\right);x,z_{*}))
\end{align}

Define ${\varrho_0^K \in \Reals}$ by

\[\varrho_0^K:=\E_{\Un_\sigma^K}[(f(\sigma^K(y)_0)-\sigma^K(y)_1)^2]\]

For any ${b \in \Words}$, define ${\varrho^K(b)}$ by

\[\varrho^K(b):=\E_{\Dist^{\eta(K)} \times \Un^{r(K)}}[(D(\Ev^{K_{n-1}}(b;x,z))-f(x))^2]\]

Consider any ${\alpha: \Nats^n \rightarrow \Words}$ s.t. ${\Abs{\alpha(K)} \leq l(K)}$. Define ${h_\alpha^K: \Words \Markov \Reals}$ by

\[\forall s,t \in \Reals: \Pr[h_\alpha^K(x) \in (s,t)]:=\Prb_{z \sim \Un^{r(K)}}[D(\Ev^{K_{n-1}}(\alpha(K);x,z)) \in (s,t)]\]

By Proposition~\ref{prp:est_err_smp}

\begin{align*}
&\E_{\Un_\sigma^K}[\E[(h_\alpha^K(\sigma^K(y)_0)-\sigma^K(y)_1)^2]] \overset{\alpha}{\equiv}\\ 
&\E_{\Dist^{\eta(K)}}[\E[(h_\alpha^K(x)-f(x))^2]] + \E_{\Un_\sigma^K}[(f(\sigma^K(y)_0)-\sigma^K(y)_1)^2] \pmod \FallM 
\end{align*}

\begin{equation}
\label{eqn:thm__exists_smp__prf1}
\E_{\Un_\sigma^K}[\E[(h_\alpha^K(\sigma^K(y)_0)-\sigma^K(y)_1)^2]] \overset{\alpha}{\equiv} \varrho^K(\alpha(K)) + \varrho_0^K \pmod \FallM
\end{equation}

${R^K(\alpha(K),y)}$ is the average of ${2\Ceil{M^2}l(K)^2}$ independent and and identically distributed bounded random variables. By \ref{eqn:thm__exists_smp__prf1}, there is ${\varepsilon \in \FallM}$ that doesn't depend on ${\alpha}$ s.t. the expected value of these random variables is in ${[\varrho^K(\alpha(K)) + \varrho_0^K - \varepsilon(K), \varrho^K(\alpha(K)) + \varrho_0^K + \varepsilon(K)]}$. Applying Hoeffding's inequality we conclude that

\[\forall b \in \Bool^{\leq l(K)}: \Prb_{\Un_R^K}[R^K(b,y) > \varrho^K(b) + \varrho_0^K  + \varepsilon(K) + l(K)^{-1/2}] \leq 2^{-\log(e) l(K)}\]

In particular, since for any ${b \in \Bool^{l(K)}}$, ${R^K(A^K(y),y) \leq R^K(b,y)}$

\begin{equation}
\label{eqn:thm__exists_smp__prf2}
\forall b \in \Bool^{\leq l(K)}: \Prb_{\Un_R^K}[R^K(A^K(y),y) > \varrho^K(b) + \varrho_0^K + \varepsilon(K) + l(K)^{-1/2}] \leq 2^{-\log(e) l(K)}
\end{equation}

Similarly, we have

\[\forall b \in \Bool^{\leq l(K)}: \Prb_{\Un_R^K}[R^K(b,y) < \varrho^K(b) + \varrho_0^K - \varepsilon(K) - l(K)^{-1/2}] \leq 2^{-\log(e) l(K)}\]

\[\Prb_{\Un_R^K}[\exists b \in \Bool^{\leq l(K)}: R^K(b,y) < \varrho^K(b) + \varrho_0^K - \varepsilon(K) - l(K)^{-1/2}] \leq 2^{-(\log(e)-1) l(K)+1}\]

\begin{equation}
\label{eqn:thm__exists_smp__prf3}
\Prb_{\Un_R^K}[R^K(A^K(y),y) < \varrho^K(A^K(y)) + \varrho_0^K - \varepsilon(K) - l(K)^{-1/2}] \leq 2^{-(\log(e)-1)l(K)+1}
\end{equation}

Combining \ref{eqn:thm__exists_smp__prf2} and \ref{eqn:thm__exists_smp__prf3}, we conclude that for any ${b \in \Bool^{\leq l(K)}}$

\[\Prb_{\Un_R^K}[\varrho^K(A^K(y)) + \varrho_0^K - \varepsilon(K) - l(K)^{-1/2} > \varrho^K(b) + \varrho_0^K + \varepsilon(K) + l(K)^{-1/2}] \leq 2^{-\log(e) l(K)} + 2^{-(\log(e)-1) l(K)+1}\]

\[\Prb_{\Un_R^K}[\varrho^K(A^K(y)) > \varrho^K(b) + 2(\varepsilon(K) + l(K)^{-1/2})] \leq 2^{-(\log(e)-1) l(K)+2}\]

It follows that for some ${M_0 \in \Reals^{>0}}$

\[\E_{\Un_R^K}[\varrho^K(A^K(y)] \leq \varrho^K(b) + 2(\varepsilon(K) + l(K)^{-1/2}) + 2^{-(\log(e)-1)l(K)+2} M_0\]

 Denote ${\varepsilon_1(K):=2(\varepsilon(K) + l(K)^{-1/2}) + 2^{-(\log(e)-1) l(K)+2} M_0}$. Note that ${\varepsilon_1 \in \FallM}$ because $\varepsilon$ is by assumption, the last two terms are monotonically decreasing, and both $\sum_{k=2}^{\infty}\frac{1}{k\log k \sqrt{\Floor{\log(k+2)}}}$ and $\sum_{k=2}^{\infty}\frac{2^{-(\log(e)-1)\Floor{\log(k+2)}}}{k\log k}$ converge.

\[\E_{\Un_R^K}[\E_{\Dist^{\eta(K)} \times \Un^{r(K)}}[(D(\Ev^{K_{n-1}}(A^K(y);x,z))-f(x))^2]] \leq \varrho^K(b) + \varepsilon_1(K)\]

\[\forall b \in \Bool^{\leq l(K)}: \E_{\Dist^{\eta(K)} \times \Un_P^K}[(P^K(x,y)-f(x))^2] \leq \varrho^K(b) + \varepsilon_1(K)\]

Consider ${Q: \Words \Scheme \Rats}$ bounded. Without loss of generality we can assume ${\sup \Abs{Q} \leq M}$. Choose ${q \in \NatPoly}$ s.t. ${q(K) \geq K_{n-1}}$ and for all ${K \in \Nats^n}$, \ref{eqn:thm__exists_all__q_rnd} and \ref{eqn:thm__exists_all__q_adv} hold.

\[\E_{\Dist^{\eta(K)} \times \Un_P^{\alpha_{q+i}(K)}}[(P^{\alpha_{q+i}(K)}(x,y)-f(x))^2] \leq \varrho^{\alpha_{q+i}(K)}(a_Q^K) + \varepsilon_1(\alpha_{q+i}(K))\]

\begin{align*}
&\E_{\Dist^{\eta(K)} \times \Un_P^{\alpha_{q+i}(K)}}[(P^{\alpha_{q+i}(K)}(x,y)-f(x))^2] \leq\\ 
&\E_{\Dist^{\eta(K)} \times \Un^{r(\alpha_{q+i}(K))}}[(D(\Ev^{q(K)+i}(a_Q^K;x,z))-f(x))^2] + \varepsilon_1(\alpha_{q+i}(K))
\end{align*}

\[\E_{\Dist^{\eta(K)} \times \Un_P^{\alpha_{q+i}(K)}}[(P^{\alpha_{q+i}(K)}(x,y)-f(x))^2] \leq \E_{\Dist^{\eta(K)} \times \Un_Q^K}[(Q^K(x,z)-f(x))^2] + \varepsilon_1(\alpha_{q+i}(K))\]

Define ${\bar{\varepsilon}_1(K):=\sup_{k \geq K_{n-1}}\varepsilon_1(\eta(K),k)}$. We have ${\bar{\varepsilon}_1 \in \FallU}$ and ${\varepsilon_1(\alpha_{q+i}(K)) \leq \bar{\varepsilon}_1(K)}$ therefore

\[\sup_{i \in \Nats} \E_{\Dist^{\eta(K)} \times \Un_P^{\alpha_{q+i}(K)}}[(P^{\alpha_{q+i}(K)}(x,y)-f(x))^2] \leq \E_{\Dist^{\eta(K)} \times \Un_Q^K}[(Q^K(x,z)-f(x))^2] \pmod \FallU\]

By the same reasoning we can choose ${p \in \NatPoly}$ s.t. ${p \geq q}$ and

\[\sup_{i \in \Nats} \E_{ \Dist^{\eta(K)} \times \Un_P^{\alpha_{p+i}(K)}}[(P^{\alpha_{p+i}(K)}(x,y)-f(x))^2] \leq \E_{\Dist^{\eta(K)} \times \Un_P^K}[(P^K(x,y)-f(x))^2] \pmod \FallU\]

Applying Proposition~\ref{prp:fall_uni_amp}, we conclude that ${P}$ is an ${\FallU(\Gamma)}$-optimal estimator for ${(\Dist^\eta,f)}$.
\end{proof}

The above existence theorems employ the fall space ${\FallU}$ whose meaning might seem somewhat obscure. To shed some light on this, consider the following observation. Informally, optimal polynomial-time estimators represent \enquote{expected values} corresponding to the uncertainty resulting from bounding computing resources. When a function can be computed in polynomial time, this \enquote{expected value} has to approximate the function within ${\Fall}$ which corresponds to a state of \enquote{complete certainty.} However, we will now demonstrate that when a function can only be computed in \emph{quasi-polynomial} time, it still corresponds to complete certainty in the context of ${\FallU(\Gamma)}$-optimal estimators.

\begin{samepage}
\begin{definition}
\label{def:perfect}

Consider ${(\Dist, f)}$ a distributional estimation problem and ${P: \Words \Scheme \Rats}$ bounded. ${P}$ is called an \emph{${\EG}$-perfect polynomial-time estimator for ${(\Dist,f)}$} when 

\begin{equation}
\E_{(x,y) \sim \Dist^K \times \Un_P^K}[(P^K(x,y)-f(x))^2] \equiv 0 \pmod \Fall
\end{equation}

For the sake of brevity, we will say \enquote{${\EG}$-perfect estimator} rather than \enquote{${\EG}$-perfect polynomial-time estimator.}

\end{definition}
\end{samepage}

Perfect polynomial-time estimators are essentially objects of \enquote{classical} average-case complexity theory. In particular, perfect polynomial-time estimators for distributional decision problems of rank 1 are closely related to heuristic algorithms in the sense of \cite{Bogdanov_2006} (their existence is equivalent under mild assumptions), whereas perfect polynomial-time estimators for rank 2 problems of the form ${(\Dist^\eta,\chi_L)}$ with ${\Dist}$ of rank 1 are related to heuristic schemes.

Comparing the definition of a perfect estimator to the definition of an inapproximable predicate, (Definition 7.9 in \cite{Goldreich_2008}), if $f$ is $(\text{poly},\rho)$-inapproximable, and $\Dist^k=\Un^k$, then for any $\zeta\in o(\rho)$, there is no $\Fall_{\zeta}(\Gamma^{1}_{0},\Gamma^{1}_{\text{poly}})$-perfect estimator for $(\Dist,f)$.

\begin{samepage}
\begin{proposition}
\label{prp:fall_uni_qpoly}
Consider ${(\Dist,f)}$ a distributional estimation problem, ${P: \Words \Scheme \Rats}$ bounded,\\ $m \in \Nats^{>0}$ and ${p \in \NatPolyJ}$ s.t. ${p \geq 2}$. Define ${q: \NatFun \Nats}$ by $q(J,k):=2^{\Floor{\log p(J)\log \max(k,1)}^m}$. Suppose that 

\begin{equation}
\label{eqn:prp__fall_uni_qpoly}
\sup_{i \in \Nats} \E_{(x,y) \sim \Dist^K \times \Un_P^{\alpha_{q+i}(K)}}[(P^{\alpha_{q+i}(K)}(x,y)-f(x))^2] \equiv 0 \pmod \FallU
\end{equation}

Then, ${P}$ is an ${\FallU(\Gamma)}$-perfect estimator for ${(\Dist,f)}$.

\end{proposition}
\end{samepage}

\begin{proof}

Define ${\varepsilon: \Nats^n \rightarrow \Reals}$ by

\[\varepsilon(K):=\E_{(x,y) \sim \Dist^K \times \Un_P^K}[(P^K(x,y)-f(x))^2]\]

We have

\[\sum_{k=2}^\infty \frac{\varepsilon(J,k)}{k \log k}=\int_2^\infty \frac{\varepsilon(J,\Floor{t})}{\Floor{t} \log \Floor{t}} \dif t\]

\[\sum_{k=2}^\infty \frac{\varepsilon(J,k)}{k \log k} \leq \frac{3}{2} \log 3 \int_2^\infty \frac{\varepsilon(J,\Floor{t})}{t \log t} \dif t\]

\[\sum_{k=2}^\infty \frac{\varepsilon(J,k)}{k \log k} \leq \frac{3}{2} (\log 3) (\ln 2)^2 \int_2^\infty \varepsilon(J,\Floor{t}) \dif (\log \log t)\]

Substitute ${t=2^{(\log p(J) \log s)^m}}$. Denoting ${s_0=2^{(\log p(J))^{-1}}}$

\[\sum_{k=2}^\infty \frac{\varepsilon(J,k)}{k \log k} \leq \frac{3}{2} (\log 3) (\ln 2)^2 m \int_{s=s_0}^\infty \varepsilon(J,\Floor{2^{(\log p(J) \log s)^m}}) \dif (\log \log s)\]

\[\sum_{k=2}^\infty \frac{\varepsilon(J,k)}{k \log k} \leq \frac{3}{2} (\log 3) m \int_{s_0}^\infty \frac{\varepsilon(J,\Floor{2^{(\log p(J) \log s)^m}})}{s \log s} \dif s\]

\[\sum_{k=2}^\infty \frac{\varepsilon(J,k)}{k \log k} \leq \frac{3}{2} (\log 3) m \int_{s_0}^\infty \frac{\sup_{i \in \Nats} \varepsilon(J,2^{\Floor{\log p(J) \log \Floor{s}}^m}+i)}{s \log s} \dif s\]

For some ${M \in \Reals}$

\[\sum_{k=2}^\infty \frac{\varepsilon(J,k)}{k \log k} \leq M + \frac{3}{2} (\log 3) m \int_{2}^\infty \frac{\sup_{i \in \Nats} \varepsilon(J,2^{\Floor{\log p(J) \log \Floor{s}}^m}+i)}{\Floor{s} \log \Floor{s}} \dif s\]

\[\sum_{k=2}^\infty \frac{\varepsilon(J,k)}{k \log k} \leq M + \frac{3}{2} (\log 3) m \sum_{k=2}^\infty \frac{\sup_{i \in \Nats} \varepsilon(J,2^{\Floor{\log p(J) \log k}^m}+i)}{k \log k}\]

Using \ref{eqn:prp__fall_uni_qpoly} we get that for some ${M_1 \in \Reals^{>0}}$ and ${p_1 \in \NatPolyJ}$

\[\sum_{k=2}^\infty \frac{\varepsilon(J,k)}{k \log k} \leq M + M_1 \log \log p_1(J)\]

Denoting ${M_2 := 2^{M_1^{-1} M}}$

\[\sum_{k=2}^\infty \frac{\varepsilon(J,k)}{k \log k} \leq M_1 \log \log p_1(J)^{M_2}\]
\end{proof}

\subsubsection{Negative Results}

The following propositions lead to disproving the existence of optimal polynomial-time estimators with no advice for certain distributional estimation problems.

\begin{samepage}
\begin{proposition}
\label{prp:tally_perfect}

Consider ${h: \Nats^n \rightarrow \Reals}$ bounded and ${\Dist}$ a word ensemble s.t. given ${K_1, K_2 \in \Nats^n}$, if ${K_1 \ne K_2}$ then ${\Supp \Dist^{K_1} \cap \Supp \Dist^{K_2} = \varnothing}$. Assume that either ${1 \in \GrowA}$ and the image of ${h}$ is a finite subset of ${\Rats}$ or ${\Fall^{\frac{1}{2}}}$ is ${\GrowA}$-ample. Define ${f: \Supp \Dist \rightarrow \Reals}$ by requiring that for any ${K \in \Nats^n}$ and ${x \in \Supp \Dist^K}$, ${f(x)=h(K)}$. Then, there exists an ${\EG}$-perfect estimator for ${(\Dist,f)}$.

\end{proposition}
\end{samepage}

\begin{proof}

The idea is that, because the estimation problem only depends on the index $K$, the advice allows the estimator to either memorize $f$ directly or closely approximate it.

In the case ${\Fall^{\frac{1}{2}}}$ is ${\GrowA}$-ample, let $\zeta: \NatFun (0,\frac{1}{2}]$ be s.t.  $\zeta \in \Fall^{\frac{1}{2}}$ and $\Floor{\log \frac{1}{\zeta}} \in \GrowA$. In the other case, let ${\zeta \equiv 0}$. For any $K \in \Nats^n$, let ${\rho(K) \in \Argmin{s \in \Rats \cap [h(K)-\zeta(K),h(K)+\zeta(K)]} \Abs{\En_\Rats(s)}}$. It is easy to see that there is $\gamma \in \GrowA$ s.t. for any ${K \in \Nats^n}$, ${\Abs{\En_\Rats(\rho(K))} \leq \gamma(K)}$. Construct ${P: \Words \Scheme \Rats}$ s.t. for any ${K \in \Nats^n}$, ${x \in \Words}$ and ${t \in \Rats}$ s.t. ${\Abs{\En_\Rats(t)} \leq \gamma(K)}$ 

\begin{align*}
\A_P(K) &= \En_\Rats(\rho(K)) \\
\R_P(K) &= 0 \\
P^K(x,\Estr,\En_\Rats(t)) &= t
\end{align*}

We have

\[\E_{x \sim \Dist^K}[(P^K(x)-f(x))^2] = (\rho(K)-h(K))^2\]

\[\E_{x \sim \Dist^K}[(P^K(x)-f(x))^2] \leq \zeta(K)^2\]
\end{proof}

In the setting of Proposition~\ref{prp:tally_perfect}, any ${\EG}$-optimal estimator for ${(\Dist,f)}$ has to be an ${\EG}$-perfect estimator. In particular, if no \emph{uniform} ${\EG}$-perfect estimator exists then no uniform ${\EG}$-optimal estimator exists (and likewise for any other condition on the estimator). 

Denote ${\Gamma_0:=(\GrowR,\Gamma_0^n)}$, ${\Gamma_1:=(\GrowR,\Gamma_1^n)}$. Taking ${\Gamma=\Gamma_1}$ in Proposition~\ref{prp:tally_perfect} and using Proposition~\ref{prp:adv_amp}, we conclude that if the image of ${h}$ is a finite subset of ${\Rats}$ and there is no ${\Fall(\Gamma_0)}$-perfect estimator for ${(\Dist,f)}$ then there is no ${\Fall(\Gamma_0)}$-optimal estimator for ${(\Dist,f)}$. 

For distributional decision problems and ${\EG}$-samplable word ensembles we have the following stronger proposition: given an optimal estimator, we get not just a perfect estimator, but a \enquote{heuristic} algorithm that depends only on $K$ and doesn't need a problem instance.

\begin{samepage}
\begin{proposition}
\label{prp:tally_smp_bpp}

Let ${\Delta=(\Delta_{\mathfrak{R}}, \Delta_{\mathfrak{A}})}$ be a pair of growth spaces of rank ${n}$ s.t. ${\Delta_{\mathfrak{R}} \subseteq \GrowR}$, ${\Delta_{\mathfrak{A}} \subseteq \GrowA}$ and ${1 \in \Delta_{\mathfrak{A}}}$. Consider ${L \subseteq \Nats^n}$ and ${\Dist}$ a word ensemble s.t. given ${K_1, K_2 \in \Nats^n}$, if ${K_1 \ne K_2}$ then ${\Supp \Dist^{K_1} \cap \Supp \Dist^{K_2} = \varnothing}$. Define ${\chi: \Supp \Dist \rightarrow \Bool}$ by requiring that for any ${K \in \Nats^n}$ and\\ $x \in \Supp \Dist^K$, ${\chi(x)=\chi_L(K)}$. Assume ${\sigma}$ is an ${\EG}$-sampler of ${\Dist}$ and ${P}$ is an ${\Fall(\Delta)}$-optimal estimator for ${(\Dist, \chi)}$. Then there is ${A: \bm{1} \Scheme \Bool}$ s.t. ${\A_A(K)=\Chev{\A_\sigma(K),\A_P(K)}}$ and

\begin{equation}
\Prb_{y \sim \Un_A^K}[A^K(y)=\chi_L(K)] \equiv 1 \pmod \Fall
\end{equation}

\end{proposition}
\end{samepage}

\begin{proof}

Construct ${A}$ s.t. for any ${K \in \Nats^n}$, ${y_1 \in \BoolR{L}}$, ${y_1 \in \BoolR{P}}$

\begin{align*}
\R_A(K) &= \R_\sigma(K)+\R_P(K) \\
A^K(y_1 y_2) &= \begin{cases}0 \text{ if } P^K(\sigma^K(y_1),y_2) \leq \frac{1}{2} \\ 1 \text{ if } P^K(\sigma^K(y_1),y_2) > \frac{1}{2} \end{cases}
\end{align*}

We get

\[\Prb_{y \sim \Un_A^K}[A^K(y) \neq \chi_L(K)] \leq \Prb_{y_1 \sim \Un_\sigma^K,y_2 \sim \Un_P^K}\left[\Abs{P^K(\sigma^K(y_1),y_2) - \chi_L(K)} \geq \frac{1}{2}\right]\]

\[\Prb_{y \sim \Un_A^K}[A^K(y) \neq \chi_L(K)] \leq \Prb_{y_1 \sim \Un_\sigma^K,y_2 \sim \Un_P^K}\left[(P^K(\sigma^K(y_1),y_2) - \chi_L(K))^2 \geq \frac{1}{4}\right]\]

\[\Prb_{y \sim \Un_A^K}[A^K(y) \neq \chi_L(K)] \leq 4 \E_{y_1 \sim \Un_\sigma^K,y_2 \sim \Un_P^K}[(P^K(\sigma^K(y_1),y_2) - \chi_L(K))^2]\]

By Proposition~\ref{prp:smp}

\[\Prb_{y \sim \Un_A^K}[A^K(y) \neq \chi_L(K)] \leq 4 \E_{x \sim \Dist^K,y_2 \sim \Un_P^K}[(P^K(x,y_2) - \chi_L(K))^2] \pmod \Fall\]

\[\Prb_{y \sim \Un_A^K}[A^K(y) \neq \chi_L(K)] \leq 4 \E_{x \sim \Dist^K,y_2 \sim \Un_P^K}[(P^K(x,y_2) - \chi(x))^2] \pmod \Fall\]

By Proposition~\ref{prp:tally_perfect}, ${P}$ is an ${\Fall(\Delta)}$-perfect estimator for ${(\Dist,\chi)}$, therefore

\[\Prb_{y \sim \Un_A^K}[A^K(y) \neq \chi_L(K)] \equiv 0 \pmod \Fall\]
\end{proof}

Again, the statement can be reversed to disprove the existence of ${\Fall(\Delta)}$-optimal estimators for ${\Delta_{\mathfrak{A}}}=\Gamma_0^n$.

Now we consider the special case ${\Fall=\FallUt{\varphi}}$, ${\GrowR = \GammaPoly^n}$. Consider the standard decomposition of the index into two parameters $J$ (which is going to be the only relevant variable in the estimation problem) and $k$ which controls the computation time available. The following proposition states that if there is an $\FallUt{\varphi}(\Delta)$-optimal estimator for $(\Dist,\chi)$, and an $\FallUt{\varphi}(\Gamma)$ sampler for $\Dist$, then quasi-polynomial computing resources suffice to get a bounded-error randomized algorithm for computing $\chi$.

\begin{samepage}
\begin{proposition}
\label{prp:tally_fall_uni}

Consider ${\varphi: \Nats^{n-1} \rightarrow \Nats}$ superquasi-polynomial i.e. for any ${m \in \Nats}$ and\\ ${p \in \NatPolyJ}$ there is at most a finite number of ${J \in \Nats^{n-1}}$ s.t. ${\varphi(J) \leq 2^{\Ceil{\log p(J)}^m}}$. Suppose ${\GrowR  = \GammaPoly^n}$. Let ${\Delta=(\Delta_{\mathfrak{R}}, \Delta_{\mathfrak{A}})}$ be a pair of growth spaces of rank ${n}$ s.t. ${\Delta_{\mathfrak{A}} \subseteq \GrowA}$ and ${1 \in \Delta_{\mathfrak{A}}}$. Consider ${L \subseteq \Nats^{n-1}}$ and ${\Dist}$ a word ensemble s.t. given ${K_1, K_2 \in \Nats^n}$, if ${K_1 \ne K_2}$ then ${\Supp \Dist^{K_1} \cap \Supp \Dist^{K_2} = \varnothing}$. Define ${\chi: \Supp \Dist \rightarrow \Bool}$ by requiring that for any ${J \in \Nats^{n-1}}$, ${k \in \Nats}$ and ${x \in \Supp \Dist^{Jk}}$, ${\chi(x)=\chi_L(J)}$. Assume ${\sigma}$ is an ${\FallUt{\varphi}(\Gamma)}$-sampler of ${\Dist}$ and ${P}$ is an ${\FallUt{\varphi}(\Delta)}$-optimal estimator for ${(\Dist, \chi)}$ s.t. ${\A_\sigma(J,k)}$ and ${\A_P(J,k)}$ don't depend on ${k}$. Then, there are\\ $m \in \Nats$, ${p \in \NatPolyJ}$ and ${B: \bm{1} \Scheme \Bool}$ s.t. ${p \geq 1}$, ${\A_B(K)=\Chev{\A_\sigma(K),\A_P(K)}}$ and, defining ${q: \Nats^{n-1} \rightarrow \Nats}$ by ${q(J):=2^{\Ceil{\log p(J)}^m}}$

\begin{equation}
\label{eqn:prp__tally_fall_uni}
\forall J \in \Nats^{n-1}:\Prb_{y \sim \Un_B^{J,q(J)}}[B^{J,q(J)}(y)=\chi_L(J)] \geq \frac{2}{3}
\end{equation}

\end{proposition}
\end{samepage}

\begin{proof}

Obviously it is enough to construct ${m}$, ${p}$ and ${B}$ s.t. \ref{eqn:prp__tally_fall_uni} holds for all but a finite number of ${J \in \Nats^{n-1}}$. Use Proposition~\ref{prp:tally_smp_bpp} to construct ${A: \bm{1} \Scheme \Bool}$. Given any ${k \in \Nats}$, define ${\omega^k \in \mathcal{P}(\Nats)}$ s.t. for some ${N \in \Reals^{>0}}$

\[\omega^k(i):=\begin{cases}\frac{N}{i \log i} \text{ if } 2 \leq i < k\\0 \text{ if } i < 2 \text { or } i \geq k\end{cases}\]

Denote ${\Gamma^1:=(\GammaPoly^1,\Gamma_0^1)}$. Adapting the standard argument that any computable distribution is samplable, we can construct ${\tau: \bm{1} \xrightarrow{\Gamma^1} \Nats}$ s.t. ${\Supp \tau_\bullet^{k} \subseteq [k]}$ and ${\Dtv(\tau_\bullet^{k}, \omega^k) \leq \frac{1}{6}}$. Construct ${B: \bm{1} \Scheme \Bool}$ s.t. for any ${J \in \Nats^{n-1}}$, ${k \in \Nats}$, ${y \in \Bool^{\R_\tau(J,k)}}$ and ${z \in \Words}$

\begin{align*}
\R_B(J,k) &\geq \R_\tau(k) + \max_{i \in [k]} \R_A(J,i) \\
B^{Jk}(y,z) &= A^{J,\tau^{k}(y)}(z_{<\R_A(J,\tau^{k}(y))}) 
\end{align*}

That is, $B$ functions by generating a distribution over numbers up to $k$ that is approximately $\frac{1}{i\log i}$, and then sampling from it to determine how much computing resources to allocate to $A$, which is a perfect estimator.

We know that for some ${M \in \Reals^{\geq 0}}$ and ${p \in \NatPolyJ}$ s.t. ${p \geq 1}$

\[\sum_{k=2}^{\varphi(J)-1} \frac{\Prb_{z \sim \Un_A^{Jk}}[A^{Jk}(z) \neq \chi_L(J)]}{k \log k} \leq M \log \log p(J)\]

Take ${m = \Ceil{\frac{6M}{(\ln 2)^2}}}$. We get

\[\E_{k \sim \omega^{q(J)}}[\Prb_{z \sim \Un_A^{Jk}}[A^{Jk}(z) \neq \chi_L(J)]]=\frac{\sum_{k=2}^{q(J)-1} \frac{\Prb_{z \sim \Un_A^{Jk}}[A^{Jk}(z) \neq \chi_L(J)]}{k \log k}}{\sum_{k=2}^{q(J)-1} \frac{1}{k \log k}}\]

Denote $I:=\{J \in \Nats^{n-1} \mid \varphi(J) < q(J)\}$. We get

\[\forall J \in \Nats^{n-1} \setminus I: \E_{k \sim \omega^{q(J)}}[\Prb_{z \sim \Un_A^{Jk}}[A^{Jk}(z) \neq \chi_L(J)]] \leq \frac{M \log \log p(J)}{\int_2^{q(J)} \frac{\dif t}{t \log t}}\]

\[\forall J \in \Nats^{n-1} \setminus I:\E_{k \sim \omega^{q(J)}}[\Prb_{z \sim \Un_A^{Jk}}[A^{Jk}(z) \neq \chi_L(J)]] \leq \frac{M \log \log p(J)}{(\ln 2)^2 \log \log q(J)}\]

\[\forall J \in \Nats^{n-1} \setminus I: \E_{k \sim \omega^{q(J)}}[\Prb_{z \sim \Un_A^{Jk}}[A^{Jk}(z) \neq \chi_L(J)]] \leq \frac{M \log \log p(J)}{(\ln 2)^2 m \log \Ceil{\log p(J)}}\]

\[\forall J \in \Nats^{n-1} \setminus I: \E_{k \sim \omega^{q(J)}}[\Prb_{z \sim \Un_A^{Jk}}[A^{Jk}(z) \neq \chi_L(J)]] \leq \frac{1}{6}\]

\[\forall J \in \Nats^{n-1} \setminus I: \E_{y \sim \Un_\tau^{q(J)}}[\Prb_{z \sim \Un_A^{J,\tau^{q(J)}(y)}}[A^{J,\tau^{q(J)}(y)}(z) \neq \chi_L(J)]] \leq \frac{1}{6}+\Dtv(\tau_\bullet^{q(J)},\omega^{q(J)})\]

\[\forall J \in \Nats^{n-1} \setminus I: \Prb_{y \sim \Un_B^{J,q(J)}}[B^{J,q(J)}(y) \ne \chi_L(J)] \leq \frac{1}{3}\]

By the assumption on ${\varphi}$, ${I}$ is a finite set therefore we got the desired result.
\end{proof}

For ${n=2}$, we can think of ${L}$ as a language using \emph{unary} encoding of natural numbers. Proposition~\ref{prp:tally_fall_uni} and Proposition~\ref{prp:adv_amp} imply that if ${\Delta_{\mathfrak{A}}=\Gamma_0^n}$, ${\sigma}$ is uniform, and this language cannot be decided in quasi-polynomial time by a bounded-error randomized algorithm, then there is no ${\FallUt{\varphi}(\Delta)}$-optimal estimator for ${(\Dist,\chi)}$.

Thanks to the results of section~\ref{sec:reductions} and Theorem~\ref{thm:con_ort}, these negative results imply non-existence results for ${\Fall^\sharp(\Delta)}$-optimal estimators\footnote{The need to use ${\Fall^\sharp(\Delta)}$-optimal estimators rather than ${\Fall(\Delta)}$-optimal estimators arises because the theorems about reductions as we formulated them don't apply to ${\Fall(\Delta)}$-optimal estimators with ${\Delta=\Gamma_0^n}$ or ${\Delta=\Gamma_1^n}$. This can be overcome by using somewhat more special reductions which still admit a similar completeness theorem, but we omit details in the present work.} for any distributional estimation problem s.t. a problem admitting a negative result has an appropriate reduction to it.

\subsection{Uniqueness}

Since we view optimal polynomial-time estimators as computing \enquote{expected values}, it is natural to expect that their values only depend on the distributional estimation problem rather than the particular optimal polynomial-time estimator. However, since they are defined via an asymptotic property exact uniqueness is impossible. Instead, different $\Fall^\sharp(\Gamma)$-optimal estimators have the expectation of the squared difference between their estimates fall fast enough to be in $\Fall$ (which is an equivalence relation on the set of arbitrary estimators).

\begin{samepage}
\begin{theorem}
\label{thm:uniq}

Consider ${(\Dist,f)}$ a distributional estimation problem. Assume there is\\ ${p \in \NatPoly}$ s.t. 

\begin{equation}
\label{eqn:thm__uniq__dist}
\Dist^K(\Bool^{\leq p(K)}) \equiv 1 \pmod \Fall
\end{equation}

Suppose ${P}$ and ${Q}$ are ${\ESG}$-optimal estimators for ${(\Dist,f)}$. Then

\begin{equation}
\label{eqn:thm__uniq}
\E_{(x,y,z) \sim \Dist^K \times \Un_P^K \times \Un_Q^K}[(P^K(x,y)-Q^K(x,z))^2] \equiv 0 \pmod \Fall
\end{equation}

\end{theorem}
\end{samepage}

\begin{proof}

Construct ${S: \Words \times \Rats \Scheme \Rats}$ bounded s.t. for any ${K \in \Nats^n}$, ${x \in \Bool^{\leq p(K)}}$, ${t \in \Img P^K}$ and ${z \in \BoolR{Q}}$

\begin{align*}
\R_{S}(K) &= \R_Q(K) \\
S^K(x,t,z) &= t-Q^K(x,z) 
\end{align*}

Construct ${T: \Words \times \Rats \Scheme \Rats}$ bounded s.t. for any ${K \in \Nats^n}$, ${x \in \Bool^{\leq p(K)}}$, ${s \in \Img Q^K}$ and\\ $y \in \BoolR{P}$

\begin{align*}
\R_{T}(K) &= \R_P(K) \\
T^K(x,s,y) &= P^K(x,y)-s 
\end{align*}

${P}$ is an ${\ESG}$-optimal estimator for ${(\Dist,f)}$, therefore

\[\E_{(x,y,z) \sim \Dist^K \times \Un_P^K \times \Un_S^K}[(P^K(x,y)-f(x))S^K(x,P^K(x,y),z)] \equiv 0 \pmod \Fall\]

The construction of ${S}$ and \ref{eqn:thm__uniq__dist} give

\begin{equation}
\label{eqn:thm__uniq__prf1}
\E_{(x,y,z) \sim \Dist^K \times \Un_P^K \times \Un_Q^K}[(P^K(x,y)-f(x))(P^K(x,y)-Q^K(x,z))] \equiv 0 \pmod \Fall
\end{equation}

${Q}$ is an ${\ESG}$-optimal estimator for ${(\Dist,f)}$, therefore

\[\E_{(x,z,y) \sim \Dist^K \times \Un_Q^K \times \Un_T^K}[(Q^K(x,z)-f(x))T^K(x,Q^K(x,z),y)] \equiv 0 \pmod \Fall\]

The construction of ${T}$ and \ref{eqn:thm__uniq__dist} give

\begin{equation}
\label{eqn:thm__uniq__prf2}
\E_{(x,z,y) \sim \Dist^K \times \Un_Q^K \times \Un_P^K}[(Q^K(x,z)-f(x))(P^K(x,y)-Q^K(x,z))] \equiv 0 \pmod \Fall
\end{equation}

Subtracting \ref{eqn:thm__uniq__prf2} from \ref{eqn:thm__uniq__prf1}, we get \ref{eqn:thm__uniq}.
\end{proof}

The notion of \enquote{conditional expected value} introduced in subsection~\ref{subsec:alg} allows conditions which are occasionally \emph{false}. In some sense this provides us with well-defined (probabilistic) answers to \enquote{what if} questions that are meaningless in formal logic due to the principle of explosion, a concept which was hypothesized to be useful for solving paradoxes in decision theory\cite{Soares_2015}. However, Theorem~\ref{thm:uniq} suggests that the values of an optimal polynomial-time estimator are only meaningful inside ${\Supp \Dist^K}$ whereas \enquote{conditional expected values} require using the word ensemble ${\Dist \mid L}$ (see Theorem~\ref{thm:cond}) so violation of the condition (i.e. ${x \not\in L}$) means falling outside the support of the word ensemble. On the other hand, we will now show that when the condition is unpredictable with the given amount of computational resources, a stronger uniqueness theorem holds that ensures \enquote{counterfactual} values are also stable, although the fall space measuring the difference of the optimal estimators is scaled up by a factor decreasing with the \enquote{degree of unpredictability}.

\begin{samepage}
\begin{theorem}
\label{thm:uniq_cond}

Consider ${(\Dist,f)}$ a distributional estimation problem and ${L \subseteq \Words}$ s.t. for all ${K \in \Nats^n}$, $\Dist^K(L) > 0$. Define $\gamma_L: \NatFun \Reals$ by $\gamma_L(K):=\Dist^{K}(L)^{-1}$ and $\Fall_L:=\gamma_L \Fall$. Assume there is\\ ${p \in \NatPoly}$ s.t. \ref{eqn:thm__uniq__dist} holds. Let ${R}$ be an ${\ESG}$-optimal estimator for ${(\Dist, \chi_L)}$. Assume ${\epsilon: \NatFun \Reals^{>0}}$ is s.t. for all ${x,y \in \Words}$, ${R^K(x,y) \geq \epsilon(K) \Dist^K(L)}$. Suppose ${P}$ and ${Q}$ are ${\Fall_L^\sharp(\Gamma)}$-optimal estimators for ${(\Dist \mid L,f)}$. Then

\begin{equation}
\E_{(x,y,z) \sim \Dist^K \times \Un_P^K \times \Un_Q^K}[(P^K(x,y)-Q^K(x,z))^2] \equiv 0 \pmod {\epsilon^{-1} \Fall_L}
\end{equation}

\end{theorem}
\end{samepage}

\begin{proof}

${R}$ is an ${\ESG}$-optimal estimator for ${(\Dist, \chi_L)}$, therefore

\[\E_{(x,y,z,w) \sim \Dist^K \times \Un_P^K \times \Un_Q^K \times \Un_R^K}[(R^K(x,w)-\chi_L(x))(P^K(x,y)-Q^K(x,z))^2] = 0 \pmod \Fall\]

\[\E_{\Dist^K \times \Un_P^K \times \Un_Q^K \times \Un_R^K}[R^K \cdot (P^K-Q^K)^2] = \E_{\Dist^K \times \Un_P^K \times \Un_Q^K \times \Un_R^K}[\chi_L \cdot (P^K-Q^K)^2] \pmod \Fall\]

\[\E_{\Dist^K \times \Un_P^K \times \Un_Q^K \times \Un_R^K}[R^K \cdot (P^K-Q^K)^2] = \Dist^K(L) \E_{\Dist^K \mid L \times \Un_P^K \times \Un_Q^K \times \Un_R^K}[(P^K-Q^K)^2] \pmod \Fall\]

\[\E_{\Dist^K \times \Un_P^K \times \Un_Q^K \times \Un_R^K}[\epsilon(K)\Dist^K(L)(P^K-Q^K)^2] \leq \Dist^K(L) \E_{\Dist^K \mid L \times \Un_P^K \times \Un_Q^K \times \Un_R^K}[(P^K-Q^K)^2] \pmod \Fall\]

\[\epsilon(K)\E_{\Dist^K \times \Un_P^K \times \Un_Q^K \times \Un_R^K}[(P^K-Q^K)^2] \leq \E_{\Dist^K \mid L \times \Un_P^K \times \Un_Q^K \times \Un_R^K}[(P^K-Q^K)^2] \pmod {\Fall_L}\]

Applying Theorem~\ref{thm:uniq} to the right hand side, we conclude

\[\epsilon(K)\E_{\Dist^K \times \Un_P^K \times \Un_Q^K \times \Un_R^K}[(P^K-Q^K)^2] \equiv 0 \pmod {\Fall_L}\]

\[\E_{\Dist^K \times \Un_P^K \times \Un_Q^K \times \Un_R^K}[(P^K-Q^K)^2] \equiv 0 \pmod {\epsilon^{-1}\Fall_L}\]
\end{proof}

Theorem~\ref{thm:uniq_cond} implies that in simple scenarios, \enquote{counterfactual} optimal estimates behave as intuitively expected, assuming ${L}$ is \enquote{sufficiently unpredictable}. For example, if there is an efficient algorithm that evaluates ${f}$ correctly given the promise ${x \in L}$ then a conditional optimal polynomial-time estimator constructed using Theorem~\ref{thm:cond} will produce approximately the same values as this algorithm whether ${x}$ is in ${L}$ or not.

\section{Discussion}
\label{sec:discussion}

The motivation for optimal polynomial-time estimators comes from the desire to quantify the uncertainty originating in computational resource bounds. We used this motivation to arrive at an intuitive definition, and proceeded to show the resulting object has many properties of \enquote{normal} probability theory, justifying its interpretation as a brand of expected value. Moreover, there are associated concepts of reductions and complete problems analogous to standard constructions in average-case complexity theory. 

Thus, the class of distributional estimation problems admitting ${\EG}$-optimal estimators (or ${\ESG}$-optimal estimators) is a natural distributional complexity class. In light of the positive and negative existence results we have demonstrated, these new classes are unlikely to trivially coincide with any of the previously known classes. Mapping the boundary of these classes and understanding their relationships with other classes in average-case complexity theory seems to be ground for much further work. Moreover, it is possible to consider generalizations by including more types of computational resources e.g. space, parallelism and/or non-determinism.

As an example of a natural open problem, consider ${(\Dist_{\textsc{NP}}, f_{\textsc{NP}})}$, the complete problem for ${\textsc{SampNP}}$ resulting from Theorem~\ref{thm:complete} with ${n=1}$, ${r(k)=s(k)=k}$, ${E=E_{\textsc{NP}}}$ and ${\mathfrak{F}=\mathfrak{F}_{\textsc{NP}}}$. Theorem~\ref{thm:exists_all} implies that e.g. there is an ${\mathcal{F}_{\text{uni}}^{(2)}(\GammaPoly^2,\GammaLog^2)}$-optimal estimator for ${(\Dist_{\textsc{NP}}^\eta, f_{\textsc{NP}})}$. On the other hand, Proposition~\ref{prp:tally_fall_uni} implies that it is unlikely that there is an ${\mathcal{F}_{\text{uni}}^{(2)}(\GammaPoly^2,\Gamma_0^2)}$-optimal estimator\footnote{More precisely, it cannot exist assuming there is a unary language in ${\textsc{NP}}$ that cannot be decided by a randomized algorithm in quasi-polynomial time with bounded probability of error.}. This, however, doesn't tell us anything about the existence of an ${\mathcal{F}_{\text{uni}}^{(2)}(\GammaPoly^2,\Gamma_1^2)}$-optimal estimator. This question fits naturally into the theme of Impagliazzo's \enquote{worlds}\cite{Impagliazzo_1995}: if there is an ${\mathcal{F}_{\text{uni}}^{(2)}(\GammaPoly^2,\Gamma_0^2)}$-\emph{perfect} estimator for ${(\Dist_{\textsc{NP}}^\eta, f_{\textsc{NP}})}$ (a version of Impagliazzo's \enquote{Heuristica} which is considered unlikely), then the answer is tautologically positive. However, if there is no such perfect polynomial-time estimator then the optimal polynomial-time estimator may or may not exist, a possible new partition of \enquote{worlds}\footnote{The relation to the worlds is somewhat disturbed by the role of $O(1)$ advice. We think there is a natural variant of this question that doesn't involve advice but it is out of the present scope.}.

One area where applying these concepts seems natural is Artificial General Intelligence. Indeed, the von Neumann–Morgenstern theorem shows that perfect rational agents are expected utility maximizers but in general the exact evaluation of expected utility is intractable. It is thus natural to substitute an optimal polynomial-time estimator for utility, as the analogue of expected value in the computationally bounded case. Further illuminating the connection, Theorem~\ref{thm:exists_smp} shows how optimal polynomial-time estimators result from agnostic PAC learning.

Some results we left out of the present work show the existence of systems of optimal polynomial-time estimators that are \enquote{reflective} i.e. estimate systems of functions which depend on the estimators themselves. We constructed such systems using the Kakutani-Glicksberg-Fan theorem which requires the use of random advice strings, as in the definition of ${\EMG}$-samplers. Such systems can be used to model game theoretic behavior of computationally bounded rational agents, similarly to the use of reflective oracles\cite{Fallenstein_2015} for unbounded agents. 

Finally, we wish to express the hope that the present work will lead to incorporating more concepts from complexity theory into the theory of AGI, serving to create a stronger theoretical foundation for AI in general. The importance of building such a theoretical foundation is enormous since it is necessary to predict and control the outcome of the eventual creation of artificial agents with superhuman intelligence, an event which might otherwise trigger a catastrophe\cite{Bostrom_2014}.


\appendix

\section{Appendix}

We review the definitions of hard-core predicate and one-way function and state the Goldreich-Levin theorem.

We will use the notation $\Gamma_{\text{det}}:=(\Gamma_0^1,\Gamma_0^1)$, $\Gamma_{\text{rand}}:=(\GammaPoly^1,\Gamma_0^1)$, ${\Gamma_{\text{circ}}:=(\Gamma_0^1,\GammaPoly^1)}$.

\begin{samepage}
\begin{definition}
\label{def:hard_core}

Given $\Dist$ a word ensemble\footnote{The standard definition of a hard-core predicate corresponds to the case $\Dist^k=\Un^k$. Here we allow for slightly greater generality.}, $f: \Supp \Dist \rightarrow \Words$ and ${B: \Words \xrightarrow{\Gamma_{\text{det}}} \Bool}$, $B$ is a called a \emph{hard-core predicate} of $(\Dist,f)$ when for any $S: \Words \xrightarrow{\Gamma_{\textnormal{rand}}} \Bool$

\begin{equation}
\Prb_{(x,y) \sim \Dist^k \times \Un_S^k}[S^k(f(x),y)=B^k(x)] \leq \frac{1}{2} \pmod {\Fall_{\text{neg}}}
\end{equation}

\end{definition}
\end{samepage}

\begin{samepage}
\begin{definition}
\label{def:hard_core_adv}

Given $\Dist$ a word ensemble, $f: \Supp \Dist \rightarrow \Words$ and ${B: \Words \xrightarrow{\Gamma_{\text{det}}} \Bool}$, $B$ is a called a \emph{non-uniformly hard-core predicate} of $(\Dist,f)$ when for any ${S: \Words \xrightarrow{\Gamma_{\textnormal{circ}}} \Bool}$ 

\begin{equation}
\Prb_{x \sim \Dist^k}[S^k(f(x))=B^k(x)] \leq \frac{1}{2} \pmod {\Fall_{\text{neg}}}
\end{equation}

\end{definition}
\end{samepage}

\begin{samepage}
\begin{definition}

$f: \Words \Alg \Words$ is called an \emph{one-way function}
when

\begin{enumerate}[(i)]

\item There is $p: \Nats \rightarrow \Nats$ polynomial s.t. $\forall x \in \Words: \T_f(x) \leq p(\Abs{x})$.

\item For any $S: \Words \xrightarrow{\Gamma_{\text{rand}}} \Words$

\begin{equation}
\Prb_{(x,y) \sim \Un^k \times \Un_S^k}[f(S^k(f(x),y))=x] \equiv 0 \pmod {\Fall_{\textnormal{neg}}}
\end{equation}

\end{enumerate}

\end{definition}
\end{samepage}

\begin{samepage}
\begin{definition}

$f: \Words \Alg \Words$ is called a \emph{non-uniformly hard to invert} one-way function
when

\begin{enumerate}[(i)]

\item There is $p: \Nats \rightarrow \Nats$ polynomial s.t. $\forall x \in \Words: \T_f(x) \leq p(\Abs{x})$.

\item For any $S: \Words \xrightarrow{\Gamma_{\text{circ}}} \Words$

\begin{equation}
\Prb_{x \sim \Un^k}[f(S^k(f(x)))=x] \equiv 0 \pmod {\Fall_{\textnormal{neg}}}
\end{equation}

\end{enumerate}

\end{definition}
\end{samepage}

It is easy to see that any non-uniformly hard-core predicate is in particular a hard-core predicate and any non-uniformly hard to invert one-way function is in particular a one-way function.

The following appears in \cite{Goldreich_2008} as Theorem 7.7. Here we state it in the notation of the present work.

\begin{theorem}[Goldreich-Levin]
\label{thm:goldreich_levin}

Consider a one-way function ${f: \Words \Alg \Words}$. Let\\ $\Dist^k:=\Un^{2k}$, $f_{\textnormal{GL}}: \Supp \Dist \rightarrow \Words$ and ${B: \Words \xrightarrow{\Gamma_{\textnormal{det}}} \Bool}$ be s.t. for any $x,y \in \WordsLen{k}$, $f_{\textnormal{GL}}(xy)=\Chev{f(x),y}$ and ${B^k(xy)=x \cdot y}$. Then, $B$ is a hard-core predicate of $(\Dist, f_{\textnormal{GL}})$.

\end{theorem}

There is also a non-uniform version of the theorem which is not stated in \cite{Goldreich_2008}, but its proof is a straightforward adaptation.

\begin{theorem}
\label{thm:goldreich_levin_circ}

In the setting of Theorem~\ref{thm:goldreich_levin}, assume $f$ is non-uniformly hard to invert. Then $B$ is a non-uniformly hard-core predicate of $(\Dist, f_{\textnormal{GL}})$.
\end{theorem}

\section*{Funding}

This work was partially supported by the Machine Intelligence Research Institute in Berkeley, California.

\section*{Acknowledgments}

We thank Patrick LaVictoire for many useful discussions and also suggesting some corrections in the paper. We thank Scott Aaronson for reading a draft version of the paper and providing important advice on the form of presentation. We thank an anonymous reviewer, writing by the request of the Open Philanthropy Project, who read a draft version of the paper and provided useful suggestions. We thank Rob Bensinger for helping to polish some of the English. We thank Brian Njenga for locating typos. We thank Scott Garrabrant for useful discussions.

\bibliographystyle{plain}
\bibliography{Optimal_Estimators}

\end{document}